\UseRawInputEncoding
\documentclass[11pt, reqno]{amsart}
\usepackage{amsfonts,latexsym, enumerate}
\usepackage{amsmath}
\usepackage{amscd}
\usepackage{float,amsmath,amssymb,mathrsfs,bm,multirow,graphics}
\usepackage[dvips]{graphicx}
\usepackage[percent]{overpic}
\usepackage{amsaddr}
\usepackage{todonotes}
\usepackage[numbers,sort&compress]{natbib}
\usepackage{grffile}
\usepackage{bbm}

\numberwithin{equation}{section}
\newcommand{\R}{{\mathbb R}}

\newcommand{\C}{{\mathbb C}}

\newcommand{\Z}{{\mathbb Z}}

\DeclareMathOperator{\Tr}{Tr}

\DeclareMathOperator{\sgn}{sgn}

\DeclareMathOperator{\Ln}{Ln}
\DeclareMathOperator{\arcsinh}{arcsinh}
\DeclareMathOperator{\arctanh}{arctanh}

\newcommand{\Li}{\text{\upshape Li}}

\newcommand{\I}{\text{\upshape I}}
\newcommand{\II}{\text{\upshape II}}
\newcommand{\III}{\text{\upshape III}}
\newcommand{\app}{\text{\upshape app}}

\newcommand{\AF}{\mathrm{AF}}
\newcommand{\F}{\mathrm{F}}

\def\XXint#1#2#3{{\setbox0=\hbox{$#1{#2#3}{\int}$}
\vcenter{\hbox{$#2#3$}}\kern-.5\wd0}}

\allowdisplaybreaks

\newtheorem{theorem}{Theorem}[section]
\newtheorem{proposition}[theorem]{Proposition}

\newtheorem{lemma}[theorem]{Lemma}

\newtheorem{remark}[theorem]{Remark}
\newtheorem*{remark*}{Remark}

\newtheorem{figuretext}{Figure}

\usepackage[colorlinks=true]{hyperref}
\hypersetup{urlcolor=blue, citecolor=red, linkcolor=blue}

\input epsf
\title{On the Hartree--Fock phase diagram for the two-dimensional Hubbard model}
    
\author{C. Charlier$^{1}$, E. Langmann$^{2}$ and J. Lenells$^{3}$}
\address{$^1$Research Institute in Mathematics and Physics, UCLouvain, \\ 1348 Louvain-La-Neuve, Belgium.
	\\
$^2$Department of Physics, KTH Royal Institute of Technology, \\ 106 91 Stockholm, Sweden.
	\\
$^3$Department of Mathematics, KTH Royal Institute of Technology, \\ 100 44 Stockholm, Sweden.}
\email{christophe.charlier@uclouvain.be}
\email{langmann@kth.se}
\email{jlenells@kth.se}

\begin{document}
\begin{abstract}
We propose an analytical method for the construction of Hartree--Fock phase diagrams for the (fermion) Hubbard model and various generalizations thereof.
Such phase diagrams are traditionally constructed numerically, but we argue that, by using asymptotic techniques, it is possible to obtain analytic formulas approximating the curves separating the different phases to very high accuracy. 
To illustrate the new method, we apply it to the two-dimensional Hubbard model on the square lattice at zero temperature. This yields formulas for the Hartree--Fock phase boundaries that agree with, but also improve on, earlier numerical results. In particular, our results provide the first rigorous proof of the existence of mixed phases in this model.
\end{abstract} 
\maketitle




\section{Introduction}

The Hubbard model describes interacting fermions on a lattice. It provides insight into how insulating, ferromagnetic, antiferromagnetic, and other more exotic phases, can arise in a material as a result of electron interaction. The model was initially conceived in the 1960's to explain the unexpected insulating behavior of certain transition metal compounds which standard band theory predicted to be conductors \cite{G1963, H1963, K1963}.
A proposal in the late 1980's that the 2D Hubbard model  is a prototype model for high-temperature superconductivity \cite{A1987} has led to large amount of work on this model up to this day; see e.g.\ \cite{Simons} for recent work by a large collaboration comparing results about this model obtained by different approximation methods (here and in the following, {\em 2D Hubbard model}  is  short for the {\em two-dimensional Hubbard model on the square lattice $\Z^2$}). The 2D Hubbard model is today considered to be a fundamental model in quantum physics, which is of comparable importance to quantum statistical mechanics as the Ising model is to statistical mechanics, see e.g. \cite{ABKR2022}. 

The Hubbard model in its original form  \cite{G1963, H1963, K1963} describes electrons on a cubic lattice $\Z^3$, hopping between nearest-neighbor sites and interacting with a repulsive on-site density-density interaction; the model depends on two parameters: the hopping parameter $t>0$, and the on-site repulsion $U > 0$.  The definition of this model can be straightforwardly generalized to other bipartite lattices, and there exist many examples of such generalizations of the Hubbard model which are of interest in physics. In particular, the one-dimensional Hubbard model, where the lattice is $\Z$, is famous due to its exact solution by Lieb and Wu  \cite{LW1968} using the Bethe ansatz, and the infinite-dimensional Hubbard model obtained as a limit $n\to \infty$ of the Hubbard model on $\Z^n$ \cite{MV1989} is the basis of a popular approximation method known as dynamical mean field theory \cite{GK1992}. Moreover, in addition to the 2D Hubbard model, which is a prototype model for high-temperature superconductors as already discussed, there is also the two-dimensional Hubbard model on the honeycomb lattice,  which is a well-known model of  graphene \cite{S1984} and which is easier to control mathematically than the 2D Hubbard model \cite{GM2010,GMP2012}; see also \cite{BGM2006,S1998}. Other examples include the Hubbard model on the diamond lattice in three dimensions \cite{SAST1993}, a limit $n\to \infty$ of the Hubbard model on diamond lattice in $n$ dimensions \cite{SAST1993}, and anisotropic variants of the Hubbard model on $\Z^n$ where the hopping constant is different in different directions; see e.g.\ \cite{LL2025}. Furthermore, there are various extensions of these  Hubbard models obtained by adding other hopping and/or interaction terms; for example, three well-known such extensions are obtained by including: (i) hopping of strength $t'\in \R$ between next-nearest neighbor sites \cite{LH1989}, (ii) next-nearest neighbor density-density interactions of strength $V\in\R$ \cite{ZC1989}, (iii) next-nearest neighbor spin-spin interactions of strength $J>0$ \cite{S1988}. These extended Hubbard models were proposed and studied mainly as generalizations of the 2D Hubbard model in the context of high-temperature superconductivity. However, mathematically, the definitions of these extended models straightforwardly generalize to other bipartite lattices.  A common terminology used to distinguish these models is by the non-zero parameters and the lattice; for example, the original Hubbard model  \cite{G1963, H1963, K1963}  can be called the $t$-$U$ model on $\Z^3$, and the extension of the 2D Hubbard model obtained by adding next-nearest neighbor hopping and next-nearest neighbor density-density interactions to the 2D Hubbard model is known as the 2D $t$-$t'$-$U$-$V$ model, etc. We refer to this large zoo of models as {\em Hubbard-like models}.

A problem of central importance in the study of a condensed matter system is the construction of a phase diagram, i.e., a diagram that shows how the phase of the physically realized state changes as the parameters of the model change. 
The construction of the full phase diagram for the 2D Hubbard model is an extraordinarily rich problem, largely beyond the reach of current (analytical and numerical) methods; but see \cite{ABKR2022, QSACG2022} for a review of some partial results.
It is easier to make progress after applying some approximation scheme. One of the most common and most powerful approximations in condensed matter physics is that provided by Hartree--Fock theory. Outstanding features of this method include its versatility  (in principle, it can be used for {\em any}  model), and its large track record of successes in condensed matter theory (see e.g.\  \cite{A2018} for a classic textbook discussion of these successes).
To be precise, Hartree--Fock theory comes in two forms: unrestricted Hartree--Fock theory and restricted Hartree--Fock theory. Mathematically, the unrestricted Hartree--Fock approximation consists of restricting attention to states for which Wick's theorem is valid. This leads to major simplifications, but since the number of variational parameters in unrestricted Hartree--Fock theory still grows linearly with the number of lattice sites, it is common to draw on experience and physics intuition to restrict the class of variational states even further. The method is then referred to as restricted Hartree--Fock theory, or simply as Hartree--Fock theory --- it is this restricted method which has the track record of successes in condensed matter theory mentioned above. In the context of Hartree--Fock theory for the Hubbard model, it is common to restrict attention to paramagnetic (P), ferromagnetic (F), and antiferromagnetic (AF) states. A version of restricted Hartree--Fock theory involving these variational states was first implemented for the 3D Hubbard model in the pioneering work by Penn \cite{P1966}; similar results were later obtained for the 2D Hubbard model by Hirsch \cite{H1985}. However, due to conceptual issues peculiar to Hubbard-like models and clarified only much later in works by Bach, Lieb, and Solovej \cite{BLS1994} and Bach and Poelchau  \cite{BP1996}, the phase diagrams obtained in these pioneering works are qualitatively wrong \cite{LW1997,LW2007,LLphysrevlong} (we explain this in Section~\ref{sec:mixed}). The correct Hartree--Fock phase diagram for the Hubbard model on $\Z^n$ appeared in \cite{LW2007} for $n=2$ and in \cite{LLphysrevlong} for $n=1,3,\infty$ (for the first time, to our knowledge). These corrected Hartree--Fock phase diagram were obtained by numerical methods, based on a mathematically rigorous approach to Hartree--Fock theory \cite{BLS1994,BP1996}.  Our aim in this paper is to develop analytical tools allowing us to promote these numeric results to mathematical theorems.

For simplicity, and since it is an important special case, our focus in this paper is on the 2D Hubbard model at zero temperature; in on-going work, we show that many of our results here can be generalized to finite temperature and other Hubbard-like models \cite{HLLLM2026,LLM2025}.

\begin{figure}
	\vspace{0.4cm}
	\begin{center}
		\hspace{-.6cm}
		\begin{overpic}[width=.49\textwidth]{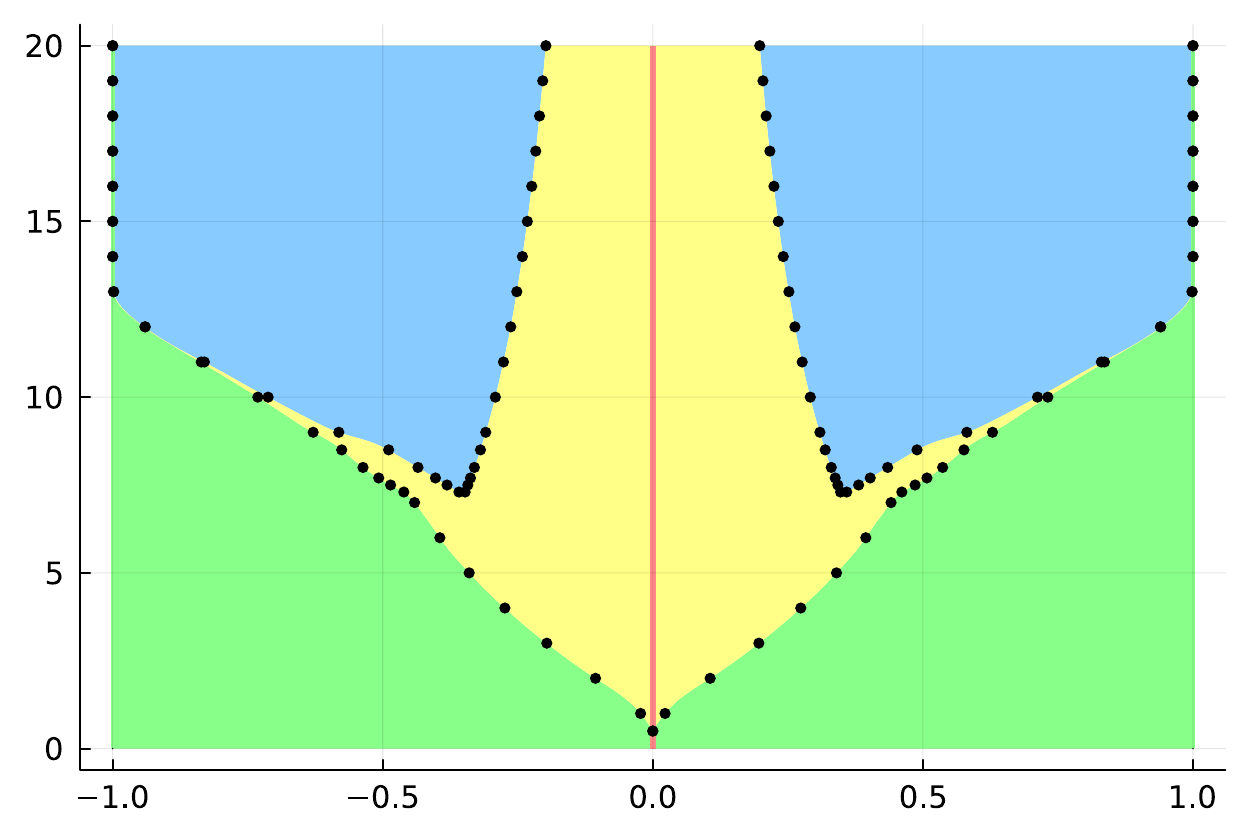 }
			\put(5.5,66.5){\footnotesize $U$}
			\put(99.5,4.5){\footnotesize $\nu$}
            \put(20,15){\footnotesize P}
            \put(80,15){\footnotesize P}
            \put(26,45){\footnotesize F}
            \put(77,45){\footnotesize F}				
			\put(38.5,23){\footnotesize Mixed}
			\put(54,23){\footnotesize Mixed}
			\put(44,39){\footnotesize AF}
			\put(48,38){\vector(2,-3){4}}
		\end{overpic}
		\hspace{0cm}
		\begin{overpic}[width=.49\textwidth]{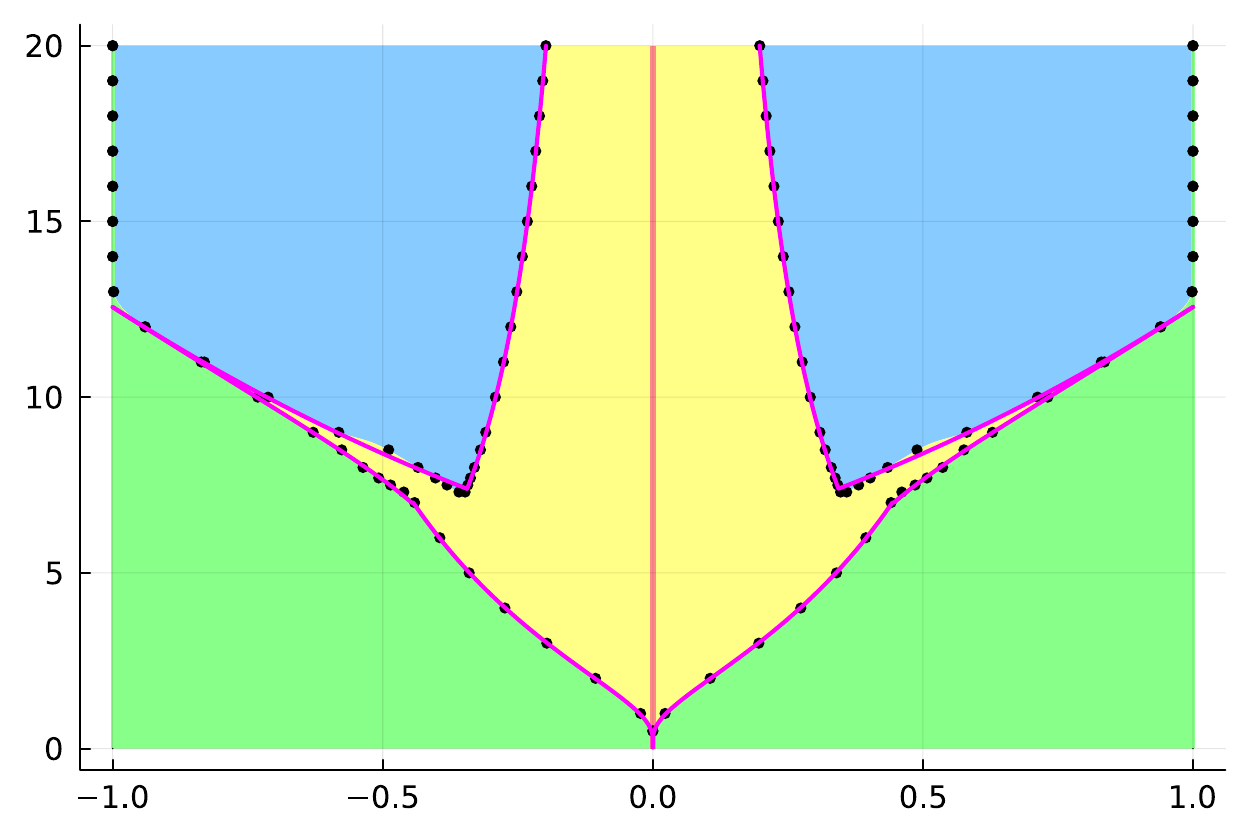 }
			\put(5.5,66.5){\footnotesize $U$}
			\put(99.5,4.5){\footnotesize $\nu$}
            \put(20,15){\footnotesize P}
            \put(80,15){\footnotesize P}
            \put(26,45){\footnotesize F}
            \put(77,45){\footnotesize F}				
            \put(63.5,50){\footnotesize $\nu_{\I}^{\F}$}		
            \put(81,38){\footnotesize $\nu_{\II}^{\F}$}		
            \put(85,32){\footnotesize $\nu_{\II}^{\mathrm{P}}$}		
            \put(62,13){\footnotesize $\nu_{\III}^{\mathrm{P}}$}		
			\put(38.5,23){\footnotesize Mixed}
			\put(54,23){\footnotesize Mixed}
			\put(44,39){\footnotesize AF}
			\put(48,38){\vector(2,-3){4}}
		\end{overpic}
		 \begin{figuretext} \label{phasediagramfig} 
		  Left: Numerically computed phase diagram of the 2D Hubbard model at zero temperature obtained by Hartree--Fock theory restricted to P, F, and AF states. The phases are shown as a function of doping $\nu$ and coupling $U$, where the doping is the average number of electrons per site minus one; the hopping parameter $t$ is set to $1$. 
			P, F, and AF regions are green, blue, and red, respectively, and yellow regions indicate mixed phases. 
			\\
		Right: Same phase diagram as on the left, but with the phase boundaries obtained analytically in this paper, $\nu_{\I}^{\F}$, $\nu_{\II}^{\F}$, $\nu_{\II}^{\mathrm{P}}$, and $\nu_{\III}^{\mathrm{P}}$, superimposed (magenta curves). Formulas for the four curves $\nu_{\I}^{\F}$, $\nu_{\II}^{\F}$, $\nu_{\II}^{\mathrm{P}}$, and $\nu_{\III}^{\mathrm{P}}$ are given in (\ref{nuIFdef})--(\ref{nuIIIPdef}). The curves $\nu_{\II}^{\F}$ and $\nu_{\II}^{\mathrm{P}}$ merge at the common end-point $(\nu, U) = (1, 4\pi)$. 
		\end{figuretext}
	\end{center}
	\vspace{-0.4cm}
\end{figure}

\subsection{Construction of Hartree--Fock phase diagrams}
In this paper, we propose an analytical method for the construction of Hartree--Fock phase diagrams for the Hubbard model in different dimensions and for various other Hubbard-like models.
Such phase diagrams are traditionally constructed numerically, but we argue that, by using asymptotic techniques, it is possible to obtain analytic formulas approximating the curves separating the different phases. These approximations become exact in certain limits.
Our specific results in the present paper are for the 2D Hubbard model at zero temperature, and for the diagram predicting the different phases as a function of the on-site repulsion, $U > 0$, and the doping, $\nu\in[-1,1]$; the doping $\nu$ is defined as the average number of electrons per site minus one (see Appendix~\ref{modelapp} for the precise definition of $\nu$). 
Our analytic expansions become exact as any of the four edges of the phase diagram is approached.
In complementary on-going work, we present similar results for the (repulsive) Hubbard model on other lattices (including $\Z^n$ for $n=1,\infty$) \cite{LLM2025} and for the attractive 1D and 2D Hubbard models with an external magnetic field \cite{HLLLM2026}.
  
Our results for the 2D Hubbard model are summarized in Figure \ref{phasediagramfig}. 
As explained in Appendix \ref{modelapp}, the hopping parameter $t>0$ can be scaled out of the problem---in Figure \ref{phasediagramfig} and in what follows we have therefore set $t = 1$.
The phase diagram on the left is numerical and was obtained using Hartree--Fock theory restricted to P, F, and AF states; the phase boundaries were computed at the points indicated by dots, and the curves between these dots are interpolations. The phase diagram on the right shows the phase boundaries obtained analytically in this paper superimposed on the numerical phase diagram. Our formulas for the phase boundaries agree with, but also improve on, the numerical phase boundaries. For example, the system is in the P state for any $U>0$ if $\nu = \pm 1$ (see green vertical lines at $\nu = \pm 1$ in Figure \ref{phasediagramfig}), because the system is completely full if $\nu = 1$ and completely empty if $\nu = -1$. However, the numerical results leave open the question whether the state of the system for $U \gtrsim 12$ and doping  close to $\pm 1$ is P or F. Our analytic results answer this question and show that in the limit as the doping $\nu$ tends to $\pm 1$ with $|\nu| < 1$, the state of the system is P if $U < 4\pi$ and F if $U > 4\pi$.

\subsection{Mixed phases}\label{sec:mixed} 
The diagrams in Figure \ref{phasediagramfig} show whether the state of lowest energy has P, F, or AF structure for given values of the doping $\nu$ (horizontal axis) and on-site repulsion $U$ (vertical axis).
The yellow regions in Figure \ref{phasediagramfig} are mixed phases and are of particular interest. Indeed, these are regions where we expect to see unconventional physics. In a mixed phase, neither of the P, F, or AF solutions is stable by itself; instead the only possibility to obtain a stable solution within our Hartree--Fock ansatz is to let two conventional phases coexist and form a mixed phase. Although such a state consisting of two coexisting conventional phases is theoretically possible, it is more likely that the restriction to P, F, or AF states is too limited in a mixed phase. In other words, the physical state of the system in a mixed phase is expected to be a more exotic state where translation invariance is broken in complicated ways. (The interested reader can find a more detailed discussion of mixed states in \cite{LL2025}.)

A by-product of our analytic construction of the phase boundaries is that we can say with certainty that there are mixed regions in the phase diagram of Figure \ref{phasediagramfig}.
This provides, as far as we know, the first mathematical proof of existence of mixed regions in a Hartree--Fock phase diagram for a Hubbard-like model.

\begin{figure}
	\vspace{0.4cm}
	\begin{center}
		\hspace{-.6cm}
		\begin{overpic}[width=.49\textwidth]{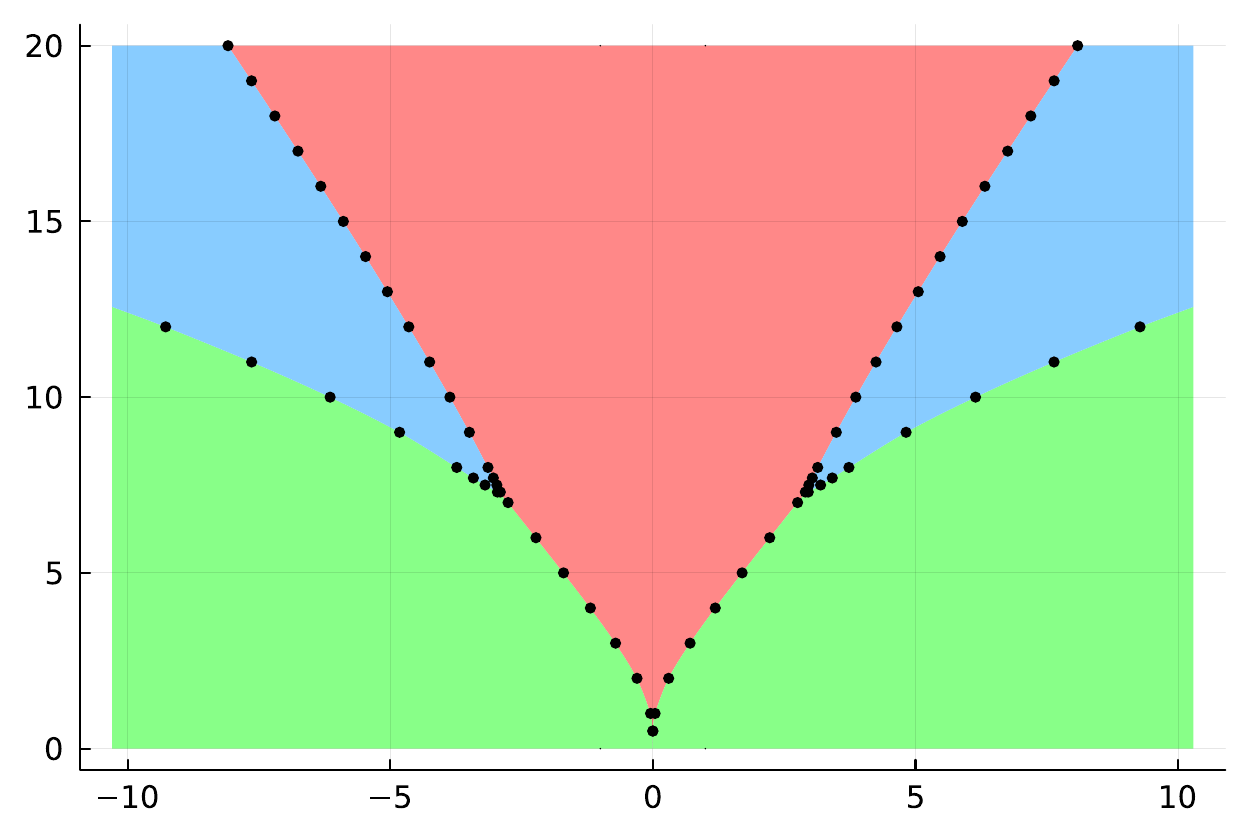 }
			\put(5.5,66.5){\footnotesize $U$}
			\put(99.5,4.5){\footnotesize $\mu$}
            \put(20,15){\footnotesize P}
            \put(80,15){\footnotesize P}
            \put(18,45){\footnotesize F}
            \put(84,45){\footnotesize F}				
			\put(49,45){\footnotesize AF}
		\end{overpic}
		\hspace{0cm}
		\begin{overpic}[width=.49\textwidth]{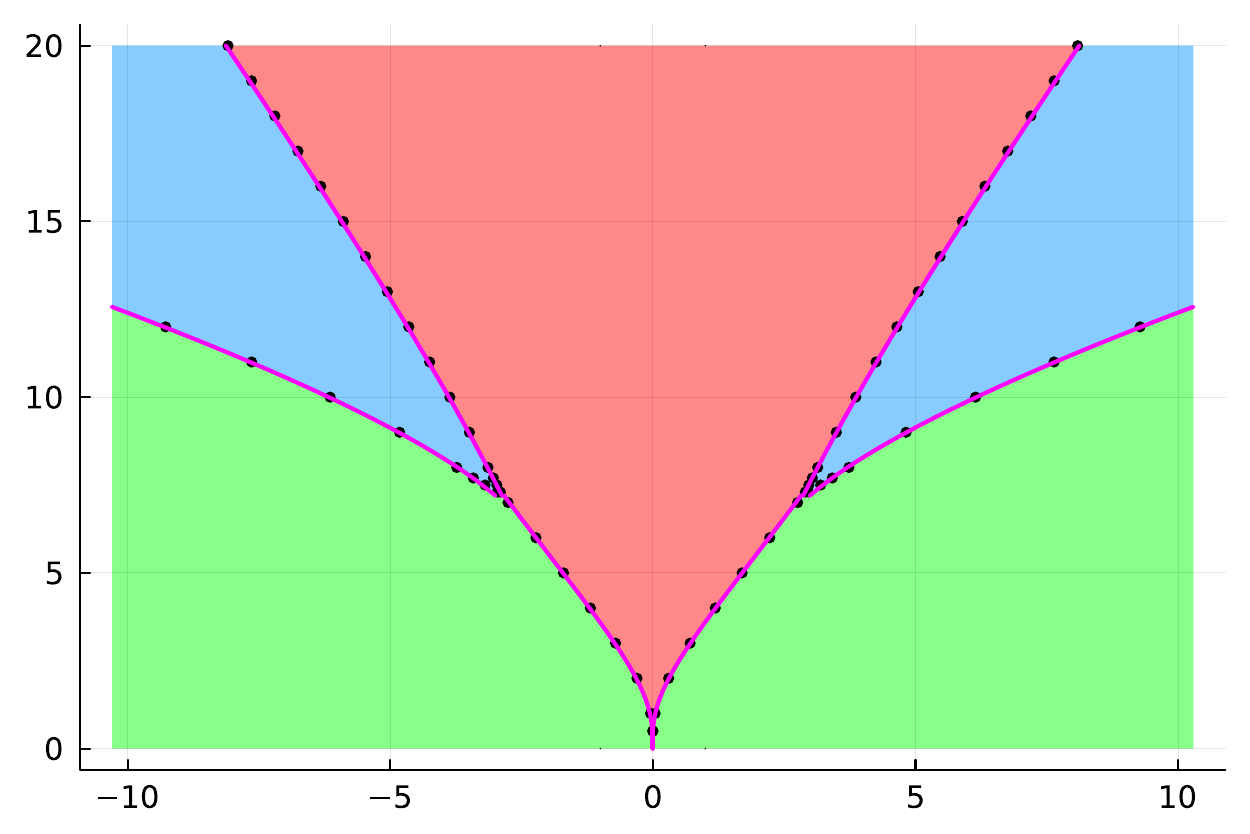 }
			\put(5.5,66.5){\footnotesize $U$}
			\put(99.5,4.5){\footnotesize $\mu$}
            \put(20,15){\footnotesize P}
            \put(80,15){\footnotesize P}
            \put(18,45){\footnotesize F}
            \put(84,45){\footnotesize F}				
            \put(69,50){\footnotesize $\mu_{\I}^{\app}$}		
            \put(82,33){\footnotesize $\mu_{\II}^{\app}$}		
            \put(57,14){\footnotesize $\mu_{\III}^{\app}$}		
			\put(49,45){\footnotesize AF}
		\end{overpic}
		 \begin{figuretext} \label{phasediagramUvsmufig} 
		Left: Same phase diagram as in Figure \ref{phasediagramfig} (left) except that here 
		 the phases are shown as a function of the chemical potential $\mu$ and the coupling $U$. 
			\\
		Right: Same phase diagram as on the left, but with the phase boundaries obtained analytically in this paper, $\mu_{\I}^{\app}$, $\mu_{\II}^{\app}$, and $\nu_{\III}^{\app}$, superimposed (magenta curves). Formulas for the curves $\mu_{\I}^{\app}$, $\mu_{\II}^{\app}$, and $\nu_{\III}^{\app}$ are given in (\ref{muIappdef})--(\ref{muIIIappdef}). 
		\end{figuretext}
	\end{center}
	\vspace{-0.4cm}
\end{figure}

The existence of mixed phases can be understood by considering Figure \ref{phasediagramUvsmufig}, which shows the same phase diagrams as Figure \ref{phasediagramfig}, except that in Figure \ref{phasediagramUvsmufig} the phases are shown as functions of the chemical potential $\mu$ (horizontal axis) and the on-site repulsion $U$ (vertical axis). 
Let us consider the phase boundary between the AF and F phases approximated by the curve $\mu_{\I}^{\app}$.
It turns out that, at a fixed $(\mu, U)$ infinitesimally to the left of this phase boundary, the AF state has the lowest free energy and has doping $\nu_{\I}^{\AF}(U) = 0$. On the other hand, at a fixed $(\mu, U)$ infinitesimally to the right of this phase boundary, the F state has the lowest free energy and has a strictly positive doping $\nu_{\I}^{\F}(U) > 0$. Thus, for a given $U$, there is no value of $\mu$ for which the state with minimal free energy among the P, F, and AF states has doping in the interval $(0, \nu_{\I}^{\F}(U))$. This gives rise to the mixed region $0 < \nu < \nu_{\I}^{\F}(U)$ in Figure \ref{phasediagramfig}. Similarly, the mixed regions $\nu_{\II}^{\F}(U) < \nu < \nu_{\II}^{\mathrm{P}}(U)$ and $0 < \nu < \nu_{\III}^{\mathrm{P}}(U)$ in Figure \ref{phasediagramfig} are due to discontinuities in the doping across the F-P and AF-P phase boundaries, respectively.

The occurrence of mixed phases was missed in the pioneering works on the Hartree--Fock phase diagrams of the Hubbard model on $\Z^n$ for $n=3$ in \cite{P1966} and $n=2$ in \cite{H1985}. This oversight was due to a shortcut which works for many models but which can give incorrect results for Hubbard-like models (this is explained in more detail in \cite{LLphysrevlong}). The phase diagrams obtained by using this shortcut predict an AF phase in a finite doping regime away from half-filling $\nu=0$  \cite{H1985}, which is in contradiction with results by other methods, as noted already in \cite{H1985}. In the corrected Hartree--Fock diagrams there is no such discrepancy. It is interesting to note that, for the 2D Hubbard model, mixed phases occur exactly in the parameter regions where real high-temperature superconductors exhibit exotic physics; see e.g.\ \cite{HTSC} for a schematic phase diagram of high-temperature superconductors. In forthcoming work, we show that Hartree--Fock theory of the 2D $t$-$t'$-$U$-$V$ model is much richer and produces phase diagrams that are more similar to phase diagrams of real high-temperature superconductors than Hartree--Fock theory of the $t$-$U$ model \cite{LLM2025}.

\subsection{The method}
Let us now describe our method. We first describe the general steps of the method, and then we explain how the curves in the right half of Figure \ref{phasediagramfig} were obtained by implementing these steps to the particular case of the 2D Hubbard model.

Our method consists of the following three steps:
\begin{enumerate}[1.]

\item Study the asymptotic behavior of the P, F, and AF mean-field equations in the limit as $(\nu, U)$ approaches one of the four edges of the phase diagram, i.e., in one of the following four limits: (i) $U \to +\infty$ (top edge of the phase diagram), (ii) $\nu \to 1$ (right edge of the phase diagram), (iii) $U \to 0$ (bottom edge of the phase diagram), (iv) $\nu \to -1$ (left edge of the phase diagram). 
Use that the functions in the mean-field equations have asymptotic expansions in the above limits to identify asymptotic sectors. The asymptotic sectors are characterized by the fact that the P, F, and AF mean-field equations have a fixed number of solutions within each sector. 

\item Find the expansion of the free energy for each of the AF mean-field solutions. Choose the solution corresponding to the free energy that is the smallest; the free energy of this solution is the AF free energy, denoted by $\mathcal{F}_{\AF}$, in the sector under consideration. In the same way, compute the P and F free energies, $\mathcal{F}_{\mathrm{P}}$ and $\mathcal{F}_{\F}$. It is important to compare free energies at the same chemical potential $\mu$ (and not at the same doping $\nu$) \cite{LL2025}; we therefore view $\mathcal{F}_{\mathrm{P}}$, $\mathcal{F}_{\F}$, and $\mathcal{F}_{\AF}$ as functions of $(U, \mu)$ (rather than as functions of $(U, \nu)$); see Appendix \ref{modelapp} for precise definitions of $\mu$ and $\nu$.

\item The curves separating the different phases in the phase diagram are obtained by finding curves on which $\mathcal{F}_{\AF} = \mathcal{F}_{\F}$, $\mathcal{F}_{\AF} = \mathcal{F}_{\mathrm{P}}$, or $\mathcal{F}_{\F} = \mathcal{F}_{\mathrm{P}}$. 
Compute asymptotic approximations to these curves by substituting the expansions obtained in the previous step into these equalities and solving the resulting equations order by order. This gives an asymptotic expression for the value of the chemical potential $\mu$ corresponding to the phase boundary as a function of $U$. By computing the dopings corresponding to this chemical potential, analytic approximations to the curves separating the different phases in the phase diagram are obtained.

\end{enumerate}

\subsection{The case of the 2D Hubbard model}\label{2Dintrosubsec}
We expect the above scheme to be applicable to a large class of Hubbard-like models. However, for reasons explained already above, we have decided to focus solely on the 2D Hubbard model at zero temperature in this paper. Below is a summary of how we obtained the phase diagram in Figure \ref{phasediagramfig} by implementing the above three steps (see Theorems \ref{sectorIth}--\ref{sectorIIIth} for precise statements). To the best of our knowledge, this is the first non-numerical construction of a phase diagram for a Hubbard-like model.

\begin{figure}
	\vspace{0.4cm}
	\begin{center}
		\begin{overpic}[width=.49\textwidth]{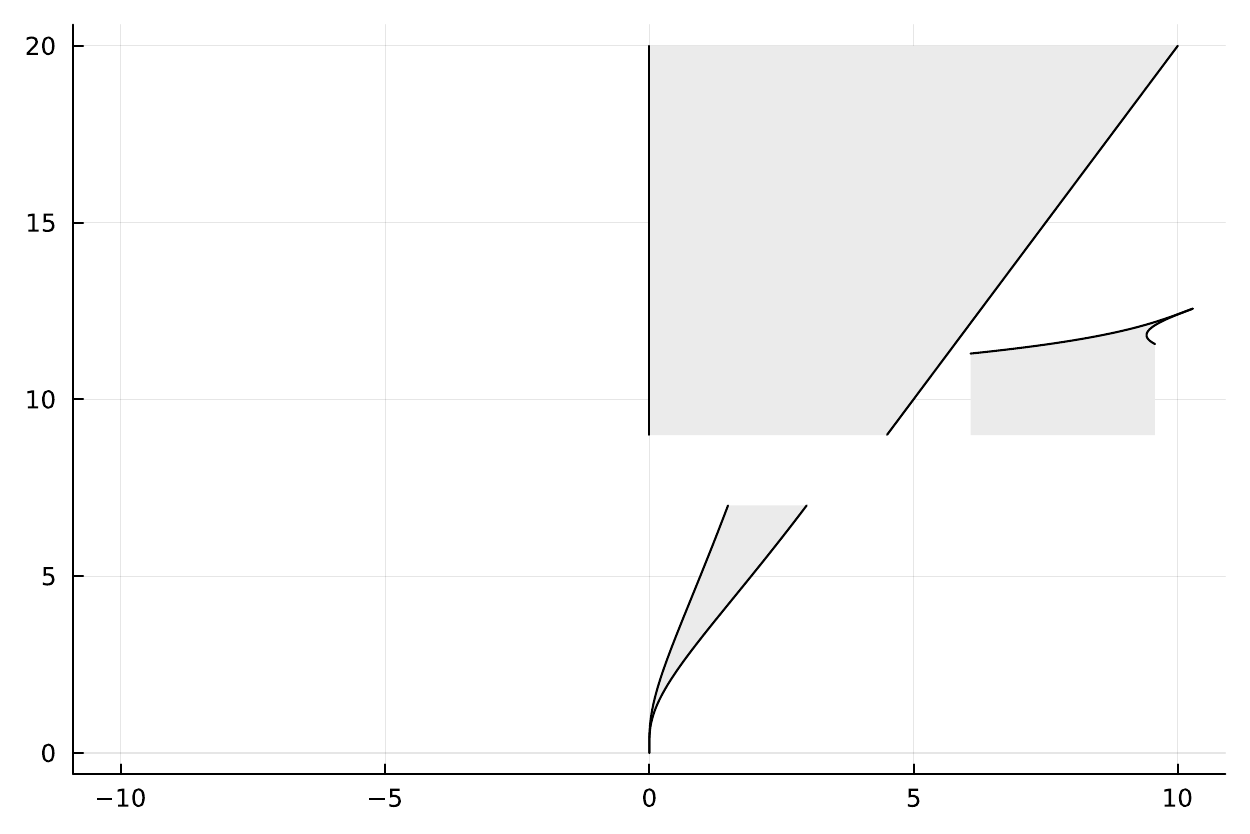 }
			\put(5,66.5){\footnotesize $U$}
			\put(99.5,4){\footnotesize $\mu$}
            \put(63,47){\footnotesize I}
            \put(85.5,34.5){\footnotesize II}
            \put(57,21){\footnotesize III}
		\end{overpic}
		\begin{overpic}[width=.49\textwidth]{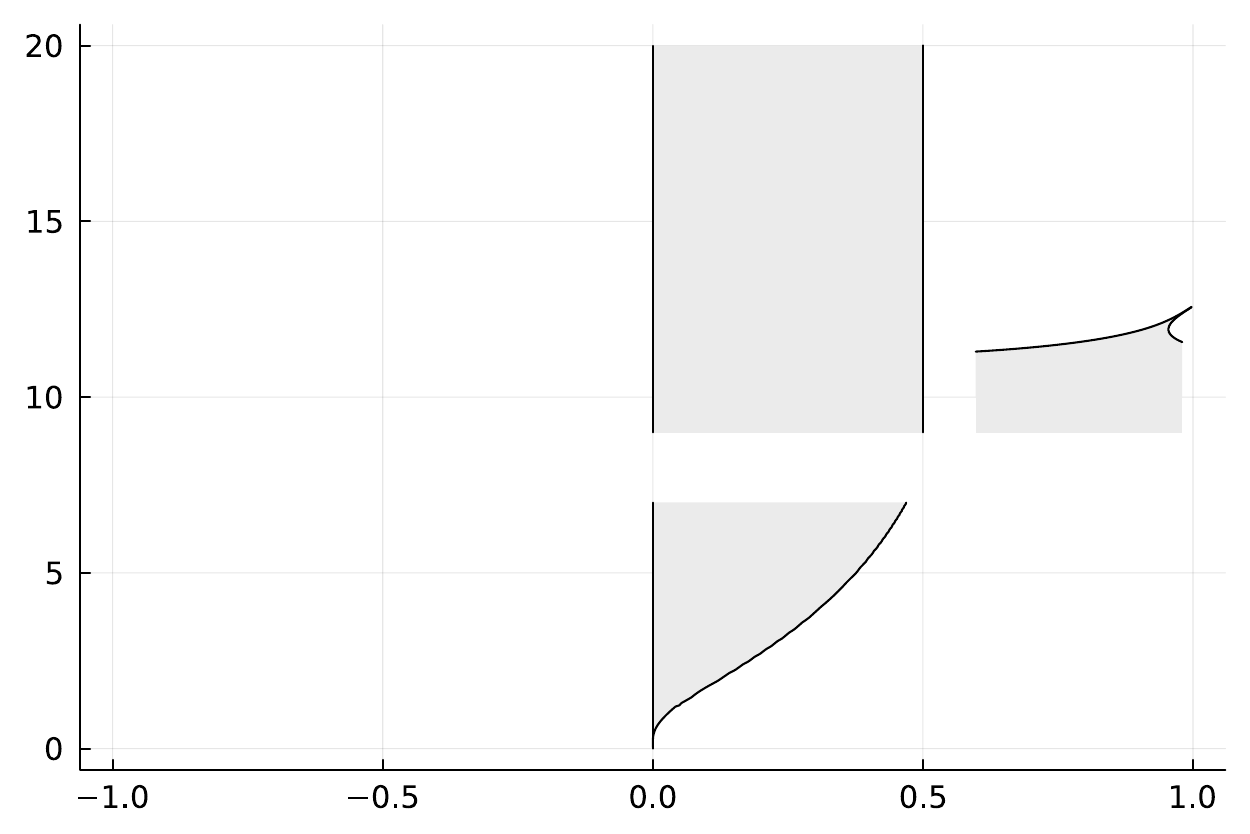 }
			\put(5,66.5){\footnotesize $U$}
			\put(99.5,4){\footnotesize $\nu$}
            \put(63,47){\footnotesize I}
            \put(85.5,34.5){\footnotesize II}
            \put(57,19){\footnotesize III}
		\end{overpic}
		 \begin{figuretext}\label{sectorsfig} 
		  Left: The three asymptotic sectors I, II, and III 
		defined in (\ref{sectorIdef})--(\ref{sectorIIIdef}) for $\delta = 0.001$, $M=1$, $U_0=9$ (for I and II), and $U_0=7$ (for III). 	
		By symmetry, it is enough to consider asymptotic sectors in the region $\mu\geq 0$.
		 \\
		  Right: The corresponding subsets of the $(\nu, U)$-plane. 
		  \end{figuretext}
	\end{center}
	\vspace{-0.4cm}
\end{figure}

\begin{enumerate}[1.]

\item The P, F, and AF mean-field equations for the 2D Hubbard model are given in (\ref{Pmeanfield})--(\ref{AFmeanfield}) below. By particle-hole symmetry, the phase diagram is symmetric under $\nu \to -\nu$, it is therefore sufficient to consider $\nu \geq 0$, or equivalently, $\mu \geq 0$ (see e.g.\ \cite[Appendix~A.3]{LLphysrevlong} for details on this symmetry). We identify the three asymptotic sectors I, II, and III displayed in Figure \ref{sectorsfig}. It is possible to consider other sectors as well, but sectors I, II, and III contain phase boundary curves, and are therefore crucial for the construction of the phase diagram. An extension of the analysis to other asymptotic sectors would determine the phase of the system with mathematical certainty in additional regions, but this is beyond the scope of the current paper.

\item In Sector I, it turns out that there is one P mean-field solution, one F mean-field solution, and one AF mean-field solution. Moreover, there is a curve $\mu = \mu_{\I}(U)$ such that the AF free energy is the smallest for $\mu < \mu_{\I}(U)$ while the F free energy is the smallest for $\mu > \mu_{\I}(U)$. Our method shows that $\mu_{\I}(U) = \mu_{\I}^{\app}(U) + O\big(U^{-9/2}\big)$ as $U \to +\infty$, where the approximation $\mu_{\I}^{\app}(U)$ is given by
\begin{align}\nonumber
\mu_{\I}^{\app}(U) = &\; \frac{U}{2}  - 4 + \frac{4\sqrt{2\pi}}{\sqrt{U}} 
- \frac{2\pi}{3 U} 
-\frac{5 \pi^{3/2}}{36 \sqrt{2} U^{3/2}} 
-\frac{11 \pi^2}{270 U^{2}} 
-\frac{\sqrt{\frac{\pi }{2}} \left(691200+1163 \pi^2\right)}{34560 U^{5/2}} 
	\\ \nonumber
& + \bigg(\frac{10 \pi }{3}-\frac{18071 \pi^3}{1088640} \bigg) \frac{1}{U^3}
+ \frac{\pi^{3/2} \left(51840000-907207 \pi^2\right)}{49766400 \sqrt{2} U^{7/2}} 
	\\  \label{muIappdef}
& + \bigg(\frac{11 \pi^2}{27} - \frac{561913\pi^{4}}{52254720}\bigg)\frac{1}{U^4}.
\end{align} 

In Sector II, there is one P mean-field solution, two F mean-field solutions, and no AF mean-field solution. Moreover, there is a curve $\mu = \mu_{\II}(U)$ such that the F free energy is the smallest for $\mu < \mu_{\II}(U)$ while the P free energy is the smallest for $\mu > \mu_{\II}(U)$. Our method shows that $\mu_{\II}(U) = \mu_{\II}^{\app}(U) + O\big((4\pi - U)^7\big)$ as $U \uparrow 4\pi$ (the notation $U \uparrow 4\pi$ indicates that $U$ approaches $4\pi$ from below), where 
\begin{align}\nonumber
\mu_{\II}^{\app}(U) & = 4 + 2\pi 
- \frac{(8+\pi ) }{2 \pi } (4 \pi - U)
+ \frac{7}{12 \pi^2} (4 \pi - U)^2
+\frac{17}{288 \pi^3} (4 \pi - U)^3
	\\ \label{muIIappdef}
&+\frac{1861}{138240 \pi^4} (4 \pi - U)^4
+ \frac{15181}{3317760 \pi^5} (4 \pi - U)^5
 + \frac{469909}{247726080 \pi^6} (4 \pi - U)^6.
\end{align}

In Sector III, there is one P mean-field solution, no F mean-field solution, and two AF mean-field solutions. Moreover, there is a curve $\mu = \mu_{\III}(U)$ such that the AF free energy is the smallest for $\mu < \mu_{\III}(U)$ while the P free energy is the smallest for $\mu > \mu_{\III}(U)$. We will show that $\mu_{\III}(U) = \mu_{\III}^{\app}(U) + O\big(U^{5/2} e^{-\frac{2\pi}{\sqrt{U}}}\big)$ as $U \downarrow 0$ (the notation $U \downarrow 0$ indicates that $U$ approaches $0$ from above), where 
\begin{align}\nonumber
\mu_{\III}^{\app}(U) = &\; \bigg\{ 16 \sqrt{2} 
+ \frac{2 \sqrt{2} (2+\ln{2})}{\pi } \sqrt{U}
+\frac{(\ln(2)-6) (\ln(8)-2)}{4 \sqrt{2} \pi^2} U 
	\\ \nonumber
& + \frac{56+\ln(2) (76+5 (\ln(2)-6) \ln{2})}{32 \sqrt{2} \pi^3} U^{3/2}
	\\ \label{muIIIappdef}
& + \frac{(\ln(2)-6) (1320+\ln(2) (396+5 \ln(2) (21 \ln(2)-10)))}{3072 \sqrt{2} \pi^4} U^2 \bigg\} e^{-\frac{2\pi}{\sqrt{U}}}.
\end{align} 
The curves $\mu_{\I}^{\app}$, $\mu_{\II}^{\app}$, and $\mu_{\III}^{\app}$ are displayed in Figure \ref{phasediagramUvsmufig} (right).

\item At the phase boundary $\mu = \mu_{\I}(U)$ in Sector I, the doping of the AF state is $\nu = 0$ and the doping of the F state, $d_0^{\F}$, is given by
$$d_0^{\F} = \nu_{\I}^{\F}(U) + O\big(U^{-7/2}\big) \qquad \text{as $U \to +\infty$},$$
where 
\begin{align}\nonumber
\nu_{\I}^{\F}(U) := &\; \sqrt{\frac{2}{\pi}}\frac{1}{\sqrt{U}} + \frac{1}{3U}
+ \frac{31 \sqrt{\frac{\pi }{2}}}{144 U^{3/2}} 
+ \frac{203 \pi }{2160 U^2} 
	\\ \label{nuIFdef}
& + \frac{13573 \pi^2-691200}{138240 \sqrt{2 \pi } U^{5/2}}
 + \bigg(\frac{979 \pi^2}{17010}-\frac{5}{3}\bigg) \frac{1}{U^3}.
 \end{align}
The curve $\nu_{\I}^{\F}$ is one of the magenta curves shown in Figure \ref{phasediagramfig}.
Since $0 < \nu_{\I}^{\F}(U)$, there is a mixed phase between the AF and F phases (yellow region between AF and F phases in Figure \ref{phasediagramfig}).

At the phase boundary $\mu = \mu_{\II}(U)$ in Sector II, the doping of the F state is given by
$$d_0^{\F} = \nu_{\II}^{\F}(U) + O\big((4\pi - U)^7\big) \qquad \text{as $U \uparrow 4\pi$},$$
and the doping of the P state is given by
$$d_0^{\mathrm{P}} = \nu_{\II}^{\mathrm{P}}(U) + O\big((4\pi - U)^7\big) \qquad \text{as $U \uparrow 4\pi$},$$
where 
\begin{align} \nonumber
\nu_{\II}^{\F}(U) := &\; 1 - \frac{1}{\pi^2} (4\pi - U)
-\frac{5}{48 \pi^3}(4\pi - U)^2
-\frac{19}{1152 \pi^4}  (4\pi - U)^3
	\\ \label{nuIIFdef}
& -\frac{2039}{552960 \pi^5}  (4\pi - U)^4
-\frac{3691}{2654208 \pi^6}   (4\pi - U)^5
 -\frac{2369993}{2972712960 \pi^7} (4\pi - U)^6
\end{align}
and
\begin{align} \nonumber
\nu_{\II}^{\mathrm{P}}(U) := &\; 1 - \frac{1}{\pi^2} (4\pi - U)
-\frac{1}{24 \pi^3} (4\pi - U)^2
+ \frac{1}{144 \pi^4} (4\pi - U)^3
	\\ \label{nuIIPdef}
& + \frac{2851}{552960 \pi^5} (4\pi - U)^4
+  \frac{15839}{6635520 \pi^6} (4\pi - U)^5
+ \frac{207463}{198180864 \pi^7} (4\pi - U)^6.
\end{align}
The curves $\nu_{\II}^{\F}$ and  $\nu_{\II}^{\mathrm{P}}$ are two of the magenta curves shown in Figure \ref{phasediagramfig}.
Since $\nu_{\II}^{\F}(U) < \nu_{\II}^{\mathrm{P}}(U)$, there is a mixed phase between the F and P phases (yellow region between F and P phases in Figure \ref{phasediagramfig}).

At the phase boundary $\mu = \mu_{\III}(U)$ in Sector III, the doping of the AF state is $\nu = 0$ and the doping of the P state is given by
\begin{align}\label{d0PnuIII}
d_0^{\mathrm{P}} = \nu_{\III}^{\mathrm{P}}(U) + O\big(\sqrt{U} e^{-\frac{2\pi}{\sqrt{U}}}\big) \qquad \text{as $U \downarrow 0$},\end{align}
where 
\begin{align}\label{nuIIIPdef}
\nu_{\III}^{\mathrm{P}}(U) := \bigg\{ \frac{32 \sqrt{2}}{\pi  \sqrt{U}}
-\frac{4 \sqrt{2} (2+\ln{2})}{\pi^2}
+\frac{1}{2} \frac{(22 - \ln{2}) (2+\ln{2})}{2 \sqrt{2} \pi^3} \sqrt{U}\bigg\} e^{-\frac{2\pi}{\sqrt{U}}};
\end{align}
see Remark \ref{expansionsremark} for a discussion of the size of the error term in (\ref{d0PnuIII}).
The curve $\nu_{\III}^{\mathrm{P}}$ is one of the magenta curves shown in Figure \ref{phasediagramfig}.
Since $0 < \nu_{\III}^{\mathrm{P}}(U)$, there is a mixed phase between the AF and P phases (yellow region between AF and P phases in Figure \ref{phasediagramfig}).
\end{enumerate}

In addition to providing the above analytic expressions (\ref{nuIFdef})--(\ref{nuIIIPdef}) for the phase boundary curves, our method also gives asymptotic expansions for the free energies, as well as for the magnetization of the F and AF states, in each asymptotic sector. 
Let us also mention that all expansions can be extended to arbitrary high order at the expense of more laborious calculations.

\subsection{Organization of the paper}
Our results for the 2D Hubbard model are presented in Section \ref{mainsec} in the form of three theorems (Theorems \ref{sectorIth}--\ref{sectorIIIth}). These theorems treat the three asymptotic sectors I, II, and III, respectively. 
Basic properties of the P, F, and AF Hartree--Fock functions are established in Sections \ref{Psec}, \ref{Fsec}, and \ref{AFsec}, respectively.
The proofs of the three theorems are presented in Sections \ref{SectorIproofsec}, \ref{SectorIIproofsec}, and \ref{SectorIIIproofsec}, respectively. 
Finally, Section \ref{conclusionssec} contains some conclusions. 
For completeness, a mathematical definition of the Hubbard model is included in Appendix \ref{modelapp}. In Appendix \ref{modelapp}, expressions for the P, F, and AF free energies are also derived. 
In Appendix \ref{N0app}, we derive several properties that we need of the function $N_{0}(\epsilon)$ defined in (\ref{2Ddensityofstates}). Appendix \ref{Delta1primeapp} contains the proof of a technical lemma.

\section{Theorems for the 2D Hubbard model}\label{mainsec}
In this section, we state three theorems for the 2D Hubbard model at zero temperature. Before stating the theorems, we need to introduce the free energy densities of the P, F, and AF states for the 2D Hubbard model, and define the three asymptotic sectors I, II, and III.

\begin{figure}
\bigskip
\begin{center}
\begin{overpic}[width=.46\textwidth]{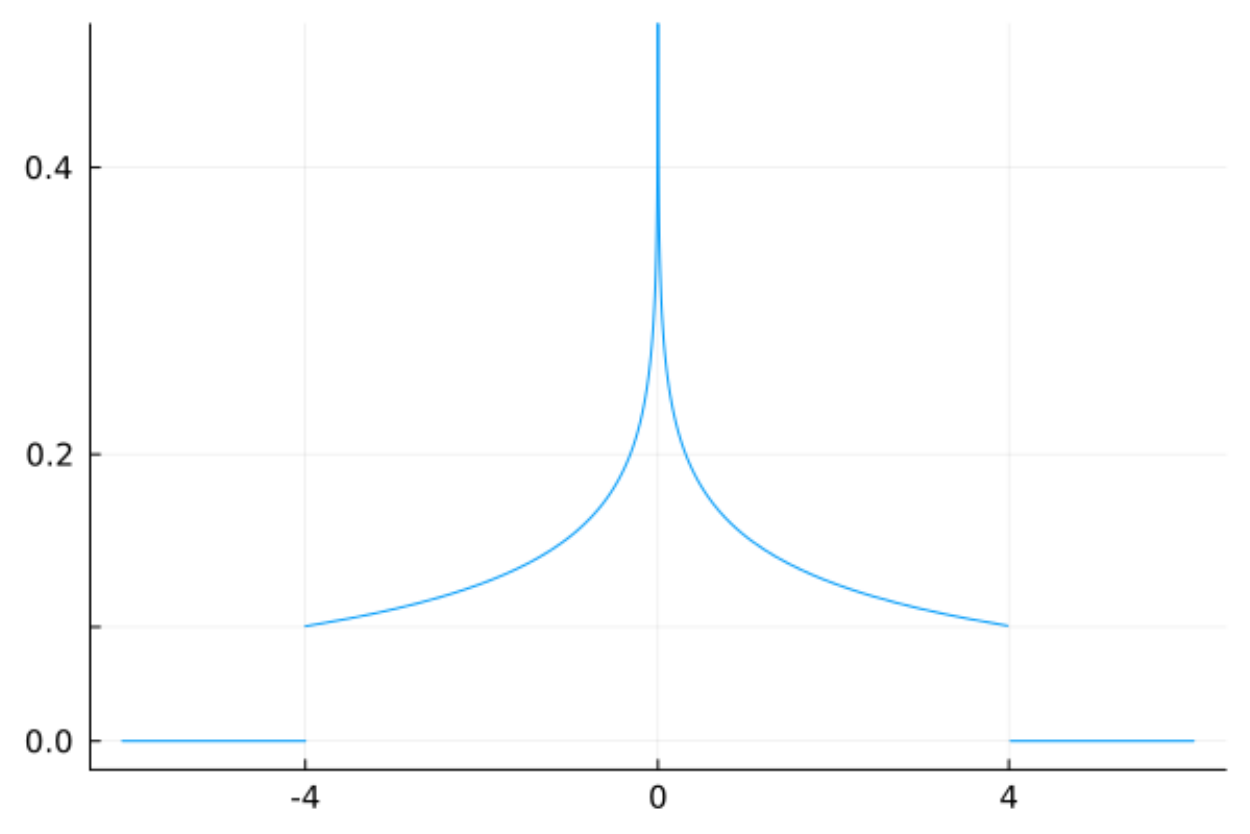}
      \put(2,68){\footnotesize $N_{0}(\epsilon)$}
      \put(101,3.7){\footnotesize $\epsilon$}
       \put(1,15.4){\footnotesize $\frac{1}{4\pi}$}
   \end{overpic}
        \begin{figuretext}\label{N0fig}
       The function $N_0(\epsilon)$ defined in (\ref{2Ddensityofstates}) has a logarithmic singularity at $\epsilon = 0$ and jump discontinuities at $\epsilon = \pm 4$. 
      \end{figuretext}
     \end{center}
\end{figure}

\subsection{The P, F, and AF free energies}
The free energy densities of the P, F, and AF states for the zero temperature 2D Hubbard model are denoted by $\mathcal{F}_{\mathrm{P}}(U, \mu)$, $\mathcal{F}_{\F}(U, \mu)$, and $\mathcal{F}_{\AF}(U, \mu)$, and are defined as follows. Define the density of states $N_{0}:\R \to [0, +\infty]$ by (the graph of $N_0$ is displayed in Figure \ref{N0fig})
\begin{align}\label{2Ddensityofstates}
N_{0}(\epsilon) = \int_{[-\pi,\pi]^2} \delta\big(2 (\cos{k_1} + \cos{k_2}) + \epsilon\big) \frac{dk_1dk_2}{(2\pi)^2},
\end{align}
where $\delta(\cdot)$ is the Dirac delta function.
Let $U > 0$ be the on-site repulsion in the Hubbard model, and let $\mu \in \R$ be the chemical potential. 
Define the {\it Hartree--Fock functions} $\mathcal{G}_{\mathrm{P}}$, $\mathcal{G}_{\F}$, and $\mathcal{G}_{\AF}$ by
\begin{align}\label{PHartreeFockFunction}
\mathcal{G}_{\mathrm{P}}(d_0, U, \mu) = & -  \frac{U}{4} d_0^2 + \frac{U}{2} d_0 - \mu - 2\int_{\R} N_0(\epsilon)\bigg(\frac{U}{2} d_0 - \mu + \epsilon\bigg) \theta\bigg(\frac{U}{2} d_0 - \mu + \epsilon\bigg) d\epsilon,
	\\ \nonumber
\mathcal{G}_{\F}(d_0, m_0, U, \mu) = &\; \frac{U}{4}  \big(m_0^2 - d_0^2\big)
 + \frac{U}{2} d_0 - \mu
 	\\ \nonumber
& - \int_{\R} N_0(\epsilon) \bigg(\frac{U(d_0 - m_0)}{2}  - \mu + \epsilon\bigg)\theta\bigg(\frac{U(d_0 - m_0)}{2}  - \mu + \epsilon\bigg)d\epsilon
	\\ \label{FHartreeFockFunction}
& - \int_{\R}  N_0(\epsilon) \bigg(\frac{U(d_0 + m_0)}{2}  - \mu + \epsilon \bigg) \theta\bigg(\frac{U(d_0 + m_0)}{2}  - \mu + \epsilon \bigg) d\epsilon,
	\\ \nonumber
\mathcal{G}_{\AF}(d_0, m_1, U, \mu) = &\; \frac{U}{4}  \big(m_1^2 - d_0^2\big)
+ \frac{U}{2} d_0 - \mu 
	\\ \nonumber
&\hspace{-2cm} - \int_{\R} N_0(\epsilon) \bigg(\frac{U}{2} d_0 - \mu + \sqrt{\frac{U^2}{4}m_1^2 + \epsilon^2}\bigg) \theta\bigg(\frac{U}{2} d_0 - \mu + \sqrt{\frac{U^2}{4}m_1^2 + \epsilon^2}\bigg)d\epsilon
	\\\label{AFHartreeFockFunction}
&\hspace{-2cm} - \int_{\R} N_0(\epsilon)  \bigg(\frac{U}{2} d_0 - \mu - \sqrt{\frac{U^2}{4}m_1^2 + \epsilon^2}\bigg) \theta\bigg(\frac{U}{2} d_0 - \mu - \sqrt{\frac{U^2}{4}m_1^2 + \epsilon^2}\bigg) d\epsilon,
\end{align}
where $d_0 \in \R$, $m_0, m_1 \in [0, +\infty)$, and $\theta(\cdot)$ is the Heaviside function; see Appendix \ref{modelapp} for the origin of the expressions (\ref{PHartreeFockFunction})--(\ref{AFHartreeFockFunction}). Given $U > 0$ and $\mu \in \R$, the {\it P mean-field equation} is given by
\begin{align}\label{Pmeanfield}
\frac{\partial \mathcal{G}_{\mathrm{P}}}{\partial d_0}(d_0^{\mathrm{P}}, U, \mu) = 0;
\end{align}
the {\it F mean-field equations} are given by
\begin{align}\label{Fmeanfield}
\frac{\partial \mathcal{G}_{\F}}{\partial d_0}(d_0^{\F}, m_0^{\F}, U, \mu) = 0, \qquad \frac{\partial \mathcal{G}_{\F}}{\partial m_0}(d_0^{\F}, m_0^{\F}, U, \mu) = 0;
\end{align}
and the {\it AF mean-field equations} are given by
\begin{align}\label{AFmeanfield}
\frac{\partial \mathcal{G}_{\AF}}{\partial d_0}(d_0^{\AF}, m_1^{\AF}, U, \mu) = 0, \qquad \frac{\partial \mathcal{G}_{\AF}}{\partial m_1}(d_0^{\AF}, m_1^{\AF}, U, \mu) = 0.
\end{align}
The {\it free energy densities} of the P, F, and AF states for the 2D Hubbard model at zero temperature in the Hartree--Fock approximation are defined by
\begin{subequations}\label{freeenergies}
\begin{align}\label{freeenergiesP}
\mathcal{F}_{\mathrm{P}}(U, \mu) & = \min\{\mathcal{G}_{\mathrm{P}}(d_0^{\mathrm{P}}, U, \mu) \, |\, \text{$d_0^{\mathrm{P}} \in \R$ solves (\ref{Pmeanfield})}\},
	\\\label{freeenergiesF}
\mathcal{F}_{\F}(U, \mu) & = \min\{\mathcal{G}_{\F}(d_0^{\F}, m_0^{\F}, U, \mu) \, |\, \text{$(d_0^{\F}, m_0^{\F}) \in \R \times (0,+\infty)$ solves (\ref{Fmeanfield})}\},
	\\\label{freeenergiesAF}
\mathcal{F}_{\AF}(U, \mu) & = \min\{\mathcal{G}_{\AF}(d_0^{\AF}, m_1^{\AF}, U, \mu) \, |\, \text{$(d_0^{\AF}, m_1^{\AF}) \in \R \times (0,+\infty)$ solves (\ref{AFmeanfield})}\}.
\end{align}
\end{subequations}
We often refer to the free energy densities simply as the {\it free energies}.
In (\ref{freeenergies}) and in many other places, we have suppressed the $(U, \mu)$-dependence of $d_0^{\mathrm{P}}$, $(d_0^{\F},m_0^{\F})$, and $(d_0^{\AF},m_1^{\AF})$.
Observe that if $d_0^{\mathrm{P}}$ solves the P mean-field equation, then $(d_0^{\F}, m_0^{\F}) = (d_0^{\mathrm{P}}, 0)$ solves the F mean-field equations. However, the solution $(d_0^{\mathrm{P}}, 0)$ is excluded in our definition of $\mathcal{F}_{\F}$. In other words, to compute $\mathcal{F}_{\F}$, we only compare the values of the F Hartree--Fock function $\mathcal{G}_{\F}$ at the F mean-field solutions with $m_0^{\F} > 0$; a similar remark applies to $\mathcal{F}_{\AF}$.

It turns out that $\mathcal{G}_{\mathrm{P}}(d_0, U, \mu)$, $\mathcal{G}_{\F}(d_0, m_0, U, \mu)$, and $\mathcal{G}_{\AF}(d_0, m_1, U, \mu)$ are strictly concave functions of $d_0$ (see Lemma \ref{Plemma1}, (\ref{d2GFdd02}), and (\ref{d2GAFdd02})). In particular, any solution of \eqref{Pmeanfield}, \eqref{Fmeanfield}, or \eqref{AFmeanfield} corresponds to a maximum in $d_{0}$. The P mean-field equation always has a unique real solution $d_0^{\mathrm{P}}$, so $\mathcal{F}_{\mathrm{P}}(U, \mu)$ is well-defined for all $U> 0$ and $\mu \in \R$, and may be expressed more simply as
$$\mathcal{F}_{\mathrm{P}}(U, \mu) = \mathcal{G}_{\mathrm{P}}(d_0^{\mathrm{P}}, U, \mu),$$
where $d_0^{\mathrm{P}} \in \R$ is the unique solution of (\ref{Pmeanfield}), see Lemma \ref{Plemma1}.
On the other hand, the F and AF mean-field equations do not in general have unique solutions with $m_0^{\F} > 0$ and $m_{1}^{\AF} > 0$, respectively.
For example, as mentioned already in Section \ref{2Dintrosubsec}, in Sectors I, II, and III, the F mean-field equations have 1, 2, and 0 solutions whereas the AF mean-field equations have 1, 0, and 2 solutions, respectively. If there are no F or AF solutions for some given value of $(U, \mu)$, then the corresponding free energy is not defined at that point.

Within our approximation (i.e., within restricted Hartree--Fock theory restricted to P, F, and AF states), the {\it free energy}  $\mathcal{F}(U, \mu)$ of the system is given by
$$\mathcal{F}(U, \mu) = \min\Big\{\mathcal{F}_{\mathrm{P}}(U, \mu), \mathcal{F}_{\F}(U, \mu), \mathcal{F}_{\AF}(U, \mu)\Big\};$$
if $\mathcal{F}_{\AF}(U, \mu)$ is not defined, this should be interpreted as $\mathcal{F}(U, \mu) = \min\{\mathcal{F}_{\mathrm{P}}(U, \mu), \mathcal{F}_{\F}(U, \mu)\}$, etc.
We say that the system is in the P phase at $(U, \mu)$ if the value of $\mathcal{F}_{\mathrm{P}}$ is smaller than the values of $\mathcal{F}_{\F}$ and $\mathcal{F}_{\AF}$ at $(U, \mu)$, that it is in the F phase if $\mathcal{F}_{\F}$ is the smallest, and that it is in the AF phase if $\mathcal{F}_{\AF}$ is the smallest. 
Observe that
$$\mathcal{G}_{\mathrm{P}}(d_0, U, \mu) = \mathcal{G}_{\AF}(d_0, 0, U, \mu)
= \mathcal{G}_{\F}(d_0, 0, U, \mu), $$
which means that the P phase can be viewed as the limiting case $m_0 = 0$ of the F phase or as the limiting case $m_1 = 0$ of the AF phase.

\subsection{Physical interpretation}
The physical interpretation of the P, F, and AF states is as follows, see also Appendix \ref{modelapp}.
Suppose $d_0^{\mathrm{P}}$, $(d_0^{\F}, m_{0}^{\F})$, and $(d_0^{\AF}, m_1^{\AF})$ are points at which the three right-hand sides of (\ref{freeenergies}) are minimized. If the system is in the P phase for some given value of $(U, \mu)$, then the expectation value of the electron density is given by $d_0^{\mathrm{P}} + 1$ at every lattice site $\mathbf{x} \in \Z^2$, whereas the expectation value of the spin operator vanishes at every site (no magnetic order). 
If the system is in the F phase, then the expectation value of the electron density is $d_0^{\F} + 1$ and the expectation value of the spin operator is $m_0^{\F} \vec{e}$ at every site for some unit vector $\vec{e} \in \R^3$  (ferromagnetic order). If the system is in the AF phase, then the expectation value of the electron density is $d_0^{\AF} + 1$ and the expectation value of the spin operator is $(-1)^{x_1+x_2}m_1^{\AF} \vec{e}$ at every site $\mathbf{x} = (x_1, x_2) \in \Z^2$ for some unit vector $\vec{e} \in \R^3$ (antiferromagnetic order).
By rotation invariance, we may assume that rotation invariance is broken along the $z$-axis in the magnetic states, so that $\vec{e} = (0,0,1)$, but this is of no relevance for us here.

We refer to the quantities $d_0^{\mathrm{P}}$, $d_0^{\F}$, and $d_0^{\AF}$ as the {\it dopings} of the P, F, and AF states, because they measure the electron density normalized such that doping $0$ corresponds to the system being half-filled (i.e., on average one electron per site), 
doping $+1$ corresponds to the system being completely filled (two electrons per site), and doping $-1$ corresponds to the system being empty (no electrons). We will see in Sections \ref{Psec}--\ref{AFsec} that all solutions of the P, F, and AF mean-field equations \eqref{Pmeanfield}--\eqref{AFmeanfield} have doping values in the interval $[-1,1]$, and magnetizations $m_0^{\F}, m_1^{\AF}$ in the interval $[0,1]$, as expected from the physics interpretation.

\subsection{The three asymptotic sectors}
Let us define the three asymptotic sectors I, II, and III. These sectors are defined as subsets of the $(U, \mu)$-plane and are shown as such in Figure \ref{sectorsfig} (left). However, by computing the dopings corresponding to the boundaries of the sectors, we can understand what regions they correspond to in the phase diagram (i.e., in the $(\nu, U)$-plane); these regions are displayed in Figure \ref{sectorsfig} (right). It is important to note that the curve $\nu_{\I}^{\F}$ shown in Figure \ref{phasediagramfig} (right) lies in the region corresponding to Sector I, the curves $\nu_{\II}^{\F}$ and $\nu_{\II}^{\mathrm{P}}$ lie in the region corresponding to Sector II, and the curve $\nu_{\III}^{\mathrm{P}}$ lies in the region corresponding to Sector III.

For $\delta>0$ and $U_0\geq 2\delta$, we define the asymptotic sector $\I = \I_{U_0, \delta}$ by
\begin{align}\label{sectorIdef}
\I_{U_0, \delta}= \Big\{(U, \mu) \in \R^2 \, \Big| \, U \geq U_0 \;\text{and}\; \mu \in [0, \tfrac{U}{2} - \delta] \Big\}.
\end{align}

For $U_0 \in (0, 4\pi)$ and $M > 0$, we define the asymptotic sector $\II = \II_{U_0, M}$ by
\begin{align}\label{sectorIIdef}
\II_{U_0, M} = \big\{(U, \mu) \in \R^2 \, \big| \,  U \in [U_0, 4\pi)  \;\text{and}\;  |\mu - \mu_{\II,0}(U)| \leq M (4\pi - U)^3 \big\} 
\end{align}
where
\begin{align}\label{muII0def}
\mu_{\II,0}(U) := 2 \pi +4 - \frac{8+\pi}{2 \pi}(4 \pi -U) + \frac{7}{12 \pi^2} (4 \pi - U)^2.
\end{align}

For $0<\delta<8$ and $U_0>0$, we define the asymptotic sector $\III = \III_{U_0, \delta}$ by
\begin{align}\label{sectorIIIdef}
\III_{U_0, \delta}= \Big\{(U, \mu) \in \R^2 \, \Big| \, U \in (0, U_0] \;\text{and}\; \mu \in \Big[(16 + \delta)e^{-\frac{2\pi}{\sqrt{U}}} , (32 - \delta)e^{-\frac{2\pi}{\sqrt{U}}}\Big] \Big\}.
\end{align}

\subsection{First theorem: Sector I}
Our first theorem shows that in Sector I, there is a curve $\mu_{\I}(U)$ such that the system is in the AF phase for $\mu < \mu_{\I}(U)$, whereas it is in the F phase for $\mu > \mu_{\I}(U)$. It also gives expressions for the associated dopings of these phases: the theorem shows that whereas the AF state has zero doping  (i.e., $d_0^{\AF} = 0$), the F state has strictly positive doping (i.e., $d_0^{\F} > 0$) at $\mu = \mu_{\I}(U)$. In particular, this establishes the existence of a mixed phase between the AF and F phases in the phase diagram, see Figure~\ref{phasediagramfig}.

\begin{theorem}[The AF-Mixed-F interface]\label{sectorIth}
Let $\delta >0$. There exist a $U_0>0$ and a smooth real-valued function $\mu_{\I}(U)$ of $U\in[U_0, +\infty)$ such that if $(U, \mu) \in \I_{U_0, \delta}$, then the following hold:
\begin{enumerate}[$(i)$]
\item If $\mu < \mu_{\I}(U)$, then the system exhibits AF order, i.e.,
$$\mathcal{F}_{\AF}(U, \mu) < \mathcal{F}_{\F}(U, \mu) < \mathcal{F}_{\mathrm{P}}(U, \mu).$$ 

\item If $\mu = \mu_{\I}(U)$, then the AF and F states have the same free energy, i.e.,
$$\mathcal{F}_{\AF}(U, \mu) = \mathcal{F}_{\F}(U, \mu) < \mathcal{F}_{\mathrm{P}}(U, \mu).$$

\item If $\mu > \mu_{\I}(U)$, then the system exhibits F order, i.e.,
$$\mathcal{F}_{\F}(U, \mu) < \mathcal{F}_{\AF}(U, \mu) < \mathcal{F}_{\mathrm{P}}(U, \mu).$$

\end{enumerate}

Moreover, it holds that:

\begin{itemize}
\item As $U \to +\infty$,
\begin{align}\label{muIexpansion}
\mu_{\I}(U) = &\; \mu_{\I}^{\app}(U) + O\bigg(\frac{1}{U^{9/2}}\bigg),
\end{align} 
where $\mu_{\I}^{\app}(U)$ is the function defined in (\ref{muIappdef}).

\item The doping of the AF state at any $(U, \mu) \in \I_{U_0, \delta}$ satisfies
$$d_0^{\AF} = -\frac{\partial \mathcal{F}_{\AF}}{\partial \mu}(U, \mu)  = 0;$$
in other words, the AF free energy is independent of $\mu$, i.e., $\mathcal{F}_{\AF}(U, \mu) = \mathcal{F}_{\AF}(U)$. 
\item The doping of the F state at $\mu = \mu_{\I}(U)$ is strictly positive,
$$d_0^{\F} = -\frac{\partial \mathcal{F}_{\F}}{\partial \mu}(U, \mu_{\I}(U)) 
> 0;$$
more precisely, the doping of the F state at $\mu = \mu_{\I}(U)$ satisfies
\begin{align}\label{d0FatmuIexpansion}
d_0^{\F} = \nu_{\I}^{\F}(U) + O\big(U^{-7/2}\big) \qquad \text{as $U \to +\infty$},
\end{align}
where $\nu_{\I}^{\F}(U)$ is the function defined in (\ref{nuIFdef}).

\item As $U \to +\infty$, the free energies satisfy
\begin{subequations}\label{freeenergiesatmuI}
\begin{align}\label{freeenergiesatmuIa}
& \mathcal{F}_{\AF}(U) = \mathcal{F}_{\F}(U, \mu_{\I}(U)) = -\frac{U}{4} - \frac{4}{U} + \frac{20}{U^3} + O\bigg(\frac{1}{U^5}\bigg),
	\\ \label{freeenergiesatmuIb}
& \mathcal{F}_{\mathrm{P}}(U, \mu_{\I}(U)) = 
-\frac{U}{4}
   +4
- \frac{4 \sqrt{2 \pi }}{\sqrt{U}}
+\frac{2 (\pi -96)}{3 U}
+ \sqrt{\frac{\pi }{2}}  \frac{(4608+5 \pi )}{36 U^{3/2}}
+ O\bigg(\frac{1}{U^2}\bigg).
\end{align}
\end{subequations}
\end{itemize}
\end{theorem}
\begin{proof}
See Section \ref{SectorIproofsec}.
\end{proof}

\begin{remark*}[Nagaoka's theorem]\upshape
Nagaoka's theorem \cite{N1966} states that the ground state of the Hubbard model with infinite repulsion $U=\infty$ doped one hole away from half-filling is ferromagnetic with maximum total spin.
Theorem \ref{sectorIth} implies the following Hartree--Fock version of Nagaoka's result: For any fixed arbitrarily small $\nu \neq 0$, there is a $U_0 >0$ such that in Hartree--Fock theory restricted to P, F, and AF states, the ground state with doping $\nu$ is ferromagnetic whenever $U \geq U_0$.
\end{remark*}

\subsection{Second theorem: Sector II}
Our second theorem shows that in Sector II, there is a curve $\mu_{\II}(U)$ such that the system is in the F phase for $\mu < \mu_{\II}(U)$, whereas it is in the P phase for $\mu > \mu_{\II}(U)$. It also proves the existence of a mixed phase between the F and P phases in the phase diagram, see Figure \ref{phasediagramfig}.

\begin{theorem}[The F-Mixed-P interface]\label{sectorIIth}
Let $M >0$. 
There exist a $U_0 < 4\pi$ and a smooth real-valued function $\mu_{\II}(U)$ of $U\in[U_0, 4\pi)$ such that if
$(U, \mu) \in \II_{U_0, M}$, then the following hold:
\begin{enumerate}[$(i)$]
\item The AF mean-field equations (\ref{AFmeanfield}) have no solution in $\R \times (0,+\infty)$, so $\mathcal{F}_{\AF}(U, \mu)$ is not defined. 

\item If $\mu < \mu_{\II}(U)$, then the system exhibits F order, i.e.,
$$\mathcal{F}_{\F}(U, \mu) < \mathcal{F}_{\mathrm{P}}(U, \mu).$$

\item If $\mu = \mu_{\II}(U)$, then the F and P states have the same free energy, i.e.,
$$\mathcal{F}_{\F}(U, \mu) = \mathcal{F}_{\mathrm{P}}(U, \mu).$$

\item If $\mu > \mu_{\II}(U)$, then the system exhibits P order, i.e.,
$$\mathcal{F}_{\mathrm{P}}(U, \mu) < \mathcal{F}_{\F}(U, \mu).$$
\end{enumerate}

Moreover, it holds that: 

\begin{itemize}
\item As $U \uparrow 4\pi$,
\begin{align}\label{muIIexpansion}
\mu_{\II}(U) = &\; \mu_{\II}^{\app}(U)  + O\big((4\pi - U)^7\big),
\end{align} 
where $\mu_{\II}^{\app}(U)$ is the function defined in (\ref{muIIappdef}).

\item As $U \uparrow 4\pi$, the doping $d_0^{\F}(U, \mu)$ of the F state at  $\mu = \mu_{\II}(U)$ satisfies 
\begin{align}\label{d0FatmuIIexpansion}
d_0^{\F} = \nu_{\II}^{\F}(U) + O\big((4\pi - U)^7\big),
\end{align}
while the doping $d_0^{\mathrm{P}}(U, \mu)$ of the P state at $\mu = \mu_{\II}(U)$ satisfies 
\begin{align}\label{d0PatmuIIexpansion}
d_0^{\mathrm{P}} = \nu_{\II}^{\mathrm{P}}(U) + O\big((4\pi - U)^7\big),
\end{align}
where $\nu_{\II}^{\F}(U)$ and $\nu_{\II}^{\mathrm{P}}(U)$ are the functions in (\ref{nuIIFdef}) and (\ref{nuIIPdef}).

\item As $U \uparrow 4\pi$, the P and F free energies at $\mu = \mu_{\II}(U)$ satisfy
\begin{align}\nonumber
 \mathcal{F}_{\F}(&U, \mu_{\II}(U)) = \mathcal{F}_{\mathrm{P}}(U, \mu_{\II}(U))
=  -4 - \pi + \bigg(\frac{1}{4} + \frac{4}{\pi}\bigg) (4\pi - U)
 -\frac{24+7 \pi }{12 \pi^3}(4\pi - U)^2 
 	\\ \nonumber
& + \frac{72-17 \pi }{288 \pi^4} (4\pi - U)^3 
 + \frac{8040-1861 \pi }{138240 \pi^5} (4\pi - U)^4
 + \frac{50952-15181 \pi }{3317760 \pi^6} (4\pi - U)^5
 	\\ \label{freeenergiesatmuII}
& + \frac{3734192-1409727 \pi }{743178240 \pi^7} (4\pi - U)^6
+ O\big((4\pi - U)^7\big).
\end{align}
\end{itemize}
\end{theorem}
\begin{proof}
See Section \ref{SectorIIproofsec}.
\end{proof}

\subsection{Third theorem: Sector III}
Our last theorem shows that in Sector III, there is a curve $\mu_{\III}(U)$ such that the system is in the AF phase for $\mu < \mu_{\III}(U)$ whereas it is in the P phase for $\mu > \mu_{\III}(U)$. It also proves that the AF state has zero doping (i.e., $d_0^{\AF} = 0$), whereas the P state has strictly positive doping (i.e., $d_0^{\mathrm{P}} > 0$) at $\mu = \mu_{\III}(U)$, thereby establishing the existence of a mixed phase between the AF and P phases in the phase diagram, see Figure \ref{phasediagramfig}.

\begin{theorem}[The AF-Mixed-P interface]\label{sectorIIIth}
Let $\delta >0$. There exist a $U_0 >0$ and a smooth real-valued function $\mu_{\III}(U)$ of $U \in (0,U_0]$ such that if $(U, \mu) \in \III_{U_0, \delta}$, then the following hold:
\begin{enumerate}[$(i)$]
\item The F mean-field equations (\ref{Fmeanfield}) have no solution in $\R \times (0,+\infty)$, so $\mathcal{F}_{\F}(U, \mu)$ is not defined. 

\item If $\mu < \mu_{\III}(U)$, then the system exhibits AF order, i.e.,
$$\mathcal{F}_{\AF}(U, \mu) < \mathcal{F}_{\mathrm{P}}(U, \mu).$$ 

\item If $\mu = \mu_{\III}(U)$, then the AF and P states have the same free energy, i.e.,
$$\mathcal{F}_{\AF}(U, \mu) = \mathcal{F}_{\mathrm{P}}(U, \mu).$$

\item If $\mu > \mu_{\III}(U)$, then the system exhibits P order, i.e.,
$$\mathcal{F}_{\mathrm{P}}(U, \mu) < \mathcal{F}_{\AF}(U, \mu).$$
\end{enumerate}

Moreover, it holds that:

\begin{itemize}
\item As $U \downarrow 0$,
\begin{align}\label{muIIIexpansion}
\mu_{\III}(U) = &\; \mu_{\III}^{\app}(U) + O\big(U^{5/2} e^{-\frac{2\pi}{\sqrt{U}}}\big),
\end{align} 
where $\mu_{\III}^{\app}(U)$ is the function defined in (\ref{muIIIappdef}).

\item The doping of the AF state at any $(U, \mu) \in \III_{U_0, \delta}$ satisfies
$$d_0^{\AF} = -\frac{\partial \mathcal{F}_{\AF}}{\partial \mu}(U, \mu)  = 0;$$
in other words, the AF free energy is independent of $\mu$, i.e., $\mathcal{F}_{\AF}(U, \mu) = \mathcal{F}_{\AF}(U)$.

\item As $U \downarrow 0$, the doping $d_0^{\mathrm{P}}(U, \mu)$ of the P state at $\mu = \mu_{\III}(U)$ satisfies 
\begin{align} \nonumber
d_0^{\mathrm{P}}(U, \mu_{\III}(U)) = &\; 
\bigg\{ \frac{32 \sqrt{2}}{\pi  \sqrt{U}}
-\frac{4 \sqrt{2} (2+\ln{2})}{\pi^2}
+\frac{ (22 - \ln{2}) (2+\ln{2})}{2 \sqrt{2} \pi^3}\sqrt{U}
	\\ \nonumber
& -\frac{ (10+\ln{2}) (44+\ln(2)(8+\ln{2}))}{16 \sqrt{2} \pi^4}U
	\\ \nonumber
& + \frac{43536+\ln (2)(16416+\ln (2)  (5208-\ln(2) (184+15 \ln{2})))}{1536 \sqrt{2} \pi^5}U^{3/2} 
   	\\\label{d0PatmuIIIexpansion}
& + O(U^{2})\bigg\} e^{-\frac{2\pi}{\sqrt{U}}}. 
\end{align}

\item As $U \downarrow 0$, the P and F free energies satisfy
\begin{align}\label{freeenergiesatmuIII}
 \mathcal{F}_{\AF}(U) = \mathcal{F}_{\mathrm{P}}(U, \mu_{\III}(U))
= -\frac{16}{\pi^2}
-\frac{512 e^{-\frac{4 \pi }{\sqrt{U}}}}{\pi \sqrt{U}}  -\frac{128 e^{-\frac{4 \pi }{\sqrt{U}}}}{\pi^2} 
+ O\bigg(\frac{e^{-\frac{8\pi}{\sqrt{U}}}}{U}\bigg).
\end{align}
\end{itemize}
\end{theorem}
\begin{proof}
See Section \ref{SectorIIIproofsec}.
\end{proof}

\begin{remark}\label{expansionsremark}\upshape
All expansions in Theorems \ref{sectorIth}--\ref{sectorIIIth} can easily be extended to higher order. However, we emphasize that the expansions in Theorems \ref{sectorIth}--\ref{sectorIIIth} are asymptotic expansions, and we do not know if they have positive radii of convergence. Thus, although the agreement with the numerical phase diagram of Figure \ref{phasediagramfig} is excellent even for values of $(\nu, U)$ far from the four boundaries $U = +\infty$, $U=0$, and $\nu = \pm 1$, there is no guarantee that the inclusion of more terms in the expansions will improve the accuracy far from the boundaries (we hope that such information will be obtained in future work). 
In our definitions (\ref{nuIFdef})--(\ref{nuIIPdef}) of $\nu_{\I}^{\F}(U)$, $\nu_{\II}^{\F}(U)$, and $\nu_{\II}^{\mathrm{P}}(U)$, we have included enough terms in the expansions that the curves appear to have converged from the point of view of Figure \ref{phasediagramfig}.
In our definition (\ref{nuIIIPdef}) of $\nu_{\III}^{\mathrm{P}}(U)$, we have included terms as follows: numerically, one finds that the series in (\ref{d0PatmuIIIexpansion}) decays slowly for $U \approx 6$. To obtain a good approximation for the limiting curve, one option is therefore to compute a large number of higher-order terms in the series. However, the alternating nature of the series (\ref{d0PatmuIIIexpansion}) suggests that a good approximation to the limiting curve can be obtained with only a small number of terms by multiplying the last term in the truncated series by $1/2$. This is indeed what is found numerically, and we have therefore only included terms up to order $\sqrt{U}$ in our definition (\ref{nuIIIPdef}) of $\nu_{\III}^{\mathrm{P}}(U)$, but the term of order $\sqrt{U}$ has an extra factor of $1/2$ compared to (\ref{d0PatmuIIIexpansion}).
\end{remark}

\section{The paramagnetic Hartree--Fock function}\label{Psec}
It follows from the definition (\ref{PHartreeFockFunction}) of the P Hartree--Fock function $\mathcal{G}_{\mathrm{P}}(d_0, U, \mu)$ that
\begin{align}\label{dGPdd0}
\frac{\partial \mathcal{G}_{\mathrm{P}}}{\partial d_0} 
= &- \frac{U}{2}\bigg\{d_0 - 1
+ 2\int_{\R} N_0(\epsilon) \theta\bigg(\frac{U}{2} d_0 - \mu + \epsilon\bigg) d\epsilon\bigg\}
= -\frac{U}{2}\bigg(\frac{\partial \mathcal{G}_{\mathrm{P}}}{\partial \mu} + d_0\bigg).
\end{align}
The density of states $N_0(\epsilon)$ is a smooth function of $\epsilon \in \R \setminus \{0, \pm 4\}$ with a logarithmic singularity at $\epsilon = 0$ and with jump discontinuities at $\epsilon = \pm 4$, see Figure \ref{N0fig}. From (\ref{dGPdd0}), we therefore deduce the following lemma. 

\begin{lemma}\label{GPlemma}
For any fixed $U > 0$, $\mathcal{G}_{\mathrm{P}}(d_0, U, \mu)$ is a $C^1$-function of $(d_0, \mu) \in \R^2$, which is smooth for $(d_0, \mu) \in \R^2 \setminus S_{\mathrm{P}}(U)$, where
$$S_{\mathrm{P}}(U) := \Big\{(d_0, \mu) \in \R^2 \,\Big|\, \tfrac{U d_0}{2} - \mu \in \{0, \pm 4\} \Big\}.$$
For $(d_0, \mu) \in \R^2 \setminus S_{\mathrm{P}}(U)$, we have
\begin{align}\label{d2GPdd02}
\frac{\partial^2 \mathcal{G}_{\mathrm{P}}}{\partial \mu \partial d_0} 
= -\frac{U}{2} \frac{\partial^2 \mathcal{G}_{\mathrm{P}}}{\partial \mu^2} 
= - \frac{2}{U} \frac{\partial^2 \mathcal{G}_{\mathrm{P}}}{\partial d_0^2} - 1
= U N_0\bigg(\frac{U}{2} d_0 - \mu\bigg).
\end{align}
\end{lemma}

\subsection{P mean-field equation}
According to Lemma \ref{GPlemma}, $\mathcal{G}_{\mathrm{P}}(d_0, U, \mu)$ is a $C^1$-function of $d_0 \in \R$. Thus, if $d_0^{\mathrm{P}} = d_0^{\mathrm{P}}(U, \mu)$ is an extremizer of the right-hand side of (\ref{freeenergiesP}), i.e., if 
$$\mathcal{F}_{\mathrm{P}}(U, \mu) = \mathcal{G}_{\mathrm{P}}(d_0^{\mathrm{P}}, U, \mu),$$
then $d_0^{\mathrm{P}}$ solves the P mean-field equation $\frac{\partial \mathcal{G}_{\mathrm{P}}}{\partial d_0}(d_0^{\mathrm{P}}, U, \mu) = 0$. Using the explicit expression (\ref{dGPdd0}) for $\frac{\partial \mathcal{G}_{\mathrm{P}}}{\partial d_0}$, we can write the P mean-field equation as
\begin{align}\label{Pmeanfieldeq}
d_0^{\mathrm{P}} = 1 - 2\int_{\R} N_0(\epsilon) \theta\bigg(\frac{U}{2} d_0^{\mathrm{P}} - \mu + \epsilon\bigg) d\epsilon.
\end{align}

\begin{lemma}\label{Plemma1}
For any choice of $U > 0$ and $\mu \in \R$, the following hold:
\begin{enumerate}[$(i)$]
\item The P mean-field equation (\ref{Pmeanfieldeq}) has a unique solution $d_0^{\mathrm{P}} \in \R$. Moreover, this solution satisfies $d_0^{\mathrm{P}} \in [-1,1]$. 

\item The P Hartree--Fock function $\mathcal{G}_{\mathrm{P}}(d_0, U, \mu)$ is a strictly concave function of $d_0 \in \R$ that attains its maximum at $d_0 = d_0^{\mathrm{P}}$.
\end{enumerate}
\end{lemma}
\begin{proof}
Fix $U > 0$ and $\mu \in \R$. We know from Lemma \ref{GPlemma} that $\frac{\partial \mathcal{G}_{\mathrm{P}}}{\partial d_0}(d_0, U, \mu)$ is a continuous function of $d_0 \in \R$, and that, except at the three isolated values of $d_0$ at which $\frac{Ud_0}{2} \in \R \setminus \{\mu, \mu \pm 4\}$, we have
$$\frac{\partial^2 \mathcal{G}_{\mathrm{P}}}{\partial d_0^2}(d_0, U, \mu) = - \frac{U}{2}\bigg\{1 + U N_0\bigg(\frac{U}{2} d_0 - \mu\bigg)\bigg\} \leq - \frac{U}{2}.$$
It follows that $\frac{\partial \mathcal{G}_{\mathrm{P}}}{\partial d_0}(d_0, U, \mu)$ is a strictly decreasing function of $d_0 \in \R$. Therefore, the P mean-field equation $\frac{\partial \mathcal{G}_{\mathrm{P}}}{\partial d_0}(d_0^{\mathrm{P}}, U, \mu) = 0$ has a unique solution $d_0^{\mathrm{P}} \in \R$, and $\mathcal{G}_{\mathrm{P}}(d_0, U, \mu)$ is a strictly concave function of $d_0 \in \R$ that attains its maximum at $d_0 = d_0^{\mathrm{P}}$. 
The fact that $d_0^{\mathrm{P}} \in [-1,1]$ is a consequence of (\ref{Pmeanfieldeq}) (to see this, recall that $N_0(\epsilon)\geq 0$ and $\int_{\R} N_0(\epsilon)d\epsilon=1$).
\end{proof}

\subsection{Properties of the P mean-field solution}
By Lemma \ref{Plemma1}, the solution $d_0^{\mathrm{P}}$ of the P mean-field equation always exists and is unique. The next lemma 
 analyzes how $d_0^{\mathrm{P}}$ depends on $U$ and $\mu$.

\begin{lemma}\label{Plemma2}
For any $U > 0$, the unique solution $d_0^{\mathrm{P}} = d_0^{\mathrm{P}}(U, \mu)$ of the P mean-field equation (\ref{Pmeanfieldeq}) is an odd continuous function of $\mu \in \R$, which is $C^1$ for $\mu \in \R \setminus \{\pm (4 + \tfrac{U}{2})\}$, smooth for $\mu \in \R \setminus \{0, \pm (4 + \tfrac{U}{2})\}$, and such that 
\begin{enumerate}[$(a)$]
\item $d_0^{\mathrm{P}}(U, \mu) = -1$ for all $\mu \leq -4 - \tfrac{U}{2}$,

\item \label{d0Pitemb}
$d_0^{\mathrm{P}}(U, \mu)$ is strictly increasing from $-1$ to $1$ as $\mu$ increases from $-4 - \tfrac{U}{2}$ to $4 + \tfrac{U}{2}$,

\item $d_0^{\mathrm{P}}(U, \mu) = 1$ for all $\mu \geq 4 + \tfrac{U}{2}$,

\item $\frac{\partial d_0^{\mathrm{P}}}{\partial \mu} (U, 0) 
= \frac{2}{U}$,

\item $|\frac{\partial^2 d_0^{\mathrm{P}}}{\partial \mu^2} (U, \mu) |
\to +\infty$ as $\mu \to 0$,

\item $\mu \mapsto \frac{\partial d_0^{\mathrm{P}}}{\partial \mu} (U, \mu)$ has a jump discontinuity at $\mu = 4 + \tfrac{U}{2}$; more precisely, $\lim_{\epsilon \downarrow 0} \frac{\partial d_0^{\mathrm{P}}}{\partial \mu} (U, 4 + \tfrac{U}{2} - \epsilon) 
= \frac{2}{4\pi + U}$
and
$\lim_{\epsilon \downarrow 0} \frac{\partial d_0^{\mathrm{P}}}{\partial \mu} (U, 4 + \tfrac{U}{2} + \epsilon) 
= 0$.
\end{enumerate}

\end{lemma}
\begin{proof}
Fix $U > 0$. In this proof, we suppress the $U$-dependence of $d_0^{\mathrm{P}}(U, \mu)$ and $\mathcal{G}_{\mathrm{P}}(d_0, U, \mu)$ for brevity.
By the implicit function theorem, $d_0^{\mathrm{P}}(\mu)$ is a smooth function of $\mu$ as long as $(d_0^{\mathrm{P}}, \mu) \notin S_{\mathrm{P}}(U)$. If $(d_0^{\mathrm{P}}, \mu) \in S_{\mathrm{P}}(U)$, then we have three cases:
\begin{enumerate}[$(1)$]
\item\label{item1} If  $\tfrac{U d_0^{\mathrm{P}}}{2} - \mu = 4$, then (\ref{Pmeanfieldeq}) gives $d_0^{\mathrm{P}} = -1$, so $\mu = -4 - \tfrac{U}{2}$. 

\item\label{item2}  If  $\tfrac{U d_0}{2} - \mu = 0$,  then (\ref{Pmeanfieldeq}) gives $d_0^{\mathrm{P}} = 0$, so $\mu = 0$. 

\item\label{item3} If  $\tfrac{U d_0}{2} - \mu = -4$,  then (\ref{Pmeanfieldeq}) gives $d_0^{\mathrm{P}} = 1$, so $\mu = 4 + \tfrac{U}{2}$. 
\end{enumerate}
We conclude that $d_0^{\mathrm{P}}$  is a smooth function of $\mu \in \R \setminus \{0, \pm (4 + \tfrac{U}{2})\}$.
Furthermore, differentiation with respect to $\mu$ of the identity
$$\frac{\partial \mathcal{G}_{\mathrm{P}}}{\partial d_0}(d_0^{\mathrm{P}}(\mu), \mu) = 0 \qquad \text{for $\mu \in \R \setminus \{0, \pm (4 + \tfrac{U}{2})\}$}$$
gives
$$\frac{\partial^2 \mathcal{G}_{\mathrm{P}}}{\partial d_0^2}(d_0^{\mathrm{P}}(\mu), \mu)
(d_0^{\mathrm{P}})'(\mu)
+ \frac{\partial^2 \mathcal{G}_{\mathrm{P}}}{\partial d_0 \partial \mu} (\mu) = 0, $$
i.e., using also (\ref{d2GPdd02}),
\begin{align}\label{d0Pprime}
(d_0^{\mathrm{P}})'(\mu)
= \frac{2N_0(X)}{1 + U N_0(X)}
\qquad \text{for $\mu \in \R \setminus \{0, \pm (4 + \tfrac{U}{2})\}$},
\end{align}
where we have introduced the short-hand notation $X := \mu - \frac{U}{2} d_0^{\mathrm{P}}$. Equation (\ref{d0Pprime}) implies that $(d_0^{\mathrm{P}})'(\mu) > 0$ whenever $X \in (-4,4)$. Together with (\ref{item1})--(\ref{item3}), this proves assertions $(a)$--$(c)$. We also see from (\ref{Pmeanfieldeq}) that $X \to 0$ as $\mu \to 0$. Hence (\ref{d0Pprime}) implies that $d_0^{\mathrm{P}}$ is $C^1$  at $\mu = 0$ and that assertion $(d)$ holds. 
Differentiation of (\ref{d0Pprime}) gives
$$(d_0^{\mathrm{P}})''(\mu)
= \frac{2 N_0'(X)}{(1 + U N_0(X))^{3}}
\qquad \text{for $\mu \in \R \setminus \{0, \pm (4 + \tfrac{U}{2})\}$},$$
from which assertion $(e)$ follows.
If $\mu \uparrow 4 + \tfrac{U}{2}$, then $X \downarrow 4$, and if $\mu \downarrow 4 + \tfrac{U}{2}$, then $X \uparrow 4$.
Since $N_0(4^-) = \frac{1}{4\pi}$ by (\ref{N0near4}) and $N_0(4^+) = 0$, assertion $(f)$ follows from (\ref{d0Pprime}).
\end{proof}

\subsection{P free energy}
By combining the above lemmas, we arrive at the following description of the P free energy $\mathcal{F}_{\mathrm{P}}$.

\begin{lemma}\label{Plemma}
For any $U >0$ and $\mu \in \R$, the P free energy defined in (\ref{freeenergiesP}) is given by
\begin{align}\label{FPGP}
\mathcal{F}_{\mathrm{P}}(U, \mu) = \mathcal{G}_{\mathrm{P}}(d_0^{\mathrm{P}}(U, \mu), U, \mu),
\end{align}
where $d_0^{\mathrm{P}} = d_0^{\mathrm{P}}(U, \mu)$ is the unique solution of the P mean-field equation (\ref{Pmeanfieldeq}).
Moreover, $\mathcal{F}_{\mathrm{P}}(U, \mu)$ is an even and concave $C^1$-function of $\mu \in \R$, which is smooth for $\mu \notin \{0, \pm (4 + \tfrac{U}{2})\}$; its maximum value is
$$\mathcal{F}_{\mathrm{P}}(U, 0) = - \frac{16}{\pi^2}  \approx -1.62114,$$
and it satisfies
\begin{align}\label{dFPdmuequalsd0P}
\frac{\partial \mathcal{F}_{\mathrm{P}}}{\partial \mu} (U, \mu) = -d_0^{\mathrm{P}}(U, \mu) \qquad \text{for $U >0$ and $\mu \in \R$}.
\end{align}
\end{lemma}
\begin{proof}
From Lemma \ref{Plemma1}, we have $\mathcal{F}_{\mathrm{P}}(U, \mu) = \mathcal{G}_{\mathrm{P}}(d_0^{\mathrm{P}}(U, \mu), U, \mu)$ and so, for $\mu \notin \{0, \pm (4 + \tfrac{U}{2})\}$,
\begin{align}\nonumber
\frac{\partial \mathcal{F}_{\mathrm{P}}}{\partial \mu} (U, \mu) 
& = \frac{\partial \mathcal{G}_{\mathrm{P}}}{\partial d_0}(d_0^{\mathrm{P}}(U, \mu), U, \mu)
\frac{\partial d_0^{\mathrm{P}}}{\partial \mu} (U, \mu)
+ \frac{\partial \mathcal{G}_{\mathrm{P}}}{\partial \mu}(d_0^{\mathrm{P}}(U, \mu), U, \mu)
	\\ \nonumber
& = \frac{\partial \mathcal{G}_{\mathrm{P}}}{\partial \mu}(d_0^{\mathrm{P}}(U, \mu), U, \mu)
= -d_0^{\mathrm{P}}(U, \mu)
\end{align}
where we used (\ref{dGPdd0}) and the fact that $\frac{\partial \mathcal{G}_{\mathrm{P}}}{\partial d_0}(d_0^{\mathrm{P}}(U, \mu), U, \mu) = 0$. 
This proves (\ref{dFPdmuequalsd0P}) for $\mu \notin \{0, \pm (4 + \tfrac{U}{2})\}$. By Lemmas \ref{GPlemma} and \ref{Plemma2}, both sides of (\ref{dFPdmuequalsd0P}) are continuous functions of $\mu \in \R$, so (\ref{dFPdmuequalsd0P}) in fact holds for all $\mu \in \R$. Lemma \ref{Plemma2} also shows that $\frac{\partial \mathcal{F}_{\mathrm{P}}}{\partial \mu} (U, \mu)$ is an odd and decreasing function of $\mu$, which is smooth for $\mu \notin \{0, \pm (4 + \tfrac{U}{2})\}$. Hence $\mu \mapsto \mathcal{F}_{\mathrm{P}}(U, \mu)$ is an even and concave $C^1$-function of $\mu \in \R$, which is smooth for $\mu \notin \{0, \pm (4 + \tfrac{U}{2})\}$. If $\mu = 0$, then (\ref{Pmeanfieldeq}) has the unique solution $d_0^{\mathrm{P}} = 0$, and so (\ref{PHartreeFockFunction}) implies that
$$\mathcal{F}_{\mathrm{P}}(U, 0) = \mathcal{G}_{\mathrm{P}}(0, U, 0)
=  - 2\int_0^4 N_0(\epsilon) \epsilon d\epsilon= - \frac{16}{\pi^2},$$
where we have used (\ref{N0epsilon8pi2}) in the last step.
The proof is complete.
\end{proof}

\section{The ferromagnetic Hartree--Fock function}\label{Fsec}
It follows from the definition (\ref{FHartreeFockFunction}) of the F Hartree--Fock function $\mathcal{G}_{\F}(d_0, m_0, U, \mu)$ that
\begin{align}\label{dGFdd0}
\frac{\partial \mathcal{G}_{\F}}{\partial d_0} 
= &- \frac{U}{2} \bigg\{d_0  - 1 
+ \sum_{r=\pm 1}  \int_{\R} N_0(\epsilon) \theta\bigg(\frac{U(d_0 - r m_0)}{2}  - \mu + \epsilon\bigg)d\epsilon
 \bigg\}
= -\frac{U}{2}\bigg(\frac{\partial \mathcal{G}_{\F}}{\partial \mu} + d_0\bigg),
	\\  \label{dGFdm0}
\frac{\partial \mathcal{G}_{\F}}{\partial m_0} 
= &\; \frac{U}{2}\bigg\{ m_0
 + \sum_{r=\pm 1} r\int_{\R} N_0(\epsilon)  \theta\bigg(\frac{U(d_0 - r m_0)}{2}  - \mu + \epsilon\bigg)d\epsilon\bigg\}.
\end{align}
Since $\epsilon \mapsto N_0(\epsilon)$ is smooth for $\epsilon \in \R \setminus \{0, \pm 4\}$, we conclude that, for any fixed $U > 0$, the function $\mathcal{G}_{\F}(d_0, m_0, U, \mu)$ is a $C^1$-function of $(d_0, m_0, \mu) \in \R \times [0,+\infty) \times \R$. In particular, the F mean-field equations (\ref{Fmeanfield}) are well-defined for all $(d_0, m_0, \mu) \in \R \times [0,+\infty) \times \R$.
Moreover, differentiation of (\ref{dGFdd0}) gives
\begin{align}\label{d2GFdd02}
\frac{\partial^2 \mathcal{G}_{\F}}{\partial d_0^2} 
= &- \frac{U}{2} \bigg\{1  
+ \frac{U}{2} \sum_{r=\pm 1} \int_{\R} N_0(\epsilon) \delta\bigg(\frac{U(d_0 - r m_0)}{2}  - \mu + \epsilon\bigg)d\epsilon \bigg\}
 < 0,
 \end{align}
implying that $\mathcal{G}_{\F}(d_0, m_0, U, \mu)$ is a strictly concave function of $d_0$.

\subsection{F mean-field equations}
Using the explicit expressions (\ref{dGFdd0}) and (\ref{dGFdm0}) for $\frac{\partial \mathcal{G}_{\F}}{\partial d_0}$ and $\frac{\partial \mathcal{G}_{\F}}{\partial m_0}$, respectively, the F mean-field equations (\ref{Fmeanfield}) can be written as
\begin{align}\label{Fmeanfield2}
\begin{cases}
 d_0^{\F} = 1 
- \int_{\R} N_0(\epsilon) \theta\big(\frac{U(d_0^{\F} - m_0^{\F})}{2}  - \mu + \epsilon\big) d\epsilon
- \int_{\R} N_0(\epsilon) \theta\big(\frac{U(d_0^{\F} + m_0^{\F})}{2}  - \mu + \epsilon \big) d\epsilon,
	\\
m_0^{\F}
 = - \int_{\R} N_0(\epsilon) \theta\big(\frac{U(d_0^{\F} - m_0^{\F})}{2}  - \mu + \epsilon\big) d\epsilon
 + \int_{\R} N_0(\epsilon) \theta\big(\frac{U(d_0^{\F} + m_0^{\F})}{2}  - \mu + \epsilon \big) d\epsilon.
\end{cases}
\end{align}
Since $N_0(\epsilon)\geq 0$ and $\int_{\R} N_0(\epsilon) d\epsilon=1$, it is immediate from these equations that any F mean-field solution must satisfy
\begin{align}\label{d0Fm0Fbound}
d_0^{\F} \in [-1,1] \quad \text{and} \quad  m_0^{\F} \in [0,1].
\end{align}
Furthermore, adding and subtracting the equations in (\ref{Fmeanfield2}), we can write the F mean-field equations as
\begin{align}\label{Fmeanfieldeqs}
\begin{cases}
 d_0^{\F} + m_0^{\F} = 1 
- 2\int_{\R} N_0(\epsilon) \theta\big(\frac{U(d_0^{\F} - m_0^{\F})}{2}  - \mu + \epsilon\big) d\epsilon,
	\\
 d_0^{\F} - m_0^{\F}
 = 1 - 2 \int_{\R} N_0(\epsilon) \theta\big(\frac{U(d_0^{\F} + m_0^{\F})}{2}  - \mu + \epsilon \big) d\epsilon.
\end{cases}
\end{align}

\section{The antiferromagnetic Hartree--Fock function}\label{AFsec}
The definition (\ref{AFHartreeFockFunction}) of the AF Hartree--Fock function $\mathcal{G}_{\AF}(d_0, m_1, U, \mu)$ implies that
\begin{align}\nonumber
\frac{\partial \mathcal{G}_{\AF}}{\partial d_0} 
= & -\frac{U}{2} \bigg\{d_0
- 1 
 + \sum_{r=\pm 1}  \int_{\R} N_0(\epsilon) \theta\bigg(\frac{U}{2} d_0 - \mu + r \sqrt{\frac{U^2}{4}m_1^2 + \epsilon^2}\bigg) d\epsilon \bigg\}
 	\\ \label{dGAFdd0}
=& -\frac{U}{2}\bigg(\frac{\partial \mathcal{G}_{\AF}}{\partial \mu} + d_0\bigg)
 \end{align}
and
\begin{align}\label{dGAFdm1}	
\frac{\partial \mathcal{G}_{\AF}}{\partial m_1} 
= \frac{U}{2} m_1
 - \sum_{r=\pm 1} r \int_{\R} N_0(\epsilon) \frac{ \frac{U^2}{4}m_1}{\sqrt{\frac{U^2}{4}m_1^2 + \epsilon^2}}  \theta\bigg(\frac{U}{2} d_0 - \mu + r \sqrt{\frac{U^2}{4}m_1^2 + \epsilon^2}\bigg) d\epsilon.
\end{align}
We infer that, for any fixed $U > 0$, the function $\mathcal{G}_{\AF}(d_0, m_1, U, \mu)$ is a $C^1$-function of $(d_0, m_1, \mu) \in \R \times [0,+\infty) \times \R$.
In particular, the AF mean-field equations (\ref{AFmeanfield}) are well-defined for all $(d_0, m_1, \mu) \in \R \times [0,+\infty) \times \R$. Moreover, differentiation of (\ref{dGAFdd0}) gives
\begin{align}\label{d2GAFdd02}
\frac{\partial^2 \mathcal{G}_{\AF}}{\partial d_0^2} 
= & -\frac{U}{2} \bigg\{1
 + \frac{U}{2} \sum_{r=\pm 1} \int_{\R} N_0(\epsilon) \delta\bigg(\frac{U}{2} d_0 - \mu + r \sqrt{\frac{U^2}{4}m_1^2 + \epsilon^2}\bigg) d\epsilon \bigg\} < 0
 \end{align}
implying that $\mathcal{G}_{\F}(d_0, m_0, U, \mu)$ is a strictly concave function of $d_0$.

\subsection{AF mean-field equations}
Using the explicit expressions (\ref{dGAFdd0}) and (\ref{dGAFdm1}) for $\frac{\partial \mathcal{G}_{\AF}}{\partial d_0}$ and $\frac{\partial \mathcal{G}_{\AF}}{\partial m_1}$, respectively, we can write the AF mean-field equations as
\begin{align}\label{AFmeanfieldeqs}
\begin{cases}
d_0^{\AF} = 1 
- \int_{\R} N_0(\epsilon) \Big\{\theta\Big(\frac{U}{2} d_0^{\AF} - \mu + \sqrt{\frac{U^2}{4}(m_1^{\AF})^2 + \epsilon^2}\Big) 
	\\
\hspace{3.8cm} + \; \theta\Big(\frac{U}{2} d_0^{\AF} - \mu - \sqrt{\frac{U^2}{4}(m_1^{\AF})^2 + \epsilon^2}\Big)\Big\} d\epsilon,
\vspace{.1cm}	\\ 
m_1^{\AF}
 = \int_{\R}  N_0(\epsilon) \frac{\frac{U}{2}m_1^{\AF}}{\sqrt{\frac{U^2}{4}(m_1^{\AF})^2 + \epsilon^2}} \Big\{\theta\Big(\frac{U}{2} d_0^{\AF} - \mu + \sqrt{\frac{U^2}{4}(m_1^{\AF})^2 + \epsilon^2}\Big) 
	\\
\hspace{6cm}  - \; \theta\Big(\frac{U}{2} d_0^{\AF} - \mu - \sqrt{\frac{U^2}{4}(m_1^{\AF})^2 + \epsilon^2}\Big)\Big\} d\epsilon.
\end{cases}
\end{align}
Since $N_0(\epsilon)\geq 0$ and $\int_{\R} N_0(\epsilon) d\epsilon=1$, it is easy to see from these equations that any AF mean-field solution must satisfy
\begin{align}\label{d0AFbound}
d_0^{\AF} \in [-1,1] \quad \text{and} \quad  m_1^{\AF} \in [0,1].
\end{align}

\section{Proof of Theorem \ref{sectorIth}}\label{SectorIproofsec}
In this section, we analyze the P, F, and AF free energies in Sector I and provide a proof of Theorem \ref{sectorIth}.
From (\ref{sectorIdef}), we recall that Sector I, denoted $\I_{U_0, \delta}$, consists of all $(U, \mu) \in \R^2$ with $U \geq U_0$ and $0 \leq \mu \leq \tfrac{U}{2} - \delta$.
We assume that $\delta > 0$ is a fixed small number, and that $U_0 > 0$.

\subsection{P free energy in Sector I}
Given $0 < \delta < \alpha$, let $\I_{U_0, \alpha, \delta}$ be the subsector of $\I_{U_0, \delta}$ in which $\mu \geq \tfrac{U}{2} - \alpha$, i.e.,
$$\I_{U_0, \alpha, \delta} = \Big\{(U, \mu) \in \R^2 \, \Big| \, U \geq U_0 \;\text{and}\; \mu \in [\tfrac{U}{2} - \alpha, \tfrac{U}{2} - \delta] \Big\}.$$
The next lemma determines the asymptotics of the P free energy in $\I_{U_0, \alpha, \delta}$.

\begin{lemma}[Asymptotics of $\mathcal{F}_{\mathrm{P}}$ in $\I_{U_0, \alpha, \delta}$]\label{PIlemma}
Let $0 < \delta < \alpha$.
As $(U, \mu) \in \I_{U_0, \alpha, \delta}$ tends to infinity, we have
\begin{align}\label{FPexpansionSectorIhat}
\mathcal{F}_{\mathrm{P}}(U, \mu) = &\; -\frac{U}{4} - 4
+\hat{\mu} -\frac{\hat{\mu}^2}{U}
+\frac{4 \pi  \hat{\mu}^2}{U^2}
 -\frac{2 \pi^2 (\hat{\mu}+24) \hat{\mu}^2}{3 U^3}
+ O\bigg(\frac{1}{U^4}\bigg),
\end{align}
where the error term is uniform with respect to $\mu$, and we have expressed the formula in terms of the variable $\hat{\mu} := \frac{U}{2} - \mu+4 \in [4 + \delta, 4 + \alpha]$ which is of order $O(1)$.
\end{lemma}
\begin{proof}
According to Lemma \ref{Plemma}, the P free energy is given by $\mathcal{F}_{\mathrm{P}}(U, \mu) = \mathcal{G}_{\mathrm{P}}(d_0^{\mathrm{P}}, U, \mu)$, where $d_0^{\mathrm{P}} = d_0^{\mathrm{P}}(U, \mu)$ is the unique solution of the P mean-field equation (\ref{Pmeanfieldeq}).

Assume first that $(U, \mu) \in \I_{U_0, \delta}$. 
Then $\mu \in [0, \frac{U}{2} - \delta]$, so Lemma \ref{Plemma2} (\ref{d0Pitemb}) implies that $d_0^{\mathrm{P}} \in (-1,1)$. Thus, it follows from (\ref{Pmeanfieldeq}) that $\mu - \frac{U}{2} d_0^{\mathrm{P}} \in (-4,4)$ (since $\int_{\R} N_{0}(\epsilon) d\epsilon=1$, $N_{0}(\epsilon) \geq 0$ for all $\epsilon \in \R$ and $N_{0}(\epsilon)=0$ for $|\epsilon|>4)$.
In particular, (\ref{Pmeanfieldeq}) can be written as
\begin{align}\label{1minusd0PSectorI}
1 - d_0^{\mathrm{P}} = 2\int_{\mu - \frac{U}{2}  d_0^{\mathrm{P}}}^4 N_0(\epsilon) d\epsilon,
\end{align}
and, recalling (\ref{PHartreeFockFunction}), the P free energy can be written as
\begin{align}\label{FPSectorI}
\mathcal{F}_{P}(U, \mu) 
= -  \frac{U}{4} (1-d_0^{\mathrm{P}})^2 + \frac{U}{4} - \mu - 2\int_{\mu - \frac{U}{2} d_0^{\mathrm{P}}}^4  N_0(\epsilon)\bigg(\epsilon  - \bigg(\mu - \frac{U}{2} d_0^{\mathrm{P}}\bigg)\bigg) d\epsilon.
\end{align}

Assume now that $(U, \mu) \in \I_{U_0, \alpha, \delta}$. The condition $\mu - \frac{U}{2} d_0^{\mathrm{P}} \in (-4,4)$ then implies that
$$\delta -4 < \tfrac{U}{2}(1 - d_0^{\mathrm{P}}) < 4 + \alpha;$$
in particular, $1 - d_0^{\mathrm{P}}(U, \mu) = O(1/U)$ as $(U, \mu) \in \I_{U_0, \alpha, \delta}$ tends to infinity.
Since $1 - d_0^{\mathrm{P}}(U, \mu) > 0$ by Lemma \ref{Plemma2}, it follows from (\ref{1minusd0PSectorI}) that $4 - (\mu - \frac{U}{2} d_0^{\mathrm{P}})$ is strictly positive and $O(1/U)$ as $(U, \mu) \to \infty$ in $\I_{U_0, \alpha, \delta}$. 
Therefore we can use the expansion (\ref{N0near4}) of $N_0$ as $\epsilon \uparrow 4$ in (\ref{1minusd0PSectorI}) to find (note that $4 - (\mu - \frac{U}{2} d_0^{\mathrm{P}})=\hat{\mu} - \frac{U}{2}(1-d_0^{\mathrm{P}})$)
\begin{align}\nonumber
1 - d_0^{\mathrm{P}} & = 2\int_{\mu - \frac{U}{2} d_0^{\mathrm{P}}}^4 \bigg(N_0^{(0)} + N_0^{(1)} (4 - \epsilon) + N_0^{(2)} (4 - \epsilon)^2 + O((4-\epsilon)^{3})\bigg) d\epsilon
	\\ \label{oneminusd0PsectorI}
& = \frac{\hat{\mu} - \frac{U}{2}(1-d_0^{\mathrm{P}})}{2\pi} 
+ \frac{\big(\hat{\mu} - \frac{U}{2}(1-d_0^{\mathrm{P}})\big)^2}{32\pi} 
+ \frac{5\big(\hat{\mu} - \frac{U}{2}(1-d_0^{\mathrm{P}})\big)^3}{1536\pi} 
+ O\bigg(\frac{1}{U^{4}}\bigg).
\end{align}
Solving for $1-d_0^{\mathrm{P}}$, we obtain
\begin{align}\label{1minusd0Pexpansion}
1-d_0^{\mathrm{P}} = \frac{2 \hat{\mu}}{U}
-\frac{8 \pi \hat{\mu}}{U^2} + \frac{2 \pi^2\hat{\mu}(16+\hat{\mu})}{U^3} +  O\bigg(\frac{1}{U^4}\bigg)
\end{align}
as $(U, \mu) \in \I_{U_0, \alpha, \delta}$ tends to infinity.

Similarly, substituting the expansion (\ref{N0near4}) of $N_0$ into (\ref{FPSectorI}), we find
\begin{align*}
\mathcal{F}_{P}(U, \mu) 
= & -  \frac{U}{4} (1-d_0^{\mathrm{P}})^2 + \frac{U}{4}  - \mu 
- \frac{(4 - (\mu - \frac{U}{2} d_0^{\mathrm{P}}))^2}{4\pi} 
- \frac{(4 - (\mu - \frac{U}{2} d_0^{\mathrm{P}}))^3}{96\pi} 
+ O\bigg(\frac{1}{U^4}\bigg)
\end{align*}
as $(U, \mu) \in \I_{U_0, \alpha, \delta}$ tends to infinity.
Employing (\ref{1minusd0Pexpansion}) to eliminate $d_0^{\mathrm{P}}$ from this expression, we arrive at (\ref{FPexpansionSectorIhat}).
\end{proof}

Choosing $\hat{\mu} = 4+\delta$ in Lemma \ref{PIlemma}, we see that $\mathcal{F}_{\mathrm{P}}(U, \frac{U}{2} - \delta) > -U/4$ for all large enough $U > 0$. Moreover, by Lemma \ref{Plemma}, $\mathcal{F}_{\mathrm{P}}(U, \mu)$ is a decreasing function of $\mu \geq 0$.
Thus, Lemma \ref{PIlemma} immediately yields the following lower bound for $\mathcal{F}_{\mathrm{P}}$ in all of Sector I.

\begin{lemma}[Lower bound for $\mathcal{F}_{\mathrm{P}}$ in Sector I]\label{FPboundIlemma}
If $U_0 > 0$ is large enough, then
\begin{align}
\mathcal{F}_{\mathrm{P}}(U, \mu) > -\frac{U}{4} \qquad \text{for all $(U, \mu) \in \I_{U_0, \delta}$}.
\end{align}
\end{lemma}

\subsection{F free energy in Sector I}
We now turn to the F free energy in Sector I.

\begin{lemma}[F mean-field solution in Sector I]\label{FsolutionIlemma}
Let $U_0 > 16$. If $(U, \mu) \in \I_{U_0, \delta}$, then there is a unique solution $(d_0^{\F}, m_0^{\F})$ of the F mean-field equations (\ref{Fmeanfieldeqs}) satisfying $m_0^{\F} > 0$. This solution is given by
\begin{align}\label{d0Fm0FSectorI}
(d_0^{\F}, m_0^{\F}) = \bigg(1 - \int_{\R}  N_0(\epsilon) \theta\Big(\frac{U}{2}  - \mu + \epsilon \Big) d\epsilon, \int_{\R}  N_0(\epsilon) \theta\Big(\frac{U}{2}  - \mu + \epsilon \Big) d\epsilon\bigg)
\end{align}
and satisfies
\begin{align}\label{Fcase2SectorI}
4 + \frac{U}{2}(d_0^{\F} - m_0^{\F}) \leq \mu.
\end{align}
\end{lemma}
\begin{proof}
Suppose $(U, \mu) \in \I_{U_0, \delta}$ with $U_0 > 16$. 
Any solution $(d_0^{\F}, m_0^{\F})$ of the F mean-field equations (\ref{Fmeanfieldeqs}) must satisfy either $0 \leq \mu < 4 + \frac{U}{2}(d_0^{\F} - m_0^{\F})$ or $4 + \frac{U}{2}(d_0^{\F} - m_0^{\F}) \leq \mu$. We consider the two cases in turn. 

Case 1. $0 \leq \mu < 4 + \frac{U}{2}(d_0^{\F} - m_0^{\F})$.

In this case, $\mu - \frac{U}{2}(d_0^{\F} \pm m_0^{\F}) < 4$, so the equations (\ref{Fmeanfieldeqs}) become
\begin{align*}
\begin{cases}
d_0^{\F} + m_0^{\F} = 1 - 2\int_{\max(-4,\mu - \frac{U(d_0^{\F}-m_0^{\F})}{2})}^4 N_0(\epsilon) d\epsilon,
 	\\
d_0^{\F} - m_0^{\F} = 1 - 2\int_{\max(-4,\mu - \frac{U(d_0^{\F}+m_0^{\F})}{2})}^4 N_0(\epsilon) d\epsilon.
\end{cases}
\end{align*}
If $\mu - \frac{U(d_0^{\F}+m_0^{\F})}{2} \leq -4$, then the second equation is 
$d_0^{\F} - m_0^{\F} = -1$, so the assumption of Case 1, $0 \leq \mu < 4 + \frac{U}{2}(d_0^{\F} - m_0^{\F}) = 4 - \frac{U}{2}$, fails since $U \geq U_0 > 8$. Thus, $\mu - \frac{U(d_0^{\F}+m_0^{\F})}{2} > -4$. Consequently, we also have $\mu - \frac{U(d_0^{\F}-m_0^{\F})}{2} > -4$, which means that the must have 
\begin{align}\label{FmeanfieldeqsSectorI}
\begin{cases}
d_0^{\F} + m_0^{\F} = 1 - 2\int_{\mu - \frac{U(d_0^{\F}-m_0^{\F})}{2}}^4 N_0(\epsilon) d\epsilon,
 	\\
d_0^{\F} - m_0^{\F} = 1 - 2\int_{\mu - \frac{U(d_0^{\F}+m_0^{\F})}{2}}^4 N_0(\epsilon) d\epsilon.
\end{cases}
\end{align}
Subtracting the second equation in (\ref{FmeanfieldeqsSectorI}) from the first, we obtain
\begin{align}\label{2DeltaFUint}
m_0^{\F} =  \int_{\mu - \frac{U(d_0^{\F}+m_0^{\F})}{2}}^{\mu - \frac{U(d_0^{\F}-m_0^{\F})}{2}} N_0(\epsilon) d\epsilon.
\end{align}
Since $\mu - \frac{U(d_0^{\F}\pm m_0^{\F})}{2} \in [-4,4]$ and $N_0(\epsilon) \geq \frac{1}{4\pi}$ for $\epsilon \in [-4,4]$, the right-hand side of (\ref{2DeltaFUint}) is $\geq \frac{U m_0^{\F}}{4\pi}$. But then (\ref{2DeltaFUint}) implies that either $m_0^{\F} = 0$ or $1 \geq \frac{U}{4\pi}$. The latter possibility is ruled out by our assumption that $U \geq U_0 \geq 16$, so there is no solution in Case 1 with $m_0^{\F} > 0$.

Case 2. $4 + \frac{U}{2}(d_0^{\F} - m_0^{\F}) \leq \mu$.

In this case, $\frac{U}{2}(d_0^{\F} - m_0^{\F}) - \mu \leq -4$, so the first equation in (\ref{Fmeanfieldeqs}) reduces to
$$d_0^{\F} + m_0^{\F} = 1,$$
and then the second equation in (\ref{Fmeanfieldeqs}) becomes
$$m_0^{\F}
 = \int_{\R}  N_0(\epsilon) \theta\Big(\frac{U}{2}  - \mu + \epsilon \Big) d\epsilon.$$
It follows that if there is a solution of the F mean-field equations in Case 2, then this solution must be given by (\ref{d0Fm0FSectorI}). 
To see that (\ref{d0Fm0FSectorI}) indeed is a solution, we must verify that it fulfills the condition $4 + \frac{U}{2}(d_0^{\F} - m_0^{\F}) \leq \mu$. 
If $d_0^{\F}$ and $m_0^{\F}$ are given by (\ref{d0Fm0FSectorI}), then this condition reduces to
\begin{align}\label{condFcase2}
\frac{U}{2} - U\int_{\R} N_0(\epsilon) \theta\Big(\frac{U}{2}  - \mu + \epsilon \Big) d\epsilon \leq \mu - 4.
\end{align}
For $(U, \mu) \in \I_{U_0, \delta}$, we have either $\mu \in [0, \tfrac{U}{2} - 4]$ or $\mu \in [\tfrac{U}{2} - 4, \tfrac{U}{2} - \delta]$.
If $\mu \in [0, \tfrac{U}{2} - 4]$, then the integral in (\ref{condFcase2}) equals $1$, so the condition holds because $U > 8$ and $\mu \geq 0$. 
If $\mu \in [\tfrac{U}{2} - 4, \tfrac{U}{2} - \delta]$, then the left-hand side of (\ref{condFcase2}) obeys
\begin{align*}
\frac{U}{2} - U\int_{\R}  N_0(\epsilon) \theta\Big(\frac{U}{2}  - \mu + \epsilon \Big) d\epsilon
=
\frac{U}{2} - U\int_{\mu - \frac{U}{2}}^4  N_0(\epsilon)  d\epsilon
\leq \frac{U}{2} - U\int_{-\delta}^4  N_0(\epsilon) d\epsilon
\leq 0,
\end{align*}
so the condition again holds because $0 \leq \mu - 4$ for $U \geq 16$. 
\end{proof}

\begin{lemma}[Expression for $\mathcal{F}_{\F}$ in Sector \I]\label{FFSectorIlemma}
Let $U_0 > 4\pi$ and suppose that $(U, \mu) \in \I_{U_0, \delta}$. If $(d_0^{\F}, m_0^{\F}) = (d_0^{\F}(U, \mu), m_0^{\F}(U, \mu))$ is the solution of Lemma \ref{FsolutionIlemma} of the F mean-field equations, then the following hold:
\begin{enumerate}[$(i)$]
\item If $0 \leq \mu \leq \frac{U}{2} - 4$, then $(d_0^{\F}, m_0^{\F}) = (0, 1)$ and the F free energy satisfies
$$\mathcal{F}_{\F}(U, \mu) = - \frac{U}{4}.$$

\item If $\frac{U}{2} - 4 \leq \mu \leq \frac{U}{2} - \delta$, then the F free energy satisfies
\begin{align}\label{GFexplicitI}
\mathcal{F}_{\F}(U, \mu) = \frac{U}{4} - \mu - \int_{\mu - \frac{U}{2}}^4  N_0(\epsilon) \bigg(\frac{U}{2} + \epsilon - \mu\bigg) d\epsilon.
\end{align}

\end{enumerate}

\end{lemma}
\begin{proof}
Let $(U, \mu) \in \I_{U_0, \delta}$ where $U_0 > 4\pi$. 
By Lemma \ref{FsolutionIlemma}, $(d_0^{\F}, m_0^{\F})$ is the unique solution of the F mean-field equations; hence 
it follows from the definition (\ref{freeenergiesF}) of $\mathcal{F}_{\F}$ that
\begin{align}\label{FFGFI}
\mathcal{F}_{\F}(U, \mu) = \mathcal{G}_{\F}(d_0^{\F}, m_0^{\F}, U, \mu)
\end{align}
for all $(U, \mu) \in \I_{U_0, \delta}$.

If $0 \leq \mu \leq \frac{U}{2} - 4$, then $\frac{U}{2}  - \mu + \epsilon \geq 0$ for all $\epsilon \in [-4,4]$, so it follows immediately from (\ref{d0Fm0FSectorI}) that $(d_0^{\F}, m_0^{\F}) = (0, 1)$, and substituting this into the definition (\ref{FHartreeFockFunction}) of $\mathcal{G}_{\F}$, we infer that
\begin{align*}
\mathcal{G}_{\F}(d_0^{\F}, m_0^{\F}, U, \mu) = &\; \frac{U}{4}  - \mu
- \int_{\R} N_0(\epsilon) \bigg(\frac{U}{2}  - \mu + \epsilon \bigg) d\epsilon
=   - \frac{U}{4}.
\end{align*}
In combination with (\ref{FFGFI}), this proves $(i)$.

Suppose that $\frac{U}{2} - 4 \leq \mu \leq \frac{U}{2} - \delta$. In this case, 
$\frac{U(d_0^{\F} - m_0^{\F})}{2}  - \mu + \epsilon \leq 0$ for all $\epsilon \in [-4,4]$ by (\ref{Fcase2SectorI}), so substitution of (\ref{d0Fm0FSectorI}) into (\ref{FHartreeFockFunction}) using that $d_0^{\F} + m_0^{\F} = 1$ gives
\begin{align*}
\mathcal{G}_{\F}(d_0^{\F}, m_0^{\F}, U, \mu) 
= &\; \frac{U}{4} - \mu - \int_{\R}  N_0(\epsilon) \bigg(\frac{U}{2}  - \mu + \epsilon \bigg) \theta\bigg(\frac{U}{2}  - \mu + \epsilon \bigg) d\epsilon.
\end{align*}
Using (\ref{FFGFI}) and the fact that $\mu - \frac{U}{2} \in [-4, -\delta]$, assertion $(ii)$ follows.
\end{proof}

\subsection{AF free energy in Sector I}\label{AFSectorIsec}
The proofs of the next few lemmas will use the short-hand notation
\begin{align}\label{Deltadef}
\Delta := \frac{U}{2}m_1^{\AF}
\end{align}
and
\begin{align}\label{Xpmdef}
X_\pm := \frac{U}{2} d_0^{\AF} - \mu \pm \sqrt{\Delta^2 + \epsilon^2}.
\end{align}
Using this notation, we can write the AF mean-field equations (\ref{AFmeanfieldeqs}) as
\begin{subequations}\label{AFeqs}
\begin{align}\label{AFeq1}
& 1 - d_0^{\AF} = 2\int_{0}^4 N_0(\epsilon) \Big(\theta(X_+) + \theta(X_-) \Big) d\epsilon,
	\\ \label{AFeq2}
& \frac{1}{U} = \int_{0}^4 N_0(\epsilon) \Big(\theta(X_+)
- \theta(X_-)\Big) \frac{1}{\sqrt{\Delta^2 + \epsilon^2}} d\epsilon.
\end{align}
\end{subequations}
Since $X_+ \geq X_-$, we have $\theta(X_+) - \theta(X_-) \geq 0$. In particular, the two sets
\begin{align}\label{Aplusminusdef}
A_\pm = A_\pm(d_0^{\AF}, m_1^{\AF},U,\mu) := \{\epsilon \in (0,4) \,|\, X_\pm \geq 0\}
\end{align}
always satisfy $A_- \subseteq A_+$, and if $A_-$ is nonempty, then $A_+ = (0,4)$.

Our first lemma shows that there is a unique AF solution with $d_0^{\AF} = 0$ in Sector I if $U$ is large enough, and determines the large $U$ behavior of this solution. 

\begin{lemma}[AF mean-field solution with $d_0^{\AF} = 0$ in Sector I]\label{AFsolutionIlemma}
There is a $U_0 > 0$ such that, if $(U, \mu) \in \I_{U_0, \delta}$, then there is a unique solution $(d_0^{\AF}, m_1^{\AF})$ of the AF mean-field equations (\ref{AFmeanfieldeqs}) satisfying $d_0^{\AF} = 0$ and $m_1^{\AF} > 0$. Moreover, this solution is such that $m_1^{\AF} = m_1^{\AF}(U)$ is independent of $\mu$, $m_1^{\AF}$ depends smoothly on $U \in [U_0, +\infty)$, and
\begin{align}\label{mAFexpansionI}
m_1^{\AF} = 1 - \frac{8}{U^2} + \frac{88}{U^4} + O\bigg(\frac{1}{U^6}\bigg) \qquad \text{as $U \to +\infty$}.
\end{align}
\end{lemma}
\begin{proof}
Assume that $d_0^{\AF} = 0$. Then the AF mean-field equations (\ref{AFeqs}) reduce to
\begin{align}\label{AFeqsd0zero}
\begin{cases}
0 = 1 
- \int_{\R} N_0(\epsilon) \big\{\theta\big( - \mu + \sqrt{\Delta^2 + \epsilon^2}\big)
+ \theta\big( - \mu - \sqrt{\Delta^2 + \epsilon^2}\big)\big\} d\epsilon,
	\\
1
 = \frac{U}{2} \int_{\R}  N_0(\epsilon) \frac{1}{\sqrt{\Delta^2 + \epsilon^2}} \big\{\theta\big( - \mu + \sqrt{\Delta^2 + \epsilon^2}\big) - \theta\big( - \mu - \sqrt{\Delta^2 + \epsilon^2}\big)\big\} d\epsilon.
\end{cases}
\end{align}
Let $\mu \geq 0$. Since $\int_{\R} N_0(\epsilon)  \theta\big( - \mu - \sqrt{\Delta^2 + \epsilon^2}\big) d\epsilon = 0$, the first equation in (\ref{AFeqsd0zero}) can be simplified to
$$1 = \int_{\R} N_0(\epsilon) \theta\big( - \mu + \sqrt{\Delta^2 + \epsilon^2}\big) d\epsilon,$$
which holds if and only if $\sqrt{\Delta^2 + \epsilon^2} \geq \mu$ for all $\epsilon \in [-4,4]$, i.e., if and only if 
$\Delta \geq \mu$. We conclude that 
\begin{enumerate}[$(i)$]
\item any solution $\Delta > 0$ of (\ref{AFeqsd0zero}) must satisfy $\Delta \geq \mu$, and 

\item if $\Delta \geq \mu$, then the first equation in (\ref{AFeqsd0zero}) is automatically fulfilled, while the second equation in (\ref{AFeqsd0zero}) reduces to
\begin{align}\label{Deltaeq}
\frac{2}{U} = &\; \int_{\R} N_0(\epsilon) \frac{1}{\sqrt{\Delta^2 + \epsilon^2}} d\epsilon. 
\end{align}
\end{enumerate}
The right-hand side of (\ref{Deltaeq}) is a decreasing function of  $\Delta$, so (\ref{Deltaeq}) has a unique solution  $\Delta = \Delta(U) > 0$ for every $U > 0$. As $U \to +\infty$, we have $\Delta \to +\infty$, which means that we can expand
\begin{align*}
\frac{2}{U} = &\; \int_{\R} N_0(\epsilon)  \bigg(\frac{1}{\Delta }-\frac{\epsilon^2}{2 \Delta^3}+ \frac{3 \epsilon^4}{8 \Delta^5} + O(\Delta^{-7})\bigg) d\epsilon
	\\
= &\; \frac{1}{\Delta} - \frac{\mathcal{M}_2}{2 \Delta^3}+ \frac{3 \mathcal{M}_4}{8 \Delta ^5} + O(\Delta^{-7}),
\end{align*}
where $\mathcal{M}_j = \int_{\R} N_0(\epsilon) \epsilon^j d\epsilon$. Inverting the series, we obtain
$$\Delta(U) = \frac{U}{2} - \frac{\mathcal{M}_2}{U} + \frac{3 \mathcal{M}_4 - 4 \mathcal{M}_2^2}{U^3} + O(U^{-5}).
$$
By Lemma \ref{N0momentslemma}, we have $\mathcal{M}_2 = 4$ and $\mathcal{M}_4 = 36$, so
\begin{align}\label{DeltaU22U}
\Delta(U) =  \frac{U}{2} - \frac{4}{U} + \frac{44}{U^3}+O(U^{-5}).
\end{align}
Consequently, if $(U, \mu) \in \I_{U_0, \delta}$ with $U_0 > 0$ sufficiently large, then $0 \leq \mu \leq \tfrac{U}{2} - \delta \leq \Delta(U)$, so the condition $\Delta \geq \mu$ is fulfilled. 
This shows that $(d_0^{\AF}, m_1^{\AF}) = (0, \frac{2}{U} \Delta(U))$ is the unique solution of the AF mean-field equations satisfying $d_0^{\AF} = 0$ and $m_1^{\AF} > 0$. Smoothness of $U \mapsto m_1^{\AF} = \frac{2}{U} \Delta(U)$ follows from (\ref{Deltaeq}).
Since the expansion (\ref{mAFexpansionI}) follows from (\ref{DeltaU22U}), the proof is complete.
\end{proof}

Our objective in the next three lemmas is to show that if $(U, \mu) \in \I_{U_0, \delta}$ is large enough, then there are no AF mean-field solutions apart from the one found in Lemma \ref{AFsolutionIlemma}.

\begin{lemma}\label{limsupdeltaUlemma}
Let $U > 0$ and $\mu \geq 0$. If $(d_0^{\AF}, m_1^{\AF}) = (d_0^{\AF}(U, \mu), m_1^{\AF}(U, \mu)) \in \R \times [0, +\infty)$ is an AF mean-field solution, then $d_0^{\AF} \in [0, 1)$. 
\end{lemma}
\begin{proof}
Fix $U > 0$ and $\mu \geq 0$, and let $(d_0^{\AF}, m_1^{\AF})$ be an AF mean-field solution. Suppose first that $d_0^{\AF} < 0$. Then equation (\ref{AFeq1}) can only be satisfied if $A_-$ is nonempty. 
In particular, $\max_{\epsilon \in [-4,4]} X_- = X_-|_{\epsilon = 0}$ is strictly positive, i.e.,
$$|\Delta| < \frac{U}{2} d_0^{\AF} - \mu.$$
Since $d_0^{\AF}  < 0$ and $\mu \geq 0$, the right-hand side is $< 0$ which is a contradiction. 

By (\ref{d0AFbound}), we therefore have $d_0^{\AF} \in [0,1]$. 
If $d_0^{\AF} = 1$, then (\ref{AFeq1}) implies that the sets $A_+$ and $A_-$ have measure $0$, and then (\ref{AFeq2}) reduces to $1/U = 0$, which is a contradiction. 
This completes the proof.
\end{proof}

\begin{lemma}\label{d0zerolemma}
Let $\gamma > 0$.
If $(d_0^{\AF}, m_1^{\AF}) \in \R \times (0, +\infty)$ is an AF mean-field solution corresponding to $(U, \mu)$ with $U > \frac{16\pi}{\gamma}$ and 
$0 \leq \mu \leq \frac{U}{2} - \max(\gamma, \frac{4U}{(\frac{U\gamma}{4\pi} - 4)^2})$, then $d_0^{\AF} = 0$.
\end{lemma}
\begin{proof}
Suppose $(d_0^{\AF}, m_1^{\AF})$ is an AF mean-field solution corresponding to $(U, \mu)$ with $U > \frac{16\pi}{\gamma}$ and 
$0 \leq \mu \leq \frac{U}{2} - \max(\gamma, \frac{4U}{(\frac{U\gamma}{4\pi} - 4)^2})$ such that $d_0^{\AF} \in (0,1)$. By Lemma \ref{limsupdeltaUlemma}, it is enough to show that this leads to a contradiction.

Together with (\ref{AFeq1}) this yields
\begin{align}\label{1int04depsilon}
2\int_{0}^4 N_0(\epsilon) \Big(\theta(X_+) + \theta(X_-) \Big) d\epsilon \in (0,1).
\end{align}
If $A_-$ is nonempty, then $A_+ = (0,4)$ and so $2\int_{0}^4 N_0(\epsilon) \theta(X_+) d\epsilon = 1$ and $2\int_{0}^4 N_0(\epsilon) \theta(X_-) d\epsilon > 0$, which contradicts (\ref{1int04depsilon}). Thus we must have $A_- = \emptyset$ and $A_+ = [b_+, 4)$ for some $b_+ \in (0, 4)$.
Consequently, the mean-field equations (\ref{AFeq1})--(\ref{AFeq2}) can be written as
\begin{align}\label{AFeqswithbplus}
\begin{cases}
1 - d_0^{\AF} = 2\int_{b_+}^4 N_0(\epsilon) d\epsilon,
	\\
\frac{1}{U} = \int_{b_+}^4 N_0(\epsilon)
 \frac{1}{\sqrt{\Delta^2 + \epsilon^2}} d\epsilon.
 \end{cases}
\end{align}
Since $A_+ = [b_+, 4)$ for some $b_+ \in (0, 4)$, the definition (\ref{Xpmdef}) of $X_+$ implies that
\begin{align}\label{bpluseq}
\frac{U d_0^{\AF}}{2} - \mu + \sqrt{\Delta^2 + b_+^2}  = 0.
\end{align}

The second equation in (\ref{AFeqswithbplus}) gives
$$\frac{1}{U} \geq  \frac{1}{\sqrt{\Delta^2 + 16}} \int_{b_+}^4 N_0(\epsilon) d\epsilon;$$
thus, using also the first equation in (\ref{AFeqswithbplus}),
\begin{align}\label{sqrtDelta216}
\sqrt{\Delta^2 + 16} \geq \frac{U}{2}(1-d_0^{\AF}).
\end{align}
Using (\ref{bpluseq}), this yields
\begin{align}\nonumber
\mu & = \frac{U d_0^{\AF}}{2} + \frac{\sqrt{\Delta^2 + b_+^2}}{\sqrt{\Delta^2 + 16}} \sqrt{\Delta^2 + 16} 
> \frac{U d_0^{\AF}}{2} + \frac{1}{\sqrt{1 + \frac{16}{\Delta^2}}} \sqrt{\Delta^2 + 16} 
	\\ \nonumber
&= \frac{U}{2} - \frac{U}{2}(1-d_0^{\AF}) + \frac{1}{\sqrt{1 + \frac{16}{\Delta^2}}} \sqrt{\Delta^2 + 16} 
 \geq \frac{U}{2} - \bigg(1 - \frac{1}{\sqrt{1 + \frac{16}{\Delta^2}}}\bigg)\frac{U}{2}(1-d_0^{\AF})
	\\ \label{muestimate}
&\geq \frac{U}{2} - \bigg(1 - \frac{1}{\sqrt{1 + \frac{16}{\Delta^2}}}\bigg)\frac{U}{2}
\geq \frac{U}{2}\bigg(1 - \frac{8}{\Delta^2}\bigg).
\end{align}

We distinguish two cases. Suppose first that $d_0^{\AF} > 1- \frac{\gamma}{2\pi}$. Then, since $N_0(\epsilon) \geq \frac{1}{4\pi}$ for $\epsilon \in (-4,4)$,
$$\frac{\gamma}{2\pi} > 1 - d_0^{\AF} = 2\int_{b_+}^4 N_0(\epsilon) d\epsilon \geq \frac{4 - b_+}{2\pi},$$
so $b_+ > 4-\gamma$ which gives
$$\mu = \frac{U d_0^{\AF}}{2} + \sqrt{\Delta^2 + b_+^2} 
> \frac{U d_0^{\AF}}{2} + \sqrt{\Delta^2 + (4- \gamma)^2} 
\geq \frac{U d_0^{\AF}}{2} + \sqrt{\Delta^2 + 16} - \gamma
\geq \frac{U}{2} - \gamma,
$$
where we used (\ref{sqrtDelta216}) in the last step.

Suppose now that $d_0^{\AF} \leq 1-\frac{\gamma}{2\pi}$. Then (\ref{sqrtDelta216}) gives
$$\Delta + 4 \geq \sqrt{\Delta^2 + 16} \geq \frac{U}{2}(1 - d_0^{\AF}) \geq \frac{U}{2}\frac{\gamma}{2\pi}.$$
Recalling our assumption that $U > \frac{16\pi}{\gamma}$, this implies that $\Delta \geq \frac{U\gamma}{4\pi} - 4 > 0$,
and hence (\ref{muestimate}) yields
$$\mu > \frac{U}{2}\bigg(1 - \frac{8}{\Delta^2}\bigg) \geq 
\frac{U}{2}- \frac{4U}{(\frac{U \gamma}{4\pi} - 4)^2}.$$

In either case, the assumption
$\mu \leq \frac{U}{2} - \max(\gamma, \frac{4U}{(\frac{U\gamma}{4\pi} - 4)^2})$ is contradicted.
\end{proof}

\begin{lemma}[AF mean-field solution in Sector I]\label{AFfinallemmaSectorI}
There is a $U_0 > 0$ such that, if $(U, \mu) \in \I_{U_0, \delta}$, then the solution $(d_0^{\AF}, m_1^{\AF}) = (0, m_1^{\AF}(U))$ found in Lemma \ref{AFsolutionIlemma} is the only solution in $\R \times (0, +\infty)$ of the AF mean-field equations (\ref{AFmeanfieldeqs}).
\end{lemma}
\begin{proof}
Let $U_\pm := \frac{16 \pi  \left(\delta^2 \pm 2 \sqrt{\pi } \sqrt{\delta^2+\pi }+2 \pi \right)}{\delta^3}$ be the two roots of the polynomial $p(U) := \delta (\frac{U\delta}{4\pi} - 4)^2 - 4U$.
For $U \geq U_+$, we have $p(U) \geq 0$, and hence also $\frac{4U}{(\frac{U\delta}{4\pi} - 4)^2} \leq \delta$.
It follows that 
\begin{align}\label{maxdelta}
\max(\delta, \frac{4U}{(\frac{U\delta}{4\pi} - 4)^2}) = \delta \quad \text{for $U \geq U_+$}.
\end{align}

Let $U_0 = \max\{U_+, \frac{16\pi}{\delta}\}$. If $(U, \mu) \in \I_{U_0, \delta}$, then $U \geq U_0 \geq \frac{16\pi}{\delta}$ and, by (\ref{maxdelta}), $0 \leq \mu \leq \frac{U}{2} - \delta = \frac{U}{2} - \max(\delta, \frac{4U}{(\frac{U\delta}{4\pi} - 4)^2})$, so the desired assertion follows from Lemma \ref{d0zerolemma} applied with $\gamma = \delta$ and Lemma \ref{AFsolutionIlemma}.
\end{proof}

Our next lemma gives the asymptotic behavior of the AF free energy $\mathcal{F}_{\AF}$ in Sector I.

\begin{lemma}[Asymptotics of $\mathcal{F}_{\AF}$ in Sector \I]\label{GAFSectorIlemma}
If $U_0 > 0$ is large enough, then the AF free energy $\mathcal{F}_{\AF}(U, \mu) = \mathcal{F}_{\AF}(U)$ is independent of $\mu$ for $(U, \mu) \in \I_{U_0, \delta}$, and $\mathcal{F}_{\AF}(U)$ depends smoothly on $U \in [U_0, +\infty)$. Furthermore, $\mathcal{F}_{\AF}(U)$ enjoys the following asymptotics as $(U, \mu) \in \I_{U_0, \delta}$ tends to $\infty$:
\begin{align}\label{FAFasymptoticsSectorI}
\mathcal{F}_{\AF}(U) = -\frac{U}{4} - \frac{4}{U} + \frac{20}{U^3} + O\bigg(\frac{1}{U^5}\bigg).
\end{align}
\end{lemma}
\begin{proof}
Lemma \ref{AFfinallemmaSectorI} shows that there is a $U_0 > 0$ such that the solution $(d_0^{\AF}, m_1^{\AF}) = (0, m_1^{\AF}(U))$ found in Lemma \ref{AFsolutionIlemma} is the only AF mean-field solution with $m_1^{\AF} > 0$ whenever $(U, \mu) \in \I_{U_0, \delta}$. Thus, the definition (\ref{freeenergiesAF}) of $\mathcal{F}_{\AF}$ reduces to 
\begin{align}\label{FAFGAFI}
\mathcal{F}_{\AF}(U,\mu) = \mathcal{G}_{\AF}(d_0^{\AF}, m_1^{\AF}, U,\mu).
\end{align}
We infer from (\ref{mAFexpansionI}) that $- \mu + \sqrt{\frac{U^2}{4}(m_1^{\AF})^2 + \epsilon^2} \geq 0$ for all $\epsilon \in \R$ and all sufficiently large $(U, \mu) \in \I_{U_0, \delta}$.
Using also that $d_0^{\AF} = 0$, we obtain from the definition (\ref{AFHartreeFockFunction}) of $\mathcal{G}_{\AF}$ that
\begin{align}\label{GAFd0m1}
\mathcal{F}_{\AF}(U,\mu) = \mathcal{G}_{\AF}(d_0^{\AF}, m_1^{\AF}, U,\mu) = &\; \frac{U}{4} (m_1^{\AF})^2  
 - \int_{\R}  N_0(\epsilon) \sqrt{\frac{U^2}{4}(m_1^{\AF})^2 + \epsilon^2} d\epsilon. 
\end{align}
Since $m_1^{\AF} = m_1^{\AF}(U)$ is independent of $\mu$ by Lemma \ref{AFsolutionIlemma}, so is $\mathcal{F}_{\AF}(U,\mu) = \mathcal{F}_{\AF}(U)$. Smoothness of $U \mapsto \mathcal{F}_{\AF}(U)$ is a consequence of (\ref{GAFd0m1}) and the smoothness of $U \mapsto m_1^{\AF}(U)$ established in Lemma \ref{AFsolutionIlemma}.
Moreover, substituting the expansion (\ref{mAFexpansionI}) for $m_1^{\AF}$ into (\ref{GAFd0m1}) and using Lemma \ref{N0momentslemma}, we arrive at (\ref{FAFasymptoticsSectorI}). 
\end{proof}

\begin{remark}\upshape
By (\ref{AFmeanfield}), we have $\frac{\partial \mathcal{G}_{\AF}}{\partial d_0}(d_0^{\AF}, m_1^{\AF}, U, \mu) = 0$. Hence the equations (\ref{FAFGAFI}) and (\ref{dGAFdd0}) imply that
$$\frac{\partial \mathcal{F}_{\AF}}{\partial \mu}(U, \mu) = \frac{\partial \mathcal{G}_{\AF}}{\partial \mu}(d_0^{\AF}, m_1^{\AF}, U, \mu) = -d_0^{\AF}(U, \mu).$$
Thus the $\mu$-independence of the AF free energy observed in Lemma \ref{GAFSectorIlemma} is a reflection of the fact that the AF state has no doping ($d_0^{\AF} = 0$). Physically this can be understood by noting that the so-called effective AF band relations that appear in (\ref{AFHartreeFockFunction}), $\frac{U}{2} d_0 \pm \sqrt{\frac{U^2}{4}m_1^2 + \epsilon^2}$, are separated by a gap (see also (\ref{Errp})). If $\mu$ lies within the AF gap, i.e., if $|\mu - \frac{U}{2} d_0| < \frac{U}{2}m_1$, then the first Heaviside function in (\ref{AFHartreeFockFunction}) is identically equal to $1$ while the second is identically equal to $0$, and hence 
\begin{align}
\mathcal{G}_{\AF}(d_0, m_1, U, \mu) = &\; \frac{U}{4}  \big(m_1^2 - d_0^2\big)
 - \int_{\R}  N_0(\epsilon) \sqrt{\frac{U^2}{4}m_1^2 + \epsilon^2} d\epsilon.
\end{align}
In particular, the AF free energy does not change as the chemical potential lies within the AF gap $|\mu| < \frac{U}{2}m_1$.
\end{remark}

\subsection{Final steps}
We are now ready to complete the proof of Theorem \ref{sectorIth}.
We know from Lemma \ref{GAFSectorIlemma} that $\mathcal{F}_{\AF}(U, \mu) = \mathcal{F}_{\AF}(U)$ is independent of $\mu$ in Sector I and satisfies (\ref{FAFasymptoticsSectorI}) as $U \to +\infty$; in particular, $\mathcal{F}_{\AF}(U) <  - \frac{U}{4}$ for all large enough $U$.
On the other hand, we know from Lemma \ref{FFSectorIlemma} that, if $0 \leq \mu \leq \frac{U}{2} - 4$, then $\mathcal{F}_{\F}(U, \mu) = - \frac{U}{4}$. It follows that $\mathcal{F}_{\AF}(U) < \mathcal{F}_{\F}(U, \mu)$ in Sector I whenever $0 \leq \mu \leq \frac{U}{2} - 4$.
On the other hand, if $\frac{U}{2} - 4 \leq \mu \leq \frac{U}{2} - \delta$, then Lemma \ref{FFSectorIlemma} shows that $\mathcal{F}_{\F}(U, \mu)$ is given in Sector I by (\ref{GFexplicitI}), and hence that 
\begin{align*}\nonumber
\frac{\partial \mathcal{F}_{\F}}{\partial\mu} (U, \mu)
 =&\; - 1 + \int_{\mu - \frac{U}{2}}^4  N_0(\epsilon) d\epsilon < 0.
\end{align*}
If $\mu = \frac{U}{2} - \delta$, then (\ref{GFexplicitI}) gives (using also \eqref{N0epsilon8pi2})
\begin{align*}
\mathcal{F}_{\F}(U, \mu) 
= -\frac{U}{4} 
 - \frac{8}{\pi^2}
+ \frac{1}{2}\delta
 - \int_{-\delta}^0  N_0(\epsilon) (\epsilon + \delta) d\epsilon.
\end{align*}
Shrinking $\delta > 0$ if necessary, we conclude that $\mathcal{F}_{\F}(U, \mu)$ is a strictly decreasing function of $\mu$ in the interval $[\frac{U}{2} - 4, \frac{U}{2} - \delta]$ such that at the left end-point $\mu = \frac{U}{2} - 4$, we have $\mathcal{F}_{\AF}(U) < \mathcal{F}_{\F}(U, \mu) = - \frac{U}{4}$, and at the right end-point $\mu = \frac{U}{2} - \delta$, we have $\mathcal{F}_{\AF}(U) > \mathcal{F}_{\F}(U, \mu)$.
Therefore there is a unique function $\mu_{\I}(U)$ such that for $(U, \mu) \in \I_{U_0, \delta}$ it holds that
$$\begin{cases}
\mathcal{F}_{\AF}(U) < \mathcal{F}_{\F}(U, \mu) & \text{if $\mu < \mu_{\I}(U)$},	\\ 
\mathcal{F}_{\AF}(U) = \mathcal{F}_{\F}(U, \mu) & \text{if $\mu = \mu_{\I}(U)$}, \\ 
\mathcal{F}_{\AF}(U) > \mathcal{F}_{\F}(U, \mu) & \text{if $\mu > \mu_{\I}(U)$}.
\end{cases}
$$
Increasing $U_0 > 0$ if necessary, we have $\mathcal{F}_{\mathrm{P}}(U, \mu) > -\frac{U}{4}$ everywhere in Sector I by Lemma \ref{FPboundIlemma}. We conclude that 
$$\max\{\mathcal{F}_{\AF}(U), \mathcal{F}_{\F}(U, \mu)\} \leq \mathcal{F}_{\F}(U, 0) 
= - \frac{U}{4} < \mathcal{F}_{\mathrm{P}}(U, \mu)$$ 
for all $(U, \mu) \in \I_{U_0, \delta}$.
The implicit function theorem applied to the defining relation $\mathcal{F}_{\AF}(U) = \mathcal{F}_{\F}(U, \mu_{\I}(U))$ implies that $\mu_{\I}(U)$ depends smoothly on $U \geq U_0$; indeed, smoothness of $U \mapsto \mathcal{F}_{\AF}(U)$ was established in Lemma \ref{GAFSectorIlemma} and smoothness of $(U, \mu) \mapsto \mathcal{F}_{\F}(U, \mu)$ in the relevant domain follows from (\ref{GFexplicitI}). This completes the proof of $(i)$--$(iii)$ of Theorem \ref{sectorIth}.
 
 We next prove (\ref{muIexpansion}).
Using (\ref{FAFasymptoticsSectorI}) and (\ref{GFexplicitI}) in the relation $\mathcal{F}_{\AF}(U) = \mathcal{F}_{\F}(U, \mu_{\I}(U))$, we infer that $\mu_{\I}(U)$ satisfies
\begin{align}\label{muIrelation}
\frac{U}{4} - \mu_{\I}(U)
- \int_{\mu_{\I}(U) - \frac{U}{2}}^4 N_0(\epsilon) \bigg(\frac{U}{2} + \epsilon - \mu_{\I}(U)\bigg) d\epsilon
= -\frac{U}{4}  - \frac{4}{U} + \frac{20}{U^3} + O\bigg(\frac{1}{U^5}\bigg) 
\end{align}
as $U \to + \infty$.
Writing $\mu_{\I}(U) = U/2 - 4 + \tilde{\mu}_{\I}(U)$, we have $\tilde{\mu}_{\I}(U) \in (0, 4-\delta)$ and (\ref{muIrelation}) can be written as
$$ \int_{-4}^{\tilde{\mu}_{\I}(U) - 4}  N_0(\epsilon) (\epsilon + 4 - \tilde{\mu}_{\I}(U)) d\epsilon
= - \frac{4}{U} + \frac{20}{U^3} + O\bigg(\frac{1}{U^5}\bigg),$$
which shows that $\tilde{\mu}_{\I}(U) \downarrow 0$ as $U \to +\infty$.
Substituting in the expansion (\ref{N0near4}) of $N_0(\epsilon) = N_0(-\epsilon)$ and evaluating the integral, we obtain
$$-\frac{\tilde{\mu}_{\I}(U)^2}{8 \pi }
-\frac{\tilde{\mu}_{\I}(U)^3}{192 \pi}
-\frac{5\tilde{\mu}_{\I}(U)^4}{12288 \pi} + O(\tilde{\mu}_{\I}(U)^5)
= - \frac{4}{U} + \frac{20}{U^3} + O\bigg(\frac{1}{U^5}\bigg),$$
which shows that $\tilde{\mu}_{\I}(U) = \frac{4\sqrt{2\pi}}{\sqrt{U}} + O(1/U)$. Extending this calculation to higher order, we arrive at 
\begin{align}\nonumber
\tilde{\mu}_{\I}(U) = &\; \frac{4\sqrt{2\pi}}{\sqrt{U}} 
- \frac{2\pi}{3 U} 
-\frac{5 \pi^{3/2}}{36 \sqrt{2} U^{3/2}} 
-\frac{11 \pi^2}{270 U^{2}} 
-\frac{\sqrt{\frac{\pi }{2}} \left(691200+1163 \pi^2\right)}{34560 U^{5/2}} 
	\\ \nonumber
& + \bigg(\frac{10 \pi }{3}-\frac{18071 \pi^3}{1088640} \bigg) \frac{1}{U^3}
+ \frac{\pi^{3/2} \left(51840000-907207 \pi^2\right)}{49766400 \sqrt{2} U^{7/2}} 
	\\ \label{muItildeexpansion}
& + \bigg(\frac{11 \pi^2}{27} - \frac{561913\pi^{4}}{52254720}\bigg)\frac{1}{U^4}
+ O\bigg(\frac{1}{U^{9/2}}\bigg),
\end{align} 
which implies the expansion in (\ref{muIexpansion}).

Let us derive the expansion (\ref{d0FatmuIexpansion}) of $d_0^{\F}(U, \mu_{\I}(U))$.
Since $\mu_{\I}(U) - \frac{U}{2} \downarrow -4$ as $U \to \infty$, we deduce from (\ref{d0Fm0FSectorI}) and (\ref{N0near4}) that
\begin{align*}
d_0^{\F}(U, \mu_{\I}(U)) 
& = 1 - \int_{\mu_{\I}(U) - \frac{U}{2}}^4 N_0(\epsilon) d\epsilon
= \int_{\frac{U}{2} - \mu_{\I}(U)}^4  N_0(\epsilon) d\epsilon
	\\
& = \frac{\tilde{\mu}_{\I}}{4 \pi }
+\frac{\tilde{\mu}_{\I}^2}{64 \pi }
+\frac{5 \tilde{\mu}_{\I}^3}{3072 \pi}
+\frac{7 \tilde{\mu}_{\I}^4}{32768 \pi }
+\frac{169 \tilde{\mu}_{\I}^5}{5242880 \pi}
+ \frac{269 \tilde{\mu}_{\I}^6}{50331648 \pi }
+ O(\tilde{\mu}_{\I}^7).
\end{align*}
Substituting in the expansion (\ref{muItildeexpansion}) of $\tilde{\mu}_{\I}$, we obtain (\ref{d0FatmuIexpansion}).

The expansion (\ref{freeenergiesatmuIa}) follows immediately from (\ref{FAFasymptoticsSectorI}). 
The expansion (\ref{freeenergiesatmuIb}) of the P free energy follows by evaluating (\ref{FPexpansionSectorIhat}) at 
$\hat{\mu} =  \frac{U}{2} - \mu_{\I}(U) +4 =  8 - \tilde{\mu}_{\I}(U)$ and substituting in the asymptotics (\ref{muItildeexpansion}) for $\tilde{\mu}_{\I}(U)$. 
This completes the proof of Theorem \ref{sectorIth}.

\section{Proof of Theorem \ref{sectorIIth}}\label{SectorIIproofsec}

For the analysis of Sector \II, it is useful to introduce the variable $\mu_E = \mu_E(U,\mu)$ by
\begin{align}\label{muEdef}
\mu_E := \frac{\mu - \mu_{\II,0}(U)}{(4\pi - U)^3},
\end{align}
where $\mu_{\II,0}(U)$ is given by (\ref{muII0def}).
With this definition, we have
\begin{align}\label{mumuE}
\mu = \mu_{\II,0}(U) + \mu_E (4 \pi - U)^3,
\end{align}
and Sector II is characterized by the fact that $\mu_E \in [-M, M]$ as $(U, \mu) \in \II_{U_0, M}$ tends to $(4\pi, 2\pi+4)$.

\subsection{P free energy in Sector II}
We first establish the behavior of the P free energy.

\begin{lemma}[Asymptotics of $\mathcal{F}_{\mathrm{P}}$ in Sector II]
As $(U, \mu) \in \II_{U_0, M}$ tends to $(4\pi, 2\pi+4)$, we have
\begin{align}\label{d0PexpansionSectorII}
d_{0}^{\mathrm{P}} = &\; 1 - \frac{4\pi - U}{\pi^2} - \frac{(4\pi - U)^2}{24\pi^3} -\frac{(1-32\pi^{3}\mu_{E})(4\pi - U)^{3}}{128\pi^{4}} + O((4 \pi -U)^4)
\end{align}
and
\begin{align}\nonumber
\mathcal{F}_{\mathrm{P}}(U, \mu) = &\; -4 - \pi + \bigg(\frac{1}{4} + \frac{4}{\pi}\bigg) (4\pi - U)
 -\frac{24+7 \pi }{12 \pi^3}(4\pi - U)^2 
 	\\ \nonumber
 & + \bigg(\frac{1}{4\pi^4} - \mu_E\bigg)(4\pi - U)^3 
+ \bigg(-\frac{1}{1152 \pi^5} + \frac{\mu_E}{\pi^2}\bigg)(4\pi - U)^4
	\\\nonumber
& - \frac{13-960 \pi^3 \mu_E }{23040 \pi^6}(4 \pi - U)^5
 -\frac{51840 \pi^3 \mu_E \left(16 \pi^3 \mu_E-1\right)+913 }{6635520 \pi^7} (U-4 \pi )^6
	\\\nonumber
&  - \frac{5040 \pi^3 \mu_E \left(576 \pi^3 \mu_E-35\right)+3209 }{92897280 \pi^8}(4 \pi -U)^7
 	\\ \label{FPexpansionSectorII}
&  -\frac{72576 \pi^3 \mu_E \left(3680 \pi^3 \mu_E-237\right)+320297 }{35672555520 \pi^9} (4 \pi - U)^8
 + O((4\pi - U)^9),
\end{align}
where the error terms are uniform with respect to $\mu$, and $\mu_E$ is given by (\ref{muEdef}).
\end{lemma}
\begin{proof}
By Lemma \ref{Plemma}, the P free energy is given by $\mathcal{F}_{\mathrm{P}}(U, \mu) = \mathcal{G}_{\mathrm{P}}(d_0^{\mathrm{P}}(U, \mu), U, \mu)$, where $d_0^{\mathrm{P}}(U, \mu)$ is the (unique) solution of (\ref{Pmeanfieldeq}).
Assume that $(U, \mu) \in \II_{U_0, M}$. Then the variable $\hat{\mu} = \frac{U}{2} - \mu+4$ obeys
\begin{align}\label{hatmuexpansion}
\hat{\mu} = \frac{4}{\pi}(4\pi - U) - \frac{7}{12\pi^2}(4\pi - U)^2 - \mu_E (4\pi - U)^3; 
\end{align}
in particular, $\hat{\mu} = O(4\pi - U)$. Thus, increasing $U_0 < 4\pi$ if necessary, Lemma \ref{Plemma2} (\ref{d0Pitemb}) implies that $d_0^{\mathrm{P}} \in (-1,1)$. It then follows from (\ref{Pmeanfieldeq}) that $\mu - \frac{U}{2} d_0^{\mathrm{P}} \in (-4,4)$, which means that (\ref{Pmeanfieldeq}) can be written as in (\ref{1minusd0PSectorI}), and that the P free energy can be expressed as in (\ref{FPSectorI}).
Using (\ref{hatmuexpansion}), we also infer from the condition $\mu - \frac{U}{2} d_0^{\mathrm{P}} \in (-4,4)$ that $1 - d_0^{\mathrm{P}} = O(4\pi - U)$ as $U \uparrow 4\pi$, which in view of (\ref{1minusd0PSectorI}) implies that $4 - (\mu - \frac{U}{2} d_0^{\mathrm{P}}) = O(4\pi - U)$ as $U \uparrow 4\pi$. 
Thus, we can use the expansion (\ref{N0near4}) in (\ref{1minusd0PSectorI}) to obtain
\begin{align*}
1 - d_0^{\mathrm{P}}  & = \frac{\hat{\mu} - \frac{U}{2}(1-d_0^{\mathrm{P}})}{2\pi} + \frac{\big(\hat{\mu} - \frac{U}{2}(1-d_0^{\mathrm{P}})\big)^2}{32\pi} + \frac{5\big(\hat{\mu} - \frac{U}{2}(1-d_0^{\mathrm{P}})\big)^3}{1536\pi} + O\big((4\pi - U)^4\big).
\end{align*}
Solving for $1-d_0^{\mathrm{P}}$ and utilizing (\ref{hatmuexpansion}), we obtain
\begin{align}\label{1minusd0PexpansionII}
1-d_0^{\mathrm{P}} = \frac{4\pi - U}{\pi^2} 
+ \frac{(4\pi - U)^2}{24 \pi^3} 
+ \frac{1 - 32 \pi^3 \mu_E}{128 \pi^4}(4\pi - U)^3 
    + O\big((4\pi - U)^{4}\big)
\end{align}
as $U \uparrow 4\pi$. On the other hand, utilizing the expansion (\ref{N0near4}) of $N_0$ in (\ref{FPSectorI}), we find
\begin{align*}
\mathcal{F}_{P}(U, \mu) 
= & -  \frac{U}{4} (d_0^{\mathrm{P}})^2 + \frac{U}{2} d_0^{\mathrm{P}} - \mu 
- \frac{(4 - (\mu - \frac{U}{2} d_0^{\mathrm{P}}))^2}{4\pi} 
- \frac{(4 - (\mu - \frac{U}{2} d_0^{\mathrm{P}}))^3}{96\pi} 
	\\
& + O\big((4\pi - U)^4\big)
\end{align*}
as $U \uparrow 4\pi$.
Employing (\ref{1minusd0PexpansionII}) to eliminate $d_0^{\mathrm{P}}$ from this expression, we obtain the first few terms on the right-hand side of (\ref{FPexpansionSectorII}); extending the calculation to higher orders, we obtain also the remaining terms.
\end{proof}

\subsection{F free energy in Sector II}
Our next lemma shows that the F mean-field equations have exactly two distinct solutions in Sector II; we denote these solutions by $(d_{0,1}^{\F}, m_{0,1}^{\F})$ and $(d_{0,2}^{\F}, m_{0,2}^{\F})$, ordered so that $m_{0,1}^{\F} < m_{0,2}^{\F}$. The lemma shows that the F state corresponding to $(d_{0,1}^{\F}, m_{0,1}^{\F})$ has slightly higher free energy than the one corresponding to $(d_{0,2}^{\F}, m_{0,2}^{\F})$. 
Recall from (\ref{dGFdd0}) that $\mathcal{G}_{\F}(d_0, m_0, U, \mu)$ is a strictly concave function of $d_0$.
This means that as $m_0$ increases from $0$ to $+\infty$, the value of $\max_{d_0 \in \R} \mathcal{G}_{\F}(d_0, m_0, U, \mu)$ behaves as follows: it starts out at the P free energy $\mathcal{F}_{\mathrm{P}}(U, \mu)$ at $m_0 = 0$, it then increases to a local maximum at $m_0 = m_{0,1}^{\F}$, then decreases to a local minimum at $m_0 = m_{0,2}^{\F}$, and finally increases again to $+\infty$ as $m_0 \to +\infty$; the curve $\mu_{\II}(U)$ of Theorem \ref{sectorIIth} is characterized by the condition that $m_0 = m_{0,2}^{\F}$ is a global minimum for $\mu < \mu_{\II}(U)$ but only a local minimum for $\mu > \mu_{\II}(U)$.

\begin{lemma}[F mean-field solutions in Sector II]\label{FsolutionIIlemma}
There is a $U_0 < 4 \pi$ such that, if $(U, \mu) \in \II_{U_0, M}$, then the F mean-field equations (\ref{Fmeanfieldeqs}) have exactly two solutions in $\R \times (0, +\infty)$. Denoting these solutions by $(d_{0,1}^{\F}, m_{0,1}^{\F})$ and $(d_{0,2}^{\F}, m_{0,2}^{\F})$, where $0 < m_{0,1}^{\F} < m_{0,2}^{\F}$, the following asymptotic formulas are valid as $(U, \mu) \in \II_{U_0, M}$ tends to $(4\pi, 2\pi+4)$:
\begin{subequations}
\begin{align}\label{d01Fexpansion}
d_{0,1}^{\F} = &\; 1 - \frac{4\pi - U}{\pi^2} + O((4 \pi -U)^2),
	\\ \label{m01Fexpansion}
m_{0,1}^{\F} = &\; \frac{4\pi - U}{\pi^2 \sqrt{2}} + O((4 \pi -U)^2),
	\\ \nonumber
\mathcal{G}_{\F}(d_{0,1}^{\F}, m_{0,1}^{\F}, U, \mu) =& -4 - \pi + \bigg(\frac{1}{4} + \frac{4}{\pi}\bigg) (4\pi - U)
 -\frac{24+7 \pi }{12 \pi^3}(4\pi - U)^2  
 	\\ \nonumber
  & + \bigg(\frac{1}{4\pi^4} - \mu_E\bigg)(4\pi - U)^3 
  	\\ \label{GFatd01Fm01Fexpansion}
  &   + \bigg(\frac{1}{2304 \pi^5} + \frac{\mu_E}{\pi^2}\bigg)(4\pi - U)^4
 + O((4\pi - U)^5),
	\\ \label{d02Fexpansion}
d_{0,2}^{\F} = &\; 1 - \frac{4\pi - U}{\pi^2} - \frac{5(4\pi - U)^2}{48\pi^3}  + O((4 \pi -U)^3),
	\\ \label{m02Fexpansion}
m_{0,2}^{\F} = &\; \frac{4\pi - U}{\pi^2} + \frac{5(4\pi - U)^2}{48\pi^3} + O((4 \pi -U)^3),
	\\\nonumber
\mathcal{G}_{\F}(d_{0,2}^{\F}, m_{0,2}^{\F}, U, \mu) =& -4 - \pi + \bigg(\frac{1}{4} + \frac{4}{\pi}\bigg) (4\pi - U)
 -\frac{24+7 \pi }{12 \pi^3}(4\pi - U)^2 
 	\\ \nonumber
& + \bigg(\frac{1}{4\pi^4} - \mu_E\bigg)(4\pi - U)^3 
+ \bigg(-\frac{1}{1152 \pi^5} + \frac{\mu_E}{\pi^2}\bigg)(4\pi - U)^4
	\\  \nonumber
& - \frac{49-1200 \pi^3 \mu_E }{11520 \pi^6} (4 \pi - U)^5
	\\  \nonumber
& -\frac{51840 \pi^3 \mu_E  \left(4 \pi^3 \mu_E -1\right)+3919}{1658880 \pi^7} (4 \pi - U)^6
	\\  \nonumber
& -\frac{5040 \pi^3 \mu_E  \left(576 \pi^3 \mu_E -133\right)+58847 }{46448640 \pi^8}(4 \pi -U)^7
	\\  \nonumber
& -\frac{48384 \pi^3 \mu_E  \left(5520 \pi^3 \mu_E -1429\right)+6284935 }{8918138880 \pi^9} (4 \pi - U)^8   
 	\\ \label{GFatd02Fm02Fexpansion}
& + O((4\pi - U)^9),
\end{align}
\end{subequations}
where the error terms are uniform with respect to $\mu$, and $\mu_E$ is given by (\ref{muEdef}).
Moreover,
\begin{align}\label{d02Fm02FSectorII}
(d_{0,2}^{\F}, m_{0,2}^{\F}) = \bigg(1 - \int_{\R} N_0(\epsilon) \theta\Big(\frac{U}{2}  - \mu + \epsilon \Big) d\epsilon, \int_{\R}  N_0(\epsilon) \theta\Big(\frac{U}{2}  - \mu + \epsilon \Big) d\epsilon\bigg)
\end{align}
and
\begin{align}\label{dGFdmuequalsd02F}
\frac{d}{d \mu} \mathcal{G}_{\F}(d_{0,2}^{\F}(U,\mu), m_{0,2}^{\F}(U,\mu), U, \mu) = -d_{0,2}^{\F}(U, \mu) \qquad \text{for $(U, \mu) \in \II_{U_0, M}$}.
\end{align}
\end{lemma}
\begin{proof}
Choose $U_0 \in (8, 4\pi)$.
Let $(d_0^{\F}, m_0^{\F})$ be a solution of the F mean-field equations (\ref{Fmeanfieldeqs}) corresponding to $(U, \mu) \in \II_{U_0, \delta}$. As in the proof of Lemma \ref{FsolutionIlemma}, we consider the two cases $0 \leq \mu < 4 + \frac{U}{2}(d_0^{\F} - m_0^{\F})$ and $4 + \frac{U}{2}(d_0^{\F} - m_0^{\F}) \leq \mu$ in turn. 

Case 1. $0 \leq \mu < 4 + \frac{U}{2}(d_0^{\F} - m_0^{\F})$.

In this case, since $U \geq U_0 > 8$, the same argument as in the proof of Lemma \ref{FsolutionIlemma} shows that $(d_0^{\F}, m_0^{\F})$ obeys $\mu - \frac{U(d_0^{\F}\pm m_0^{\F})}{2} \in (-4, 4)$. Hence the F mean-field equations take the form (\ref{FmeanfieldeqsSectorI}), and (\ref{FHartreeFockFunction}) becomes
\begin{align}\nonumber
\mathcal{G}_{\F}(d_0^{\F}, m_0^{\F}, U, \mu) = &\; \frac{U}{4}  \big((m_0^{\F})^2 - (d_0^{\F})^2\big)
 + \frac{U}{2} d_0^{\F} - \mu
 	\\ \nonumber
& - \int_{\mu - \frac{U(d_0^{\F}-m_0^{\F})}{2}}^4 N_0(\epsilon) \bigg(\frac{U(d_0^{\F} - m_0^{\F})}{2}  - \mu + \epsilon\bigg) d\epsilon
	\\  \label{FHartreeFockSectorII}
& - \int_{\mu - \frac{U(d_0^{\F}+m_0^{\F})}{2}}^4 N_0(\epsilon) \bigg(\frac{U(d_0^{\F} + m_0^{\F})}{2}  - \mu + \epsilon \bigg) d\epsilon.
\end{align}
Let us introduce the new variables $\hat{\mu}, a,b$ by
\begin{align}\label{lol1}
\hat{\mu} := \frac{U}{2} - \mu+4, \qquad
a := 4 - \mu + \frac{U(d_0^{\F} - m_0^{\F})}{2}, \qquad
b := 4 - \mu + \frac{U(d_0^{\F} + m_0^{\F})}{2}.
\end{align} 
Since $d_0^{\F} + m_0^{\F} \in (-1,1)$ by (\ref{FmeanfieldeqsSectorI}), we have $0 < a < b < \hat{\mu}$. 
In terms of these variables, we can write (\ref{FmeanfieldeqsSectorI}) as
\begin{align}\label{Fmeanfieldeqsxy}
\begin{cases}
b  = F(a),
 	\\
a = F(b),
\end{cases}
\end{align}
where the function $F = F_{U, \mu}$ is defined by
$$F(x) = \hat{\mu} - U\int_{4-x}^4 N_0(\epsilon) d\epsilon,$$
and we can write (\ref{FHartreeFockSectorII}) as
\begin{align*}
\mathcal{G}_{\F}(d_0^{\F}, m_0^{\F}, U, \mu) = &\; \frac{U}{4} -\frac{(\hat{\mu}-a) (\hat{\mu}-b)}{U} -\mu
 	\\ \nonumber
& - \int_{4-a}^4 N_0(\epsilon) (a - 4  + \epsilon) d\epsilon
 - \int_{4-b}^4  N_0(\epsilon) (b - 4 + \epsilon ) d\epsilon.
\end{align*}
If $(a,b)$ is a solution of (\ref{Fmeanfieldeqsxy}), then $\varphi(a) = 0$, where $\varphi(x) := x - F(F(x))$. Conversely, if 
$\varphi(a) = 0$ for some $a \in (0, \hat{\mu})$ and $a < F(a)$, then $(a,b) = (a,F(a))$ is a solution of (\ref{Fmeanfieldeqsxy}). Similarly, if 
$\varphi(b) = 0$ for some $b \in (0, \hat{\mu})$ and $b > F(b)$, then $(a,b) = (F(b),b)$ is a solution of (\ref{Fmeanfieldeqsxy}).

A calculation using (\ref{hatmuexpansion}) shows that
\begin{align}\label{varphiexpansion}
\varphi(x \hat{\mu}) = \frac{(1-2x) (8 x^2-8 x+1)}{48 \pi^3}(4\pi - U)^3 + O((4\pi - U)^4) \qquad \text{as $U \uparrow 4\pi$}
\end{align}
uniformly for $x \in [0,1]$, and using that
$$\varphi'(x \hat{\mu}) = 1 - F'(F(x \hat{\mu}))F'(x \hat{\mu}) = 1 - U^2 N_0(4 - F(x \hat{\mu})) N_0(4 - x\hat{\mu})$$
we infer that (\ref{varphiexpansion}) can be differentiated term-wise with respect to $x$ without increasing the error term. 
It follows that $\varphi(x \hat{\mu}(U))/(4\pi - U)^3$ tends to the function $\frac{(1-2x) (8 x^2-8 x+1)}{48 \pi^3}$ as $U \uparrow 4\pi$. In particular, if $U$ is sufficiently close to $4\pi$, the function $x \mapsto \varphi(x \hat{\mu}(U))$ has exactly three zeros for $x \in [0,1]$. Denoting these zeros $x_1 < x_2 < x_3$, we obtain
\begin{align}\nonumber
x_1 = &\;  \frac{2-\sqrt{2}}{4} + \frac{-288 \sqrt{2} \pi^3 \mu_E +17 \sqrt{2}+3}{192 \pi } (4\pi - U) + O((4\pi - U)^2),
	\\\nonumber
 x_2 = &\;  \frac{1}{2} +\frac{4 \pi -U}{32 \pi } + O((4 \pi -U)^2),
 	\\ \label{x1x2x3expansions}
 x_3 = &\; \frac{2+\sqrt{2}}{4} + \frac{288 \sqrt{2} \pi^3 \mu_E -17 \sqrt{2}+3}{192 \pi } (4\pi - U) + O((4\pi - U)^2).
\end{align}
By Lemma \ref{Plemma1}, one of these zeros corresponds to the P solution and must satisfy $x\hat{\mu}=F(x\hat{\mu})$; a direct verification shows that this zero is $x_{2}$, i.e. $x_2 \hat{\mu} = F(x_2 \hat{\mu})$. Furthermore, $x_3 \hat{\mu}= F(x_1 \hat{\mu} )$, $x_3 \hat{\mu} = F(x_1 \hat{\mu})$, and $0 < x_1 \hat{\mu} < x_3 \hat{\mu} < \hat{\mu}$. We infer that there is a unique solution of the F mean-field equations in Case 1 with $m_0^{\F} > 0$; this is the solution corresponding to $(a,b) = (x_1 \hat{\mu}, x_3 \hat{\mu})$ and we denote it by $(d_{0,1}^{\F}, m_{0,1}^{\F})$. It follows from (\ref{hatmuexpansion}), \eqref{lol1}, and (\ref{x1x2x3expansions}) that  
\begin{align}
\begin{cases}
d_{0,1}^{\F} = 1 - \frac{2\hat{\mu} - a - b}{U} = 1 - \hat{\mu}\frac{2 - x_1 - x_3}{U} 
= 1 - \frac{4\pi - U}{\pi^2} + O((4 \pi -U)^2),
	\\ 
m_{0,1}^{\F} =\frac{b - a}{U} = \hat{\mu}\frac{x_3 - x_1}{U} 
= \frac{4\pi - U}{\pi^2 \sqrt{2}} + O((4 \pi -U)^{2}),
\end{cases}
\end{align}
which proves (\ref{d01Fexpansion}) and (\ref{m01Fexpansion}). On the other hand, utilizing (\ref{hatmuexpansion}), (\ref{x1x2x3expansions}), and (\ref{N0near4}) in the expression
\begin{align*}
\mathcal{G}_{\F}(d_{0,1}^{\F}, m_{0,1}^{\F}, U, \mu) = &\;-\frac{U}{4}-4+\hat{\mu} - \hat{\mu}^2 \frac{(1-x_1) (1-x_3)}{U}
 	\\ \nonumber
& - \int_{4-x_1 \hat{\mu}}^4  N_0(\epsilon) (x_1 \hat{\mu} - 4  + \epsilon) d\epsilon
 - \int_{4-x_3 \hat{\mu}}^4  N_0(\epsilon) (x_3 \hat{\mu} - 4 + \epsilon ) d\epsilon,
\end{align*}
we obtain the first few terms in (\ref{GFatd01Fm01Fexpansion}); extending the calculation to higher order gives \eqref{GFatd01Fm01Fexpansion}.

Case 2. $4 + \frac{U}{2}(d_0^{\F} - m_0^{\F}) \leq \mu$.

In this case, the F mean-field equations (\ref{Fmeanfieldeqs}) reduce to
$$d_0^{\F} + m_0^{\F} = 1, \qquad
m_0^{\F}
 = \int_{\R}  N_0(\epsilon) \theta\Big(\frac{U}{2}  - \mu + \epsilon \Big) d\epsilon,$$
so (\ref{d02Fm02FSectorII}) is the unique solution of (\ref{Fmeanfieldeqs}) provided that the condition $4 + \frac{U}{2}(d_{0,2}^{\F} - m_{0,2}^{\F}) \leq \mu$ is fulfilled.
To verify this condition, we use (\ref{d02Fm02FSectorII}) to write it as
\begin{align}\label{condFcase2SectorII}
 \mu - \frac{U}{2} - 4 + U\int_{\R} N_0(\epsilon) \theta\Big(\frac{U}{2}  - \mu + \epsilon \Big) d\epsilon \geq 0.
\end{align}
If $(U, \mu) \in \II_{U_0, M}$, then we see using (\ref{N0near4}) that the left-hand side of (\ref{condFcase2SectorII}) equals
 \begin{align*}
 \mu - \frac{U}{2} - 4 + U\int_{4 - \hat{\mu}}^4 N_0(\epsilon) d\epsilon
&= \frac{(4 \pi - U)^3}{48 \pi^3}
 + O((4 \pi -U)^4).
 \end{align*}
Increasing $U_0 < 4\pi$ if necessary, we conclude that (\ref{condFcase2SectorII}) indeed holds for all $(U, \mu) \in \II_{U_0, \delta}$. A calculation using (\ref{d02Fm02FSectorII}), (\ref{hatmuexpansion}), and (\ref{N0near4}) gives (\ref{d02Fexpansion}) and (\ref{m02Fexpansion}).
Finally, by (\ref{FHartreeFockFunction}),
 \begin{align*}
\mathcal{G}_{\F}(d_{0,2}^{\F}, m_{0,2}^{\F}, U, \mu) = 
 \frac{U}{4}  \big((m_{0,2}^{\F})^2 - (d_{0,2}^{\F})^2\big)
 + \frac{U}{2} d_{0,2}^{\F} - \mu
 - \int_{4 - \hat{\mu}}^4  N_0(\epsilon) \big(\hat{\mu} - 4 + \epsilon \big) d\epsilon,
\end{align*}
and employing (\ref{hatmuexpansion}), (\ref{d02Fexpansion}), (\ref{m02Fexpansion}), and (\ref{N0near4}) in this formula, we arrive at the first few terms in (\ref{GFatd02Fm02Fexpansion}); extending the calculation to higher order gives (\ref{GFatd02Fm02Fexpansion}).
\end{proof}

\begin{lemma}\label{FsolutionIIlemma2}
There is a $U_0 < 4 \pi$ such that, if $(U, \mu) \in \II_{U_0, M}$, then 
\begin{align}\label{FequalsGFd02Fm02F}
\mathcal{F}_{\F}(U, \mu) = \mathcal{G}_{\F}(d_{0,2}^{\F}, m_{0,2}^{\F}, U, \mu)
\end{align}
and
\begin{align}\label{dFFdmuequalsd02F}
\frac{\partial \mathcal{F}_{\F}}{\partial \mu} (U, \mu) = -d_{0,2}^{\F}(U, \mu),
\end{align}
where $(d_{0,2}^{\F}, m_{0,2}^{\F})$ is given by (\ref{d02Fm02FSectorII}).
\end{lemma}
\begin{proof}
The expansions (\ref{GFatd01Fm01Fexpansion}) and (\ref{GFatd02Fm02Fexpansion}) imply that 
$$\mathcal{G}_{\F}(d_{0,1}^{\F}, m_{0,1}^{\F}, U, \mu) - \mathcal{G}_{\F}(d_{0,2}^{\F}, m_{0,2}^{\F}, U, \mu)
=\bigg(\frac{1}{2304 \pi^5} + \frac{1}{1152 \pi^5}\bigg)(4\pi - U)^4 + O((4\pi - U)^5)$$
as $U\uparrow 4\pi$ in Sector  II. 
Since $(d_{0,1}^{\F}, m_{0,1}^{\F})$ and $(d_{0,2}^{\F}, m_{0,2}^{\F})$ are the only two F mean-field solutions in $\R \times (0, +\infty)$ by Lemma \ref{FsolutionIIlemma}, the lemma follows from (\ref{freeenergiesF}) and (\ref{dGFdmuequalsd02F}).
\end{proof}

\subsection{AF free energy in Sector II}

We will show that the AF mean-field equations have no solutions with $m_1^{\AF} > 0$ whenever $\mu \geq U/2$. Increasing $U_0 < 4\pi$ if necessary, we have $\mu \geq U/2$ whenever $(U, \mu) \in \II_{U_0, M}$ (see Figure \ref{sectorsfig}), which means that no AF states exist in Sector II.

\begin{lemma}[No AF mean-field solutions in Sector II]\label{AFnosolutionlemma}
Suppose $U > 0$ and $\mu \geq U/2$. Then the AF mean-field equations (\ref{AFmeanfieldeqs}) have no solution $(d_0^{\AF}, m_1^{\AF}) \in \R \times (0, +\infty)$.
\end{lemma}
\begin{proof}
Let $(d_0^{\AF}, m_1^{\AF}) \in \R \times (0, +\infty)$ be a solution of (\ref{AFmeanfieldeqs}).
As in Section \ref{AFSectorIsec}, we let $\Delta = \frac{U}{2}m_1^{\AF}$ and let $X_\pm$ and $A_\pm$ be as in (\ref{Xpmdef}) and (\ref{Aplusminusdef}), respectively.

Suppose first that $d_0^{\AF} = 0$. As in the proof of Lemma \ref{AFsolutionIlemma}, this implies that $\Delta \geq \mu$ and that $\Delta$ satisfies (\ref{Deltaeq}). Our assumption $\mu \geq U/2$ gives $\Delta \geq U/2$ and hence (\ref{Deltaeq}) leads to the following contradiction:
\begin{align*}
\frac{2}{U} \leq \int_{\R} N_0(\epsilon) \frac{1}{\sqrt{\frac{U^2}{4} + \epsilon^2}}  d\epsilon
< \int_{\R}  N_0(\epsilon) \frac{1}{\sqrt{\frac{U^2}{4}}} d\epsilon = \frac{2}{U},
\end{align*}
showing that there is no solution with $d_0^{\AF} = 0$.

Suppose now that $d_0^{\AF} \neq 0$. By Lemma \ref{limsupdeltaUlemma}, we then have $d_0^{\AF} \in (0,1)$. 
As in the proof of Lemma \ref{d0zerolemma}, we find that there is a $b_+ \in (0, 4)$ such that (\ref{AFeqswithbplus}) and (\ref{bpluseq}) hold. In other words, the system of equations
\begin{subequations}\label{AFthreeeqs}
\begin{align}\label{AFthreeeqsa}
& 1 - d_0^{\AF} = 2\int_{b_+}^4 N_0(\epsilon) d\epsilon,
	\\\label{AFthreeeqsb}
& \frac{1}{U} = \int_{b_+}^4  N_0(\epsilon) 
 \frac{1}{\sqrt{\Delta^2 + \epsilon^2}} d\epsilon,
	\\\label{AFthreeeqsc}
& \frac{U d_0^{\AF}}{2} - \mu + \sqrt{\Delta^2 + b_+^2}  = 0,
\end{align}
\end{subequations}
has a solution $(d_0^{\AF}, \Delta, b_+)$ with $d_0^{\AF} \in (0,1)$, $\Delta > 0$, and $b_+ \in (0,4)$. We will show that this is only possible if $\mu < U/2$, which will complete the proof that no solution can exist if $\mu \geq U/2$.

Since $b_+ \in (0,4)$, equation (\ref{AFthreeeqsb}) implies that
$$\frac{1}{U} < \int_{b_+}^4  N_0(\epsilon) \frac{1}{\sqrt{\Delta^2 + b_+^2}} d\epsilon.$$
The function $c_+ \mapsto \int_{c_+}^4 N_0(\epsilon) \frac{1}{\sqrt{\Delta^2 + c_+^2}} d\epsilon$ decreases from 
$\int_{b_+}^4 N_0(\epsilon) \frac{1}{\sqrt{\Delta^2 + b_+^2}} d\epsilon$ to $0$ as $c_+$ increases from $b_+$ to $4$. It follows that the equation
\begin{align}\label{cpluseq}
\frac{1}{U} = \int_{c_+}^4 N_0(\epsilon) \frac{1}{\sqrt{\Delta^2 + c_+^2}} d\epsilon
\end{align}
has a unique solution $c_+ \in (b_+, 4)$. 
From (\ref{cpluseq}) and (\ref{AFthreeeqsa}), we deduce that $b_+ < c_+$ satisfies
$$\sqrt{\Delta^2 + b_+^2} < \sqrt{\Delta^2 + c_+^2} = U\int_{c_+}^4 N_0(\epsilon) d\epsilon < U\int_{b_+}^4  N_0(\epsilon) d\epsilon = \frac{U}{2}(1 - d_0^{\AF}).$$
Using (\ref{AFthreeeqsc}) to eliminate $d_0^{\AF}$ from the right-hand side, we can write this as
$$\sqrt{\Delta^2 + b_+^2} < \frac{U}{2} - \mu + \sqrt{\Delta^2 + b_+^2},$$
which shows that $\mu < U/2$ as desired.
\end{proof}

\subsection{Final steps}
By Lemma \ref{AFnosolutionlemma}, the AF mean-field equations (\ref{AFmeanfieldeqs}) have no solution in $\R \times (0, +\infty)$ in Sector II, so we only have to consider the P and F free energies.
Comparing (\ref{d0PexpansionSectorII}) and (\ref{d02Fexpansion}), we see that, increasing $U_0 \in (0, 4\pi)$ if necessary, we have $d_{0,2}^{\F} < d_{0}^{\mathrm{P}}$ for all $(U, \mu) \in \II_{U_0, M}$.
By (\ref{dFPdmuequalsd0P}) and (\ref{dFFdmuequalsd02F}), this means that the function $\mu \mapsto \mathcal{F}_{\mathrm{P}}(U, \mu) - \mathcal{F}_{\F}(U, \mu)$ is strictly decreasing for $(U, \mu) \in \II_{U_0, M}$.
Moreover, we deduce from (\ref{FPexpansionSectorII}), (\ref{GFatd02Fm02Fexpansion}), and (\ref{FequalsGFd02Fm02F}) that $\mathcal{F}_{\mathrm{P}}(U, \mu) - \mathcal{F}_{\F}(U, \mu)$ is strictly positive for $\mu = \mu_{\II,0}(U) - M (4\pi - U)^3$ and strictly negative for $\mu = \mu_{\II,0}(U) + M (4\pi - U)^3$, whenever $U$ is sufficiently close to $4\pi$.
Consequently, increasing $U_0$ if necessary, there is a unique function $\mu_{\II}(U)$ such that for $(U, \mu) \in \II_{U_0, M}$ it holds that
$$\begin{cases}
\mathcal{F}_{\F}(U, \mu) < \mathcal{F}_{\mathrm{P}}(U, \mu) & \text{if $\mu < \mu_{\II}(U)$},	\\ 
\mathcal{F}_{\F}(U, \mu) = \mathcal{F}_{\mathrm{P}}(U, \mu) & \text{if $\mu = \mu_{\II}(U)$}, \\ 
\mathcal{F}_{\F}(U, \mu) > \mathcal{F}_{\mathrm{P}}(U, \mu) & \text{if $\mu > \mu_{\II}(U)$}.
\end{cases}
$$
It is a consequence of Lemma \ref{FsolutionIIlemma2} that $\mathcal{F}_{\F}$ depends smoothly on $(U, \mu) \in \II_{U_0, M}$. Smoothness of $U \mapsto \mu_{\II}(U)$ therefore follows from the implicit function theorem and the smoothness properties of $\mathcal{F}_{\mathrm{P}}$ established in Lemma \ref{Plemma}.
Defining $\mu_{E, \II} = \mu_{E, \II}(U)$ by
\begin{align}\label{muEIIdef}
\mu_{\II}(U) = \mu_{\II,0}(U) + \mu_{E, \II} (4 \pi - U)^3,
\end{align}
it follows from (\ref{FPexpansionSectorII}), (\ref{GFatd02Fm02Fexpansion}), and (\ref{FequalsGFd02Fm02F}) that 
\begin{align*}
& \frac{\left(17-288 \pi^3 \mu_{E, \II}\right) (4 \pi-U )^5}{4608 \pi^6}
 + \frac{\left(4921-51840 \pi^3 \mu_{E, \II}\right) (4 \pi-U )^6}{2211840 \pi^7}
	\\
& + \frac{\left(432 \pi^3 \mu_{E, \II} \left(192 \pi^3 \mu_{E, \II}-77\right)+3271\right)
   (4 \pi -U)^7}{2654208 \pi^8}
	\\
& + \frac{\left(24192 \pi^3 \mu_{E, \II} \left(33120 \pi^3 \mu_{E, \II}-10721\right)+24819443\right) (4 \pi-U )^8}{35672555520 \pi^9} = O((4\pi - U)^9).
\end{align*}
It results that
\begin{align*}
 \mu_{E, \II}(U) = &\; \frac{17}{288 \pi^3}
 + \frac{1861}{138240 \pi^4} (4 \pi - U)
 + \frac{15181}{3317760 \pi^5} (4 \pi - U)^2
 	\\
& +  \frac{469909}{247726080 \pi^6} (4 \pi - U)^3
 + O((4\pi - U)^{4}).
\end{align*}
Substituting this expansion into (\ref{muEIIdef}) and (\ref{GFatd02Fm02Fexpansion}), as well as into higher-order versions of (\ref{d0PexpansionSectorII}) and (\ref{d02Fexpansion}), we find the asymptotic formulas (\ref{muIIexpansion}) 
and (\ref{freeenergiesatmuII}), as well as (\ref{d0FatmuIIexpansion}) and (\ref{d0PatmuIIexpansion}). 
This completes the proof of Theorem \ref{sectorIIth}.

\section{Proof of Theorem \ref{sectorIIIth}}\label{SectorIIIproofsec}

\subsection{P free energy in Sector III}
The following lemma provides the asymptotic behavior of the P free energy in Sector III.

\begin{lemma}[Asymptotics of $\mathcal{F}_{\mathrm{P}}$ in Sector III]\label{FPsectorIIIlemma}
Let $\delta \in (0,8)$.
As $(U, \mu) \in \III_{U_0, \delta}$ tends to $(0,0)$, we have
\begin{align}\label{d0PexpansionSectorIII}
d_0^{\mathrm{P}} = &\;  e^{-\frac{2\pi}{\sqrt{U}}} \hat{\mu} \bigg(\frac{2}{\pi \sqrt{U}} 
+ \frac{\ln \big(\frac{16}{\hat{\mu}}\big)-1}{\pi^2}
- \frac{2 \ln \big(\frac{16}{\hat{\mu}}\big)-1}{\pi^3}\sqrt{U}
+ O(U)\bigg)
\end{align}
and
\begin{align}\nonumber
\mathcal{F}_{\mathrm{P}}(U, \mu) 
= & - \frac{16}{\pi^2} + e^{-\frac{4\pi}{\sqrt{U}}} \hat{\mu}^2 \bigg\{ 
-\frac{1}{\pi  \sqrt{U}}
 + \frac{1-2 \ln (\frac{16}{\hat{\mu}})}{4 \pi^2}
+ \frac{\ln (\frac{16}{\hat{\mu}})}{\pi^3} \sqrt{U}
	\\ \label{FPexpansionSectorIII}
& + \frac{(\ln (\frac{16}{\hat{\mu}})-4) \ln
   (\frac{16}{\hat{\mu}})+1}{4 \pi^4} U
+ \frac{3 (4-3 \ln (\frac{16}{\hat{\mu}})) \ln
   (\frac{16}{\hat{\mu}})-5}{12 \pi^5} U^{3/2}
   + O(U^2)
   \bigg\},
\end{align}
where $\hat{\mu} := \mu e^{\frac{2\pi}{\sqrt{U}}} \in [16 + \delta, 32 - \delta]$, and the error term is uniform with respect to $\hat{\mu}$.
\end{lemma}
\begin{proof}
By Lemma \ref{Plemma2}, $d_0^{\mathrm{P}} \in (0,1]$ whenever $U > 0$ and $\mu > 0$.
Consequently, as $(U, \mu)$ tends to  $(0,0)$ in Sector III, the quantity
\begin{align}\label{XdefIII}
X := \mu - \frac{U d_0^{\mathrm{P}}}{2}
\end{align}
tends to $0$, which means that the P mean-field equation (\ref{Pmeanfieldeq}) takes the form
\begin{align}\label{PXmeanfieldeq}
2\frac{\mu - X}{U}
= 1 - 2 \int_{X}^4 N_0(\epsilon) d\epsilon = 2 \int_0^X  N_0(\epsilon) d\epsilon.
\end{align}
It follows that $X \in [0, \mu]$ in Sector III, because otherwise the two sides of (\ref{PXmeanfieldeq}) have different signs. 
Furthermore, expanding the right-hand side of (\ref{PXmeanfieldeq}) using (\ref{N0expansionfirstfew}), we find 
\begin{align}\label{d0eqsmallU}
2\frac{\mu - X}{U}
= \frac{X \big(\ln \big(\frac{16}{X}\big)+1\big)}{\pi^2}
+ O\bigg( X^3 \ln\frac{1}{X}\bigg).
\end{align}

Recall that the Lambert $W$ function $W_{-1}(x)$ is defined for $-\frac{1}{e} \leq x < 0$ as follows:
for each $-\frac{1}{e} \leq x < 0$, $W_{-1}(x)$ is the unique solution of the equation $y e^y = x$ that lies in $(-\infty,-1]$.
According to \cite[Eq. (4.19); see also text below (4.20)]{CGHJK1996}, the $W_{-1}$ obeys the following expansion:
\begin{align}\label{Wasymptotics}
W_{-1}(-x) = L_1 - L_2 + \frac{L_2}{L_1} + \frac{L_2(L_2-2)}{2L_1^2} + O\bigg(\frac{L_2^3}{L_1^3}\bigg) \qquad \text{as $x \downarrow 0$},
\end{align}
where $L_1 = -\ln(\frac{1}{x})$ and $L_2 = \ln(\ln(\frac{1}{x}))$.

Solving (\ref{d0eqsmallU}) for $X$, we obtain
\begin{align}\label{Xsolvedfor}
X = -\frac{2 \pi^2 \mu(1+ O(Y))}{U W_{-1}(-Z)},
\end{align}
where $Y, Z$ are short-hand notations for 
$$Y := \frac{U X^3}{\mu} \ln\frac{1}{X}, \qquad Z := \frac{\pi^2 e^{-\frac{2 \pi^2}{U}-1} \mu(1+ O(Y))}{8 U}.$$ 
Since $X \in [0, \mu]$, we have $Y = O(U X^2 \ln\frac{1}{X})$. In particular, as $(U, \mu)$ tends to  $(0,0)$ in Sector III, $Y$  tends to zero, and we obtain
\begin{align*}
\ln\frac{1}{Z} = \frac{2 \pi^2}{U} + \frac{2\pi}{\sqrt{U}} + \ln(U) + 1 + \ln \frac{8}{\pi^2 \hat{\mu}} + O(Y)
\end{align*}
and
\begin{align*}
\ln \bigg(\ln \frac{1}{Z}\bigg) = &\; \ln(2\pi^2) - \ln(U) + \frac{\sqrt{U}}{\pi} 
+ \frac{U \left(\ln \left(\frac{8}{\pi^2 \hat{\mu}}\right)+ \ln{U}\right)}{2
   \pi^2}
   	\\
& -\frac{U^{3/2} \left(3 \ln \left(\frac{8}{\pi^2 \hat{\mu}}\right)+3
   \ln{U} +1\right)}{6 \pi^3}
   + O(U^2 |\ln U|^2).
\end{align*}
Hence, the asymptotic formula (\ref{Wasymptotics}) yields
\begin{align*}
W_{-1}(-Z)
& = -\ln\frac{1}{Z} - \ln \bigg(\ln \frac{1}{Z}\bigg) - \frac{\ln (\ln \frac{1}{Z}) }{\ln\frac{1}{Z}} + O\bigg(\bigg(\frac{\ln (\ln \frac{1}{Z}) }{\ln\frac{1}{Z}}\bigg)^2\bigg)
	\\
& = - \frac{2 \pi^2}{U} - \frac{2\pi}{\sqrt{U}} - 1 - \ln\bigg(\frac{16}{\hat{\mu}} \bigg)  - \frac{\sqrt{U}}{\pi} - \frac{\ln\big(\frac{16}{\hat{\mu}}\big)}{2\pi^2} U + O(U^{3/2})
\end{align*}
as $(U, \mu)$ tends to  $(0,0)$ in Sector III.
Substituting this expansion into (\ref{Xsolvedfor}), we get
\begin{align}\label{Xexpansion}
X = &\; \mu\bigg(1 - \frac{\sqrt{U}}{\pi} 
-\frac{\ln \big(\frac{16}{\hat{\mu}}\big)-1}{2 \pi^2}U
+ \frac{2 \ln \big(\frac{16}{\hat{\mu}}\big)-1}{2 \pi^3}U^{3/2}
+ O(U^2)\bigg).
\end{align}
The expansion (\ref{d0PexpansionSectorIII}) is a direct consequence of (\ref{XdefIII}) and (\ref{Xexpansion}).

According to (\ref{FPGP}) and (\ref{PHartreeFockFunction}), the P free energy is given by
\begin{align*}
\mathcal{F}_{\mathrm{P}}(U, \mu) = & -  \frac{U}{4} (d_0^{\mathrm{P}})^2 - X - 2\int_X^4  N_0(\epsilon)(\epsilon  - X) d\epsilon
	\\
= & - \frac{(\mu - X)^2}{U} - \frac{16}{\pi^2} + 2\int_0^X  N_0(\epsilon)(\epsilon  - X) d\epsilon,
\end{align*}
where we have used (\ref{N0epsilon8pi2}) in the second step. Employing (\ref{N0expansionfirstfew}) again, we find
\begin{align*}
\mathcal{F}_{\mathrm{P}}(U, \mu) 
= & - \frac{(\mu - X)^2}{U} - \frac{16}{\pi^2} -\frac{X^2 \left(2 \ln \left(\frac{16}{X}\right)+3\right)}{4 \pi^2} + O\bigg( X^4 \ln\frac{1}{X}\bigg).
\end{align*}
Substituting (\ref{Xexpansion}) into this formula, we arrive at (\ref{FPexpansionSectorIII}).
\end{proof}

\subsection{F free energy in Sector III}
In this section, we prove the following proposition.

\begin{proposition}[No F mean-field solutions in Sector III]\label{FsolutionIIIprop}
There is a $U_0 > 0$ such that, if $(U, \mu) \in \III_{U_0, \delta}$, then the F mean-field equations (\ref{Fmeanfieldeqs}) have no solution $(d_0^{\F}, m_{0}^{\F}) \in \R \times (0, +\infty)$.
\end{proposition}

The remainder of this section is devoted to the rather involved proof of Proposition \ref{FsolutionIIIprop}.
 
First note that if $(d_0^{\F}, m_0^{\F})$ solves the F mean-field equations (\ref{Fmeanfieldeqs}), then, by (\ref{d0Fm0Fbound}), we have $d_0^{\F} \in [-1,1]$ and $m_0^{\F} \in [0,1]$. Thus, shrinking $U_0 > 0$ if necessary, it holds that $\mu - \frac{U(d_0^{\F}\pm m_0^{\F})}{2} \in (-4, 4)$ whenever $(U, \mu) \in \III_{U_0, \delta}$, and thus the F mean-field equations take the form (\ref{FmeanfieldeqsSectorI}).
Introducing new variables $(u, v)$ by 
$$u := \frac{U(d_0^{\F} - m_0^{\F})}{2}, \qquad
v := \frac{U(d_0^{\F} + m_0^{\F})}{2},$$
we can write (\ref{FmeanfieldeqsSectorI}) as
\begin{align}\label{Fmeanfieldeqsuv}
\begin{cases}
v  = G(u),
 	\\
u = G(v),
\end{cases}
\end{align}
where the function $G = G_{U, \mu}$ is defined by 
\begin{align}\label{GIIIdef}
 G(x) = U \int_0^{\mu - x} N_0(\epsilon) d\epsilon.
\end{align}
Observe that $G(x)$ obeys
$$G'(x) = -U N_0(\mu - x) \leq 0,$$
where the inequality is strict for $x \in (-4+\mu, 4 + \mu)$. Moreover, $G(0) > 0$ and $G(\mu) = 0$. Hence there is a unique point $x_1 = x_1(U,\mu) \in (0, \mu)$ such that
\begin{align}\label{x1SectorIIIdef}
G(x_1) = x_1.
\end{align}
The solution $(u,v) = (x_1, x_1)$ of (\ref{Fmeanfieldeqsuv}) has $m_0^{\F} = 0$ and therefore corresponds to the P solution. Solutions of (\ref{Fmeanfieldeqs}) with $m_0^{\F} > 0$ correspond to solutions $(u,v)$ of (\ref{Fmeanfieldeqsuv}) such that $u < v$; in fact, $- U \leq u < v \leq U$ as a consequence of (\ref{d0Fm0Fbound}). Moreover, $u < \mu$ because otherwise we would have $v =G(u) \leq 0 <\mu \leq u$ contradicting the inequality $u<v$.
Also, if $(u,v)$ solves (\ref{Fmeanfieldeqsuv}), then
\begin{align}\label{uGGu}
u - G(G(u)) = 0.
\end{align}
In order to prove Proposition \ref{FsolutionIIIprop}, it is therefore sufficient to show the following:
\begin{align}\label{Gclaim}
 \text{If $(U, \mu) \in \III_{U_0, \delta}$, then (\ref{uGGu}) has no solution $u \in [-U, \mu)$ with $u < G(u) \leq U$.}
\end{align}

In order to prove (\ref{Gclaim}), we will analyze the function $G$ and the solution $x_1$ in a series of lemmas.

\begin{lemma}\label{Gnear0lemma}
The function $G$ defined in (\ref{GIIIdef}) satisfies
\begin{align}\label{Gnear0}
G(x) = U\frac{(\mu - x ) \big(\ln \big(\frac{16}{|\mu -x|}\big) + 1\big)}{2 \pi^2}
 + O\bigg(U|\mu - x|^3 \ln\frac{1}{|\mu - x|}\bigg)
 \end{align}
uniformly for $x \in [-U, U]$ as $(U, \mu) \in \III_{U_0, \delta}$ tends to $(0,0)$.
\end{lemma}
\begin{proof}
Using the expansion (\ref{N0expansionfirstfew}) of $N_0(\epsilon)$ as $\epsilon \to 0$, we obtain, for $x \in [-U, \mu]$,
\begin{align*}
G(x) 
& = U\int_0^{\mu - x} \bigg(\frac{\ln(\frac{16}{\epsilon})}{2\pi^2}  + O\bigg(\epsilon^2 \ln\bigg(\frac{16}{\epsilon}\bigg)\bigg)\bigg) d\epsilon,
\end{align*}
from which (\ref{Gnear0}) follows for $x \in [-U, \mu]$ after integration. Since $G(\mu-x) = -G(\mu + x)$, we find that (\ref{Gnear0}) holds also for $x \in [\mu, U]$.
\end{proof}

\begin{lemma}\label{x1lemma}
The solution $x_1 = x_1(U, \mu)$ of (\ref{x1SectorIIIdef}) satisfies
  $$x_1= O(\mu \sqrt{U})$$
uniformly as $(U, \mu) \in \III_{U_0, \delta}$ tends to $(0,0)$.
\end{lemma}
\begin{proof}
Since $\mu - x_1 = O(\mu)$ in Sector III, we conclude from Lemma \ref{Gnear0lemma} that
$$0 = G(x_1) - x_1 
= U\frac{(\mu - x_1 ) \big(\ln\big(\frac{16}{\mu -x_1}\big) + 1\big)}{2 \pi^2} - x_1 + O\bigg(U|\mu - x_1|^3 \ln\frac{1}{|\mu - x_1|}\bigg).$$
Writing $X := \mu - x_1$, this equation is exactly (\ref{d0eqsmallU}) and can be solved in the same way, which implies that $X = \mu - x_1$ obeys (\ref{Xexpansion}). Consequently,
\begin{align}
 x_1 = &\; \mu\bigg(  \frac{\sqrt{U}}{\pi} 
+\frac{\ln \big(\frac{16}{\hat{\mu}}\big)-1}{2 \pi^2}U
- \frac{2 \ln \big(\frac{16}{\hat{\mu}}\big)-1}{2 \pi^3}U^{3/2}
+ O(U^2)\bigg),
\end{align}
from which the desired conclusion follows.
\end{proof}

It will be important to know where $G'(x) = -U N_0(\mu - x)$ is larger than $-1$. It is not always true that $G'(x) > -1$, because $G'(\mu) = -\infty$. However, the next lemma shows that $G'(x) > -1$ whenever $x$ stays a distance of order $e^{-1/U}$ away from $\mu$.

\begin{lemma}\label{Gprimelemma}
There are constants $U_0 > 0$ and $C_1 > 0$ such that the function $G$ defined in (\ref{GIIIdef}) satisfies
$$G'(x) \geq -\frac{1}{2 \pi^2} - C_1 U$$ 
for all $(U, \mu) \in \III_{U_0, \delta}$ and all  $x \in [-U,2U] \setminus [\mu - e^{-1/U}, \mu + e^{-1/U}]$.
\end{lemma}
\begin{proof}
Suppose that $x \in [-U,2U] \setminus [\mu - e^{-1/U}, \mu + e^{-1/U}]$. Then $\mu - x \to 0$ as $(U, \mu) \in \III_{U_0, \delta}$ tends to $(0,0)$, so using (\ref{N0expansionfirstfew}), we infer that
$$G'(x) = -U N_0(\mu - x)
= -U \frac{\ln(\frac{16}{|\mu - x|})}{2 \pi^2} + O\bigg(U |\mu - x|^2 \ln \frac{16}{|\mu - x|}\bigg).$$
In particular, since $|\mu - x| \in [e^{-1/U}, 2U]$, there are constants $C_0, C_1 > 0$ such that
$$G'(x) \geq -U \frac{\ln(\frac{16}{e^{-1/U}})}{2 \pi^2} - C_0 U^3 \ln\frac{1}{U} \geq -\frac{1}{2 \pi^2} - C_1 U,$$
whenever $(U, \mu) \in \III_{U_0, \delta}$ is sufficiently close to $(0,0)$. 
\end{proof}

For $x$ in the narrow sector of size $\sim e^{-1/U}$ around $\mu$, Lemma \ref{Gprimelemma} does not apply. In this narrow sector, we will instead use the following lemma. 

\begin{lemma}\label{Gsmallnearmulemma}
As $(U, \mu) \in \III_{U_0, \delta}$ tends to $(0,0)$, the function $G$ defined in (\ref{GIIIdef}) satisfies
 $G(x) = O(e^{-1/U})$
uniformly for $x \in [\mu - e^{-1/U}, \mu + e^{-1/U}]$.
\end{lemma}
\begin{proof}
Suppose $(U, \mu) \in \III_{U_0, \delta}$ tends to $(0,0)$ with $x \in [\mu - e^{-1/U}, \mu + e^{-1/U}]$. Using that $\mu - x = O(e^{-1/U})$, we obtain from Lemma \ref{Gnear0lemma} that
$$G(x) = U\frac{(\mu - x ) \left(1 - \ln \left(\frac{|\mu -x|}{16}\right)\right)}{2 \pi^2}
 + O(e^{-3/U})
 = O(e^{-1/U}),$$
 which is the desired assertion.
\end{proof}

The inequalities of the next lemma state that the graph of $G$ lies below the line $2x_1 - x$ for $x < x_1$, and above this line for $x > x_1$. Since solutions of (\ref{uGGu}) correspond to points $(u,v) = (u,G(u))$ which belong both to the graph of $G$ and to the reflection of this graph in the line $u=v$, we will be able to use these inequalities to prove (\ref{Gclaim}).

\begin{lemma}\label{Gbelowabovelemma}
There is a $U_0 > 0$ such that the following inequalities hold for all $(U, \mu) \in \III_{U_0, \delta}$:
\begin{subequations}\label{Gbounds}
\begin{align}\label{Gbelowline}
& G(x) < 2x_1 - x \qquad \text{for $x \in [-U,x_1)$},
	\\ \label{Gaboveline}
& G(x) > 2x_1 - x \qquad \text{for $x \in (x_1, 2x_1 + U]$}.
\end{align}
\end{subequations}
\end{lemma}
\begin{proof}
Let $\psi(x):= x + G(x) - 2x_1$. 
By Lemma \ref{x1lemma}, the solution $x_1 > 0$ satisfies $x_1 = O(\mu \sqrt{U})$ as $(U, \mu) \in \III_{U_0, \delta}$ tends to $(0,0)$. 
Using also Lemma \ref{Gsmallnearmulemma}, we see that $\psi(x) = x + O(e^{-1/U}) + O(\mu \sqrt{U}) = \mu + O(\sqrt{U}e^{-2\pi/\sqrt{U}})$ uniformly for $x \in [\mu - e^{-1/U}, \mu + e^{-1/U}]$ as $(U, \mu) \in \III_{U_0, \delta}$ tends to $(0,0)$. Since $\mu > 16 e^{-\frac{2\pi}{\sqrt{U}}}$ in Sector III, this shows that $\psi(x) > 0$, and hence that (\ref{Gaboveline}) holds, for all $x \in [\mu - e^{-1/U}, \mu + e^{-1/U}]$ whenever $(U, \mu) \in \III_{U_0, \delta}$ is sufficiently close to $(0,0)$. 
On the other hand, since $\psi'(x) = G'(x) + 1$, Lemma \ref{Gprimelemma} implies that the function $\psi$ is strictly increasing for all $x \in [-U,2U] \setminus [\mu - e^{-1/U}, \mu + e^{-1/U}]$ whenever $(U, \mu) \in \III_{U_0, \delta}$.
Since $\psi(x_1) = 0$, it follows that (\ref{Gbelowline}) holds, and that (\ref{Gaboveline}) holds for $x \in (x_1, \mu - e^{-1/U})$. Since $\psi(\mu + e^{-1/U}) > 0$ by the first half of the proof, it follows that (\ref{Gaboveline}) also holds for $x \in (\mu + e^{-1/U}, 2x_1 + U]$.
The proof is complete.
\end{proof}

In light of the discussion above (\ref{Gclaim}), the following lemma completes the proof of Proposition \ref{FsolutionIIIprop}.

\begin{lemma}
  The claim (\ref{Gclaim}) holds for any sufficiently small $U_0> 0$.
\end{lemma}
\begin{proof}
Let $u \in [-U, \mu)$ be such that $u < G(u) \leq U$. If $u\in [x_1, \mu)$, then, since $G$ is nonincreasing, $G(u) \leq G(x_1) = x_1 \leq u$, which contradicts the assumption that $u < G(u)$. Thus we must have $u\in [-U, x_1)$, and we can apply (\ref{Gbelowline}) to conclude that $G(u) < 2x_1 - u$ where $2x_1 - u \in (x_1, 2x_1 + U]$, and so, using first that $G$ is nonincreasing and then (\ref{Gaboveline}), we obtain
$$G(G(u)) \geq G(2x_1 - u) > 2x_1 - (2x_1 - u) = u,$$
which contradicts (\ref{uGGu}).
\end{proof}

\subsection{AF free energy in Sector III}
Suppose $(d_0^{\AF}, m_1^{\AF})$ is an AF mean-field solution corresponding to $(U, \mu) \in \III_{U_0,\delta}$.
Lemma \ref{limsupdeltaUlemma} implies that $d_0^{\AF} \in [0,1)$, and so, in view of (\ref{AFeq1}),
$$2\int_{0}^4 N_0(\epsilon) \big(\theta(X_+) + \theta(X_-) \big) d\epsilon \in (0,1],$$
where we recall the short-hand notation in (\ref{Deltadef})--(\ref{Xpmdef}).
Hence the sets $A_\pm$ defined in (\ref{Aplusminusdef}) satisfy $A_- = \emptyset$ and $A_+ = (b_+, 4)$ for some $b_+ \geq 0$, where $b_+ = 0$ if $d_0^{\AF} = 0$ and $b_+ \in (0,4)$ if $d_0^{\AF} \in (0,1)$. We treat these two cases in turn. 
We will show that there exists one solution of the AF mean-field equations with $m_0^{\AF} > 0$ in each case, and we will denote these solutions by $(d_{0,1}^{\AF}, m_{1,1}^{\AF})$ and $(d_{0,2}^{\AF}, m_{1,2}^{\AF})$, where $d_{0,1}^{\AF} = 0$ and $d_{0,2}^{\AF} \in (0,1)$.

\subsubsection{The case $d_0^{\AF} = 0$}
The next lemma treats the relatively easy case of $d_0^{\AF} = 0$.

\begin{lemma}[AF mean-field solution with $d_0^{\AF} = 0$ in Sector III]\label{AFsolutionIIIlemma}
There is a $U_0 > 0$ such that, if $(U, \mu) \in \III_{U_0, \delta}$, then there is a unique solution $(d_{0,1}^{\AF}, m_{1,1}^{\AF})$ of the AF mean-field equations (\ref{AFmeanfieldeqs}) satisfying $d_{0,1}^{\AF} = 0$ and $m_{1,1}^{\AF} > 0$. Moreover, $m_{1,1}^{\AF} = m_{1,1}^{\AF}(U)$ is independent of $\mu$ and, as $(U, \mu) \in \III_{U_0, \delta}$ tends to $(0,0)$,
\begin{align}\label{GAFsmallU}
\mathcal{G}_{\AF}(d_{0,1}^{\AF}, m_{1,1}^{\AF}, U, \mu)
=&-\frac{16}{\pi^2}
-\frac{512 e^{-\frac{4 \pi }{\sqrt{U}}}}{\pi \sqrt{U}}  -\frac{128 e^{-\frac{4 \pi }{\sqrt{U}}}}{\pi^2} 
+ O\bigg(\frac{e^{-\frac{8\pi}{\sqrt{U}}}}{U}\bigg),
\end{align}
where the error term is uniform with respect to $\mu$.
\end{lemma}
\begin{proof}
Assume that $d_0^{\AF} = 0$. Then the AF mean-field equations (\ref{AFeqs}) reduce to the single equation
\begin{align}\label{Delta04def}
\frac{1}{U} = \int_{0}^4 N_0(\epsilon) \frac{1}{\sqrt{\Delta^2 + \epsilon^2}} d\epsilon.
\end{align}
By \cite[Theorem 1.1]{LLa0}, the unique solution $\Delta = \Delta(U)$  of this equation satisfies
\begin{align}\label{DeltasmallU}
	\Delta(U) 
	= 32 e^{-\frac{2\pi}{\sqrt{U}}}\Bigg(1 + O\bigg(\frac{e^{-\frac{4\pi}{\sqrt{U}}}}{\sqrt{U}}\bigg)\Bigg)
	\qquad \text{as $U\downarrow 0$.}
\end{align}
Furthermore, since $d_0^{\AF} = 0$, we find from (\ref{AFHartreeFockFunction}) that
\begin{align}\label{lol3}
\mathcal{G}_{\AF}(0, m_1^{\AF}, U, \mu) = &\; \frac{U}{4} (m_1^{\AF})^2
 - \int_{\R}  N_0(\epsilon) \sqrt{\frac{U^2}{4}(m_1^{\AF})^2 + \epsilon^2} d\epsilon
= \frac{\Delta^2}{U} - 2 \varphi(\Delta),
\end{align}
where
$$\varphi(\Delta) := \int_{0}^4  N_0(\epsilon) \sqrt{\Delta^2 + \epsilon^2} d\epsilon.$$
By \cite[Proof of Proposition 1.2]{LLa0},
\begin{align}\nonumber
\varphi'(\Delta) & = \Delta \int_{0}^4 N_0(\epsilon) \frac{1}{\sqrt{\Delta^2 + \epsilon^2}} d\epsilon
	\\ \nonumber
& = \Delta\bigg(a_0 + \frac{6 (\ln \Delta )^{2} -60 \ln(2) \ln{\Delta} +\pi^2+126 (\ln 2 )^{2}}{24 \pi^2}
+ O(\Delta^2 |\ln\Delta|^2)\bigg) \quad \text{as $\Delta \downarrow 0$},
\end{align}
where (see \cite[Eq. (3.1)]{LLa0})
\begin{align}\label{a0def}
a_0 := \int_{0}^4 \bigg(N_0(\epsilon) - \frac{\ln(\frac{16}{\epsilon})}{2 \pi^2}\bigg)\frac{1}{\epsilon} d\epsilon
= \frac{(\ln 2)^2}{\pi^2}-\frac{1}{24}.
\end{align}
Integration yields
\begin{align}\nonumber
\varphi(\Delta) = 
&\; \varphi(0)
+ \frac{\Delta ^2 \ln ^2(\Delta )}{8 \pi^2}
- \frac{\Delta ^2 (1+10 \ln (2)) \ln (\Delta )}{8 \pi^2}
	\\ \label{varphiDeltaexpansion}
& + \frac{\Delta ^2 \left(1+50 \ln ^2(2)+10 \ln (2)\right)}{16 \pi^2}
+ O(\Delta^4 |\ln\Delta|^2)  \quad \text{as $\Delta \downarrow 0$}.
\end{align}
Hence, using (\ref{DeltasmallU}) and the fact that $\varphi(0) = \int_{0}^4 \epsilon N_0(\epsilon) d\epsilon = \frac{8}{\pi^2}$ by (\ref{N0epsilon8pi2}), the expansion (\ref{GAFsmallU}) follows.
\end{proof}

\subsubsection{The case $d_0^{\AF} \in (0,1)$}
We now turn to the more difficult case of $d_0^{\AF} \in (0,1)$. In this case $b_+ \in (0,4)$, and the AF mean-field equations (\ref{AFeqs}) can be written as
\begin{align}\label{AF1}
1 - d_0^{\AF} = &\; 2\int_{b_+}^4 N_0(\epsilon) d\epsilon,
	\\\label{AF2}
\frac{1}{U} = &\; \int_{b_+}^4 N_0(\epsilon)
 \frac{1}{\sqrt{\Delta^2 + \epsilon^2}} d\epsilon.
\end{align}
Moreover, since $A_+ = [b_+, 4)$ with $b_+ > 0$, the definition (\ref{Xpmdef}) of $X_+$ implies that
\begin{align}\label{bplusmurelation}
\frac{U d_0^{\AF}}{2}  + \sqrt{\Delta^2 + b_+^2} - \mu = 0.
\end{align}

Our first objective is to define a function $\Xi(b_+)$ which is such that solutions of (\ref{AF1})--(\ref{bplusmurelation}) are in one-to-one correspondence with the zeros of $\Xi$. To this end, we define $b_+^{\max} = b_+^{\max}(U) \in (0,4)$ by
\begin{align}\label{bplusmaxdef}
\frac{1}{U} = \int_{b_+^{\max}}^4  \frac{N_0(\epsilon)}{\epsilon} d\epsilon.
 \end{align}
Since the right-hand side of (\ref{bplusmaxdef}) decreases from $+\infty$ to $0$ as $b_+^{\max}$ increases from $0$ to $4$, we see that $b_+^{\max}(U)$ is well-defined for any $U > 0$.
For any fixed $U>0$, (\ref{AF2}) has a solution $\Delta \geq 0$ (which is necessarily unique) if and only if $b_+ \in [0, b_+^{\max}]$. We call this solution $\Delta_1(b_+)$. The map $b_+ \mapsto \Delta_1(b_+)$ is a continuous strictly decreasing function 
\begin{align}\label{def of Delta1}
\Delta_1:[0, b_+^{\max}] \to [0,+\infty),
\end{align}
which is smooth for $b_+ \in (0, b_+^{\max})$ and such that $\Delta_1(b_+^{\max}) = 0$. 

Let $\delta > 0$ be the fixed small number defining Sector III according to (\ref{sectorIIIdef}). 
We define $\Xi(b_+)$ for $b_+ \in [0, b_+^{\max}]$ by
\begin{align}\label{Xidef}
\Xi(b_+) := \Delta_1(b_+) - \Delta_2(b_+),
\end{align}
where the function $\Delta_2:[0, b_+^{\max}] \to [0,+\infty)$ is defined by
\begin{align}\label{Delta2def}
\Delta_2(b_+) := \sqrt{\bigg(\mu - \frac{U}{2} + U \int_{b_+}^4 N_0(\epsilon) d\epsilon\bigg)^2 - b_+^2}.
\end{align}
(For intuition behind the definition of $\Delta_{2}$, we refer the reader to the proof of Lemma \ref{Xicorrespondencelemma} below.)
For conciseness, we have suppressed the $U$-dependence of the function $\Delta_1$, and the $(U, \mu)$-dependence of the functions $\Delta_2$ and $\Xi$.
To ensure that $\Xi$ is well-defined by (\ref{Xidef}), we need to verify that the expression under the square root in (\ref{Delta2def}) is $\geq 0$ for any $b_+ \in [0, b_+^{\max}]$ for $(U, \mu)$ in Sector III. The next four lemmas are devoted to verifying this claim.

\begin{lemma}\label{intN0overepsilonlemma}
As $b_+ \downarrow 0$, it holds that
$$\int_{b_+}^4 \frac{N_0(\epsilon)}{\epsilon} d\epsilon
= \frac{(\ln(\frac{16}{b_+}))^2}{4 \pi^2}-\frac{1}{24}
+ O(b_+^2 \ln(b_+)).$$
\end{lemma}
\begin{proof}
We write
\begin{align}
\int_{b_+}^4 \frac{N_0(\epsilon)}{\epsilon} d\epsilon
= \int_{b_+}^4  \frac{\ln(\frac{16}{\epsilon})}{2 \pi^2} \frac{1}{\epsilon} d\epsilon +a_0 - E_2(b_+),
\end{align}
where $a_0$ is the constant in (\ref{a0def}) and
$$E_2(b_+) := \int_0^{b_+}  \bigg(N_0(\epsilon) - \frac{\ln(\frac{16}{\epsilon})}{2 \pi^2}\bigg)\frac{1}{\epsilon} d\epsilon$$
satisfies
\begin{align}\label{E2bound}
|E_2(b_+)| \leq C \int_0^{b_+}  \epsilon \ln\Big(\frac{16}{\epsilon}\Big) d\epsilon \leq C b_+^2 |\ln b_+|
\end{align}
for all sufficiently small $b_+$.
Since
$$\int_{b_+}^4  \frac{\ln(\frac{16}{\epsilon})}{2 \pi^2} \frac{1}{\epsilon} d\epsilon
= \frac{\ln^2\big(\frac{16}{b_+}\big)-4 (\ln 2)^2}{4 \pi^2}$$
the lemma follows. 
\end{proof}

The next lemma shows that $b_+^{\max}(U) < 16 e^{-\frac{2\pi}{\sqrt{U}}}$ for all sufficiently small $U>0$.

\begin{lemma}\label{bmaxlemma}
As $U \downarrow 0$,
\begin{align}\label{bmaxexpansion}
b_+^{\max} = \bigg(16 - \frac{2\pi}{3}\sqrt{U} + O(U)
\bigg)e^{-\frac{2\pi}{\sqrt{U}}}.
\end{align}
\end{lemma}
\begin{proof}
Lemma \ref{intN0overepsilonlemma} shows that, if $\sigma = \pm 1$ and
$$b_+ = \bigg\{16 - \frac{2\pi}{3}\sqrt{U} + \bigg(\frac{\pi^2}{72} + \sigma \delta\bigg)U \bigg\}e^{-\frac{2\pi}{\sqrt{U}}}$$ 
then
$$\int_{b_+}^4 \frac{N_0(\epsilon)}{\epsilon} d\epsilon
= \frac{1}{U} - \sigma \delta  \frac{\sqrt{U}}{16 \pi } + O(U^2)\qquad \text{as $U \downarrow 0$}.$$
In view of the definition (\ref{bplusmaxdef}) of $b_+^{\max}$, the estimate (\ref{bmaxexpansion}) follows.
\end{proof}

\begin{lemma}\label{intbplus4N0lemma}
As $b_+ \downarrow 0$, it holds that
$$\int_{b_+}^4  N_0(\epsilon) d\epsilon
= \frac{1}{2}-\frac{b_+ \big(\ln (\frac{16}{b_+})+1\big)}{2 \pi^2}
    + O\bigg(b_+^3 \ln\bigg(\frac{16}{b_+}\bigg)\bigg).$$
\end{lemma}
\begin{proof}
As $b_+ \downarrow 0$, we obtain with the help of (\ref{N0expansionfirstfew}) that
\begin{align*}
\int_{b_+}^4  N_0(\epsilon) d\epsilon
= & \int_{0}^4 N_0(\epsilon)   d\epsilon
-
\int_0^{b_+}  \frac{\ln(\frac{16}{\epsilon})}{2 \pi^2} d\epsilon
-
\int_0^{b_+} \bigg(N_0(\epsilon) - \frac{\ln(\frac{16}{\epsilon})}{2 \pi^2}\bigg) d\epsilon
	\\
= &\; \frac{1}{2}-\frac{b_+ \big(\ln(\frac{16}{b_+})+1\big)}{2 \pi^2}
 + O\bigg(\int_{0}^{b_+}  \epsilon^2 \ln\bigg(\frac{16}{\epsilon}\bigg)\bigg) d\epsilon,
\end{align*}
from which the lemma follows.
\end{proof}

The next lemma completes our proof that $\Delta_2(b_+)$ is well-defined by (\ref{Delta2def}), and hence also that $\Xi(b_+)$ is well-defined by (\ref{Xidef}), if $(U, \mu) \in \III_{U_0,\delta}$ is sufficiently close to $(0,0)$.
In the rest of this section, we use the notation $\hat{\mu}$ and $\hat{b}$ for the quantities
\begin{align}\label{def of hat quantities}
\hat{\mu} := \mu e^{\frac{2\pi}{\sqrt{U}}}, \qquad \hat{b}_+ := b_+ e^{\frac{2\pi}{\sqrt{U}}}.
\end{align}
By definition of Sector III, we have $\hat{\mu} \in [16 + \delta, 32 - \delta]$ for all $(U, \mu) \in \III_{U_0,\delta}$.

\begin{lemma}\label{Delta2welldefinedlemma}
As $U \downarrow 0$, the expression under the square root in (\ref{Delta2def}) satisfies
\begin{align}
\bigg(\mu - \frac{U}{2} + U \int_{b_+}^4 N_0(\epsilon) d\epsilon\bigg)^2 - b_+^2
= \big(\hat{\mu}^2 - \hat{b}_+^2 + O(\sqrt{U}) \big)e^{-\frac{4\pi}{\sqrt{U}}}
\end{align}
uniformly for $b_+ \in [0, b_+^{\max}]$ and $\hat{\mu} \in [16 + \delta, 32 - \delta]$.
\end{lemma}
\begin{proof}
Suppose $b_+ \in [0, b_+^{\max}]$. By Lemma \ref{bmaxlemma}, this implies that $b_+ = \hat{b}_+ e^{-\frac{2\pi}{\sqrt{U}}}$ with $\hat{b}_+ \in [0, 16]$. In particular, $b_+$ tends to $0$ as $U \downarrow 0$, so we can employ Lemma \ref{intbplus4N0lemma} to infer that
\begin{align}
\mu - \frac{U}{2} + U \int_{b_+}^4 N_0(\epsilon) d\epsilon
& = \mu - U\frac{b_+ \big(\ln (\frac{16}{b_+})+1\big)}{2 \pi^2} + O\bigg(U b_+^3 \ln\bigg(\frac{16}{b_+}\bigg)\bigg) \nonumber
	\\
& = \bigg(\hat{\mu}  - U\frac{\hat{b}_+ \big(\ln (\frac{16}{\hat{b}_+}) + \frac{2\pi}{\sqrt{U}} +1\big)}{2 \pi^2}\bigg)e^{-\frac{2\pi}{\sqrt{U}}} + O\big(U b_+^2 \big) \nonumber
	\\
& = \big(\hat{\mu}  + O(\sqrt{U}) \big)e^{-\frac{2\pi}{\sqrt{U}}}, \label{lol2}
\end{align}
uniformly for $\hat{b}_+ \in [0, 16]$ and $\hat{\mu} \in [16 + \delta, 32 - \delta]$ as $U \downarrow 0$, from which the desired assertion follows.
\end{proof}

Now that we have verified that $\Xi$ is well-defined by (\ref{Xidef}), we show that the zeros of $\Xi$ correspond to solutions $(d_0^{\AF}, m_1^{\AF}) \in (0,1) \times (0, +\infty)$ of the AF mean-field equations.

\begin{lemma}\label{Xicorrespondencelemma}
There is a $U_0 > 0$ such that, if $(U, \mu) \in \III_{U_0,\delta}$, then solutions $(d_0^{\AF}, m_1^{\AF}) \in (0,1) \times (0, +\infty)$ of the AF mean-field equations (\ref{AFmeanfieldeqs}) are in one-to-one correspondence with solutions $b_+ \in (0, b_+^{\max})$ of the equation
\begin{align}\label{Xizero}
  \Xi(b_+) = 0.
\end{align}
\end{lemma}
\begin{proof}
Suppose $(d_0^{\AF}, m_1^{\AF}) \in (0,1) \times (0, +\infty)$ satisfies (\ref{AFmeanfieldeqs}).
Since $d_0^{\AF} \in (0,1)$, we may define $b_+ \in (0,4)$ as the unique solution of (\ref{AF1}), and then $\Delta := \frac{U}{2}m_1^{\AF}$ satisfies (\ref{AF2}) and (\ref{bplusmurelation}). Since (\ref{bplusmurelation}) has the solution $\Delta > 0$, 
it follows that $b_+ < b_+^{\max}$, and that $\Delta_1(b_+) = \Delta$.
On the other hand, using (\ref{AF1}) to eliminate $d_0^{\AF}$ from (\ref{bplusmurelation}), we obtain
\begin{align}\label{Delta2eq}
\frac{U}{2} - U\int_{b_+}^4 N_0(\epsilon) d\epsilon + \sqrt{\Delta^2 + b_+^2} - \mu = 0,
\end{align}
and, solving for $\Delta$, we obtain
$$\Delta^2 = \bigg(\mu - \frac{U}{2} + U \int_{b_+}^4 N_0(\epsilon) d\epsilon\bigg)^2 - b_+^2,$$
which shows that $\Delta_2(b_+) = \Delta$. Thus, $b_+$ lies in $(0, b_+^{\max})$ and solves (\ref{Xizero}).

For the converse, assume that $b_+ \in (0, b_+^{\max})$ satisfies $\Xi(b_+) = 0$.
Then $\Delta := \Delta_1(b_+) = \Delta_2(b_+)$ is a strictly positive solution of (\ref{AF2}).
If we define $d_0^{\AF} := 1 - 2\int_{b_+}^4 N_0(\epsilon) d\epsilon$, then $d_0^{\AF} \in (0,1)$ and the equations (\ref{AF1}), (\ref{AF2}), and (\ref{bplusmurelation}) are fulfilled.
Consequently, setting $m_1^{\AF} = \frac{2}{U}\Delta$, we conclude that $(d_0^{\AF}, m_1^{\AF})$ is an AF mean-field solution in $(0,1) \times (0, +\infty)$.
\end{proof}

Our next objective is to show that the function $\Xi$ has exactly one zero if $(U, \mu) \in \III_{U_0,\delta}$ is small enough. 
To show this, we first need to derive some properties of $\Delta_1$ and $\Delta_2$. We begin with $\Delta_2$.

\begin{lemma}\label{Delta2lemma}
The function $\Delta_2$ defined in (\ref{Delta2def}) is a strictly decreasing function of $[0, b_+^{\max}]$.
Furthermore, as $U \downarrow 0$, 
\begin{align}\label{Delta2expansion}
\Delta_2(b_+)& = \big(\sqrt{\hat{\mu}^2 - \hat{b}_+^2} + O(\sqrt{U})\big)e^{-\frac{2\pi}{\sqrt{U}}}
	\\ \label{Delta2primeexpansion}
\Delta_2'(b_+) & = \frac{- \hat{b}_+}{\sqrt{\hat{\mu}^2 - \hat{b}_+^2}} + O(\sqrt{U}) + O\bigg(U \ln\bigg(\frac{16}{\hat{b}_+}\bigg) \bigg)
\end{align}
uniformly for $b_+ \in (0, b_+^{\max}]$ and $\hat{\mu} \in [16 + \delta, 32 - \delta]$.
\end{lemma}
\begin{proof}
The expansion (\ref{Delta2expansion}) is an immediate consequence of Lemma \ref{Delta2welldefinedlemma}. The asymptotic formula \eqref{lol2} implies that 
\begin{align*}
\mu - \frac{U}{2} + U \int_{b_+}^4 N_0(\epsilon) d\epsilon > 0
\end{align*}
holds for all $(U, \mu) \in \III_{U_0,\delta}$, provided that $U_{0}$ is chosen sufficiently small. Hence,
\begin{align*}
\mu - \frac{U}{2} + U \int_{b_+}^4 N_0(\epsilon) d\epsilon = \sqrt{\Delta_2(b_+)^2 + b_+^2},
\end{align*}
and differentiation of (\ref{Delta2def}) yields
$$\Delta_2'(b_+) = - \frac{U N_0(b_+) \sqrt{\Delta_2(b_+)^2 + b_+^2} + b_+}{\Delta_2(b_+)}.$$
Utilizing (\ref{Delta2expansion}) and (\ref{N0expansionfirstfew}) in this expression, we obtain
\begin{align*}
\Delta_2'(b_+) & = - \frac{U \Big(\frac{\ln(\frac{16}{b_+})}{2 \pi^2} + O(b_+^2 \ln(\frac{16}{b_+}))\Big) \big(\hat{\mu} + O(\sqrt{U})\big)e^{-\frac{2\pi}{\sqrt{U}}} + \hat{b}_+ e^{-\frac{2\pi}{\sqrt{U}}}}{\big(\sqrt{\hat{\mu}^2 - \hat{b}_+^2} + O(\sqrt{U})\big)e^{-\frac{2\pi}{\sqrt{U}}}}.
	\\
& = - \frac{U \Big(\frac{\ln(\frac{16}{\hat{b}_+}) + \frac{2\pi}{\sqrt{U}} }{2 \pi^2} + O(b_+)\Big) \big(\hat{\mu} + O(\sqrt{U})\big) + \hat{b}_+ }{\sqrt{\hat{\mu}^2 - \hat{b}_+^2} + O(\sqrt{U})}
= - \frac{\hat{b}_+ + O(\sqrt{U}) + O(U \ln(\frac{16}{\hat{b}_+}) )}{\sqrt{\hat{\mu}^2 - \hat{b}_+^2} + O(\sqrt{U})}
\end{align*}
uniformly for $b_+ \in [0, b_+^{\max}]$ and $\hat{\mu} \in [16 + \delta, 32 - \delta]$ as $U \downarrow 0$, from which (\ref{Delta2primeexpansion}) follows.
\end{proof}

We now turn to the asymptotic behavior of $\Delta_1$. We will need the following lemma.

\begin{lemma}\label{intN0oneoversqrtlemma}
As $\Delta \downarrow 0$, it holds that
\begin{align}\nonumber
\int_{x \Delta}^4  N_0(\epsilon) \frac{1}{\sqrt{\Delta^2 + \epsilon^2}} d\epsilon
=&\; \frac{(\ln \Delta)^2}{4 \pi^2}
+ \frac{\arcsinh(x)-\ln (32)}{2 \pi^2} \ln \Delta
	\\ \label{intN0oneoversqrt}
& 
+ O\big(1 + (\ln(1 + x))^2\big)
\end{align}
uniformly for $x$ in the sector $0 \leq x \leq 4/\Delta$.
\end{lemma}
\begin{proof}
We write the integral in (\ref{intN0oneoversqrt}) as
\begin{align}\label{intsplit}
\int_{x \Delta}^4 N_0(\epsilon)\frac{1}{\sqrt{\Delta^2 + \epsilon^2}} d\epsilon
= \int_{x \Delta}^4 \frac{\ln(\frac{16}{\epsilon})}{2 \pi^2} \frac{1}{\sqrt{\Delta^2 + \epsilon^2}} d\epsilon
+ E(\Delta, x),
\end{align}
where 
$$E(\Delta, x) := \int_{x \Delta}^4  \bigg(N_0(\epsilon) - \frac{\ln(\frac{16}{\epsilon})}{2 \pi^2}\bigg)\frac{1}{\sqrt{\Delta^2 + \epsilon^2}} d\epsilon$$
in view of (\ref{N0expansionfirstfew}) satisfies
\begin{align}\label{Ebound}
|E(\Delta, x)| \leq C \int_0^4 \epsilon^2 \ln(\frac{16}{\epsilon})\frac{1}{\epsilon} d\epsilon \leq C
\end{align}
for all $\Delta \geq 0$ and all $x \geq 0$ such that $x \Delta \leq 4$.
For $\epsilon > 0$ and $\Delta > 0$, we have the primitive function
\begin{align*}
& \frac{1}{16 \pi^2} \frac{d}{d\epsilon} \bigg\{
-2 \bigg(\ln \frac{1 + \sqrt{\frac{\Delta^2}{\epsilon^2}+1}}{2}\bigg)^2
-\bigg(\ln \frac{\Delta^2}{\epsilon^2}\bigg)^2
+4 \ln \bigg(\frac{1 + \sqrt{\frac{\Delta^2}{\epsilon^2}+1}}{2}\bigg) \ln \bigg(\frac{\Delta
^2}{\epsilon^2}\bigg)
	\\
& +4 \ln \bigg(\frac{256}{\Delta^2}\bigg) \arctanh\bigg(\frac{1}{\sqrt{\frac{\Delta^2}{\epsilon^2}+1}}\bigg)
+4  \Li_2\bigg(\frac{1 - \sqrt{\frac{\Delta^2}{\epsilon^2}+1}}{2}\bigg)\bigg\}
= \frac{\ln(\frac{16}{\epsilon})}{2 \pi^2} 
\frac{1}{\sqrt{\Delta^2 + \epsilon^2}},
\end{align*}
which leads to the following expression:
$$\int_{x \Delta}^4 \frac{\ln(\frac{16}{\epsilon})}{2 \pi^2} \frac{1}{\sqrt{\Delta^2 + \epsilon^2}} d\epsilon
= I_1(x) + I_2(\Delta)
-\frac{\ln \left(\frac{16}{\Delta}\right) \arcsinh(x)}{2 \pi^2},
$$
where
\begin{align*}
I_1(x) :=&\; \frac{4 Y_1(x) \ln (x)+2 \ln ^2(x)+Y_1(x)^2-2 Y_2(x)}{8 \pi^2},
	\\
I_2(\Delta) := &\; \frac{1}{16 \pi^2}\bigg\{-\ln ^2\left(\frac{16}{\Delta^2}\right)+2 (\ln (8)-Y_3(\Delta)) (-2 \ln (\Delta^2)+Y_3(\Delta)+\ln (32))
	\\
& +2 \ln \left(\frac{256}{\Delta^2}\right) (Y_3(\Delta)-Y_4(\Delta))+4 Y_5(\Delta)\bigg\},
   	\\
Y_1(x) := &\; \ln \left(\frac{1}{2} \left(1 + \sqrt{\frac{1}{x^2}+1}\right)\right), 
\qquad Y_2(x) := \text{Li}_2\left(\frac{1}{2} \left(1-\sqrt{\frac{1}{x^2}+1}\right)\right),
   	\\
Y_3(\Delta) := &\; \ln \left(\sqrt{\Delta^2+16}+4\right),
\qquad Y_4(\Delta) := \ln \left(\sqrt{\Delta^2+16}-4\right),
   	\\
Y_5(\Delta) := &\; \text{Li}_2\left(\frac{1}{2}-\frac{\sqrt{\Delta^2+16}}{8}\right).
   \end{align*}
As $\Delta \downarrow 0$,
   \begin{align*}
Y_3(\Delta) := &\; \ln(8) + \frac{\Delta^2}{64} + O(\Delta^4),
\qquad Y_4(\Delta) := \ln \left(\frac{\Delta^2}{8}\right)-\frac{\Delta^2}{64} + O(\Delta^4),
   	\\
Y_5(\Delta) := &\; -\frac{\Delta^2}{64} + O(\Delta^4).
   \end{align*}
We obtain
\begin{align*}
\int_{x \Delta}^4 \frac{\ln(\frac{16}{\epsilon})}{2 \pi^2} \frac{1}{\sqrt{\Delta^2 + \epsilon^2}} d\epsilon
= &\; \frac{(\ln \Delta)^2}{4 \pi^2}
+ \frac{\left(\arcsinh(x)-\ln (32)\right)}{2 \pi^2} \ln \Delta
	\\
& + \frac{\ln (2) \left(\ln (32)-2 \arcsinh(x)\right)}{\pi^2}
+ I_1(x) + O(\Delta^2 \ln\Delta),
\end{align*}
as $\Delta \downarrow 0$ uniformly for $x \geq 0$ such that $x \Delta \leq 4$.
Substituting this expansion into (\ref{intsplit}) and using that
$$\frac{\ln (2) \left(\ln (32)-2 \arcsinh(x)\right)}{\pi^2}
+ I_1(x) = \begin{cases} O(1) & \text{as $x \downarrow 0$}, \\
O((\ln x)^2) & \text{as $x \to +\infty$},
\end{cases}$$
we arrive at (\ref{intN0oneoversqrt}) after recalling (\ref{Ebound}).
\end{proof}

We will only need uniform asymptotics of $\Delta_1$ for $\hat{b}_+$ in compact subsets of $(0,16)$, so we will restrict ourselves to this case.

\begin{lemma}\label{Delta1lemma}
As $U \downarrow 0$, the function $\Delta_1$ satisfies
\begin{align}\label{Delta1estimate}
  \Delta_1(b_+) = \Big(8 \sqrt{16 - \hat{b}_+} + O(\sqrt{U})\Big) e^{-\frac{2\pi}{\sqrt{U}}}
\end{align}
uniformly for $\hat{b}_+$ in compact subsets of $(0, 16)$.
\end{lemma}
\begin{proof}
Recall that $\Delta_1(b_+)$ is defined for any $b_+ \in [0, b_+^{\max}]$ as the unique positive solution of
\begin{align}\label{AFequation2Delta1}
\frac{1}{U} = \int_{b_+}^4 N_0(\epsilon)
 \frac{1}{\sqrt{\Delta_1(b_+)^2 + \epsilon^2}} d\epsilon.
\end{align}
Let $K$ be a compact subset of $(0,16)$. 
According to \eqref{def of Delta1} and Lemma \ref{bmaxlemma}, $\Delta_1(b_+)$ is defined for all $b_+ = \hat{b}_+ e^{-\frac{2\pi}{\sqrt{U}}}$ with $\hat{b}_+ \in K$ whenever $U$ is small enough. 
 
Let $C_1 > 0$ be a constant and define
\begin{align}\label{Deltapmdef}
\Delta_\pm = \Delta_\pm(U, \hat{b}_+) := \Big(8\sqrt{16 - \hat{b}_+}  \pm C_1 \sqrt{U}\Big) e^{-\frac{2\pi}{\sqrt{U}}}.
\end{align}
It is enough to show that if $C_1$ is large enough, then 
\begin{align}\label{Deltapminequalities}
\int_{b_+}^4 N_0(\epsilon)
 \frac{1}{\sqrt{\Delta_+^2 + \epsilon^2}} d\epsilon
 < \frac{1}{U} < \int_{b_+}^4 N_0(\epsilon)
 \frac{1}{\sqrt{\Delta_-^2 + \epsilon^2}} d\epsilon
 \end{align}
for all small enough $U > 0$ and all $\hat{b}_+ \in K$. Indeed, if this is true, then it follows from (\ref{AFequation2Delta1}) that $\Delta_- < \Delta_1(b_+) < \Delta_+$, so that (\ref{Delta1estimate}) follows.

Letting $x = b_+/\Delta_+$, we can write 
\begin{align}\label{1UlnDelta2}
\int_{b_+}^4 N_0(\epsilon)
 \frac{1}{\sqrt{\Delta_+^2 + \epsilon^2}} d\epsilon
 = \frac{(\ln \Delta_+)^2}{4 \pi^2}
+ \frac{\arcsinh(x)-\ln (32)}{2 \pi^2} \ln \Delta_+
+ E_3(\Delta_+, x)
 \end{align}
where the function $E_3(\Delta_+, x)$ according to Lemma \ref{intN0oneoversqrtlemma} obeys the estimate
$$|E_3(\Delta_+, x)| \leq C\big(1 + (\ln(1 + x))^2\big)$$
for all $\Delta_+ > 0$ and all $0 \leq x \leq 4/\Delta_+$. 
As a result of (\ref{Deltapmdef}) and the assumption $\hat{b}_+ \in K$, $x = b_+/\Delta_+$ belongs to a compact subset of $(0, +\infty)$ for all small enough $U > 0$; in particular, $E_3(\Delta_+, x)$ remains uniformly bounded for all $\hat{b}_+ \in K$ and all small enough $U > 0$.
Moreover, (\ref{1UlnDelta2}) yields after a long calculation that
 \begin{align}\label{1Uln}
\int_{b_+}^4 N_0(\epsilon)
 \frac{1}{\sqrt{\Delta_+^2 + \epsilon^2}} d\epsilon
 =  \frac{1}{U} + Q(b_+, \Delta_+)
 + O(\sqrt{U}) \qquad \text{as $U \downarrow 0$}
 \end{align}
uniformly for $\hat{b}_+ \in K$, where
$$Q(b_+, \Delta_+) := -\frac{\sqrt{16-\hat{b}_+}}{4 \pi  (32 - \hat{b}_+)} C_1
-\frac{\ln^2(64 (16 - \hat{b}_+))}{16 \pi^2}+E_3(\Delta_+, x).$$
Since the error term $E_3(\Delta_+, x)$ is uniformly bounded, it follows that the first inequality in (\ref{Deltapminequalities}) holds if we choose $C_1$ large enough; the proof of the second inequality in (\ref{Deltapminequalities}) is analogous. 
\end{proof}

Our next lemma, whose proof is given in Appendix \ref{Delta1primeapp}, shows that the asymptotic formula (\ref{Delta1lemma}) can be differentiated without increasing the error term, as long as $\hat{b}_+$ stays in a compact subset of $(0, 16)$. 

\begin{lemma}\label{Delta1primelemma}
As $U \downarrow 0$, the derivative of the function $\Delta_1$ satisfies
\begin{align}
  \Delta_1'(b_+) = - \frac{4}{\sqrt{16 -\hat{b}_+}} + O(\sqrt{U})
  \end{align}
uniformly for $\hat{b}_+$ in compact subsets of $(0, 16)$.
\end{lemma}

We are now ready to show that $\Xi = 0$ has exactly one solution in $(0, b_+^{\max})$.

\begin{lemma}\label{Xionesolutionlemma}
There is a $U_0 > 0$ such that, if $(U, \mu) \in \III_{U_0,\delta}$, then the equation $\Xi(b_+) = 0$ has exactly one solution in $(0, b_+^{\max})$.
\end{lemma}
\begin{proof}
Let $\delta_1 > 0$ be such that $\delta_1 < \frac{\delta  (\delta +32)}{64}$. We will show the following three claims from which the lemma immediately follows: for all small enough $(U, \mu) \in \III_{U_0,\delta}$, it holds that
\begin{align}\label{Xipositive}
 & \Xi(b_+) > 0 \quad \text{for all $b_+ \in [0, \delta_1e^{-\frac{2\pi}{\sqrt{U}}}]$},
  	\\ \label{Xinegative}
 &  \Xi(b_+) < 0 \quad \text{for all $b_+ \in [(16-\delta_1)e^{-\frac{2\pi}{\sqrt{U}}}, b_+^{\max} ]$},
   	\\ \label{Xiprimenegative}
 &  \Xi'(b_+) < 0 \quad \text{for all $b_+ \in [\delta_1e^{-\frac{2\pi}{\sqrt{U}}}, (16-\delta_1)e^{-\frac{2\pi}{\sqrt{U}}}]$}.
 \end{align} 

By Lemma \ref{Delta2lemma}, $\Delta_2$ is a strictly decreasing function of $b_{+}\in [0, b_+^{\max}]$, and the same is true of $\Delta_1$ by Lemma \ref{Delta1primelemma}. It results that, if $b_+ \in [0, \delta_1e^{-\frac{2\pi}{\sqrt{U}}}]$, then
\begin{align*}
  \Xi(b_+) = \Delta_1(b_+) - \Delta_2(b_+) \geq \Delta_1(\delta_1 e^{-\frac{2\pi}{\sqrt{U}}}) - \Delta_2(0).
\end{align*}  
Using that $\Delta_2(0) = \mu$ and the asymptotic formula for $\Delta_1$ of Lemma \ref{Delta1lemma}, we deduce that there is a $C>0$ such that
\begin{align*}
  \Xi(b_+) \geq \Big(8 \sqrt{16 - \delta_1} - \hat{\mu} - C\sqrt{U}\Big) e^{-\frac{2\pi}{\sqrt{U}}} 
\end{align*}  
for all small enough $(U, \mu) \in \III_{U_0,\delta}$ and all $b_+ \in [0, \delta_1e^{-\frac{2\pi}{\sqrt{U}}}]$. This proves (\ref{Xipositive}).

Similarly, if $b_+ \in [(16-\delta_1)e^{-\frac{2\pi}{\sqrt{U}}}, b_+^{\max} ]$, then
\begin{align*}
  \Xi(b_+) = \Delta_1(b_+) - \Delta_2(b_+) \leq \Delta_1((16-\delta_1) e^{-\frac{2\pi}{\sqrt{U}}}) - \Delta_2(b_+^{\max}).
\end{align*}  
By Lemma \ref{Delta2lemma} and Lemma \ref{bmaxlemma}, we have
\begin{align*}
\Delta_2(b_+^{\max}) 
& = \bigg(\sqrt{\hat{\mu}^2 - \Big(16 - \frac{2\pi}{3}\sqrt{U} + O(U)\Big)^2} + O(\sqrt{U})\bigg)e^{-\frac{2\pi}{\sqrt{U}}}
	\\
& = \big(\sqrt{\hat{\mu}^2 - 256} + O(\sqrt{U})\big)e^{-\frac{2\pi}{\sqrt{U}}}.
\end{align*}
Consequently, using also the asymptotic formula for $\Delta_1$ of Lemma \ref{Delta1lemma}, we conclude that there is a $C>0$ such that
\begin{align*}
  \Xi(b_+) \leq \Big(8 \sqrt{\delta_1} - \sqrt{\hat{\mu}^2 - 256} + C\sqrt{U}\Big) e^{-\frac{2\pi}{\sqrt{U}}} 
\end{align*}  
for all small enough $(U, \mu) \in \III_{U_0,\delta}$ and all $b_+ \in [(16-\delta_1)e^{-\frac{2\pi}{\sqrt{U}}}, b_+^{\max} ]$. Since the assumptions $\hat{\mu} \geq 16 + \delta$ and $\delta_1 < \frac{\delta  (\delta +32)}{64}$ ensure that $8 \sqrt{\delta_1} - \sqrt{\hat{\mu}^2 - 256} < 0$, this proves (\ref{Xinegative}).

By Lemma \ref{Delta2lemma} and Lemma \ref{Delta1primelemma}, we have
\begin{align}
  \Xi'(b_+) = - \frac{4}{\sqrt{16 -\hat{b}_+}} + \frac{\hat{b}_+}{\sqrt{\hat{\mu}^2 - \hat{b}_+^2}} + O(\sqrt{U}) \qquad \text{as $U \downarrow 0$}
  \end{align}
uniformly for $b_+ \in [\delta_1e^{-\frac{2\pi}{\sqrt{U}}}, (16-\delta_1)e^{-\frac{2\pi}{\sqrt{U}}}]$ and $\hat{\mu} \in [16 + \delta, 32 - \delta]$. Since $\hat{\mu} > 16$, we have
$$- \frac{4}{\sqrt{16 -\hat{b}_+}} + \frac{\hat{b}_+}{\sqrt{\hat{\mu}^2 - \hat{b}_+^2}} 
< - \frac{1}{\sqrt{16 -\hat{b}_+}}\bigg(4 - \frac{\hat{b}_+}{\sqrt{16 + \hat{b}_+}}\bigg).$$
Since $\frac{\hat{b}_+}{\sqrt{16 + \hat{b}_+}} \leq 2\sqrt{2}$ for $b_+ \in [0,16]$, (\ref{Xiprimenegative}) follows.
\end{proof}

Lemma \ref{Xicorrespondencelemma} and Lemma \ref{Xionesolutionlemma} imply that the AF mean-field equations (\ref{AFmeanfieldeqs}) have a unique solution $(d_{0,2}^{\AF}, m_{1,2}^{\AF})$ in $(0,1) \times (0, +\infty)$, and that this solution satisfies (\ref{AF1})--(\ref{bplusmurelation}) with $\Delta := \frac{U}{2} m_{1,2}^{\AF}$ and with $b_+$ given by $b_+^{\AF}$, where $b_+^{\AF}$ is the unique solution of the equation $\Xi(b_+) = 0$.
Furthermore, since $A_- = \emptyset$ and $A_+ = (b_+^{\AF}, 4)$, (\ref{AFHartreeFockFunction}) gives
\begin{align*}
\mathcal{G}_{\AF}(d_{0,2}^{\AF}, m_{1,2}^{\AF}, U, \mu) = &\; \frac{U}{4}  \big((m_{1,2}^{\AF})^2 - (d_{0,2}^{\AF})^2\big)
+ \frac{U}{2} d_{0,2}^{\AF} - \mu 
	\\
& - 2 \int_{b_+^{\AF}}^4 N_0(\epsilon) \bigg(\frac{U}{2} d_{0,2}^{\AF} - \mu + \sqrt{\frac{U^2}{4}(m_{1,2}^{\AF})^2 + \epsilon^2}\bigg) d\epsilon. 
\end{align*}
Using (\ref{AF1}), this can be rewritten in terms of $\Delta$ as
\begin{align}\label{GAFIII}
\mathcal{G}_{\AF}(d_{0,2}^{\AF}, m_{1,2}^{\AF}, U, \mu) = &\; \frac{\Delta^2}{U} 
+ \frac{U}{4}  (d_{0,2}^{\AF})^2
- \mu d_{0,2}^{\AF}
  - 2\int_{b_+^{\AF}}^4 N_0(\epsilon) \sqrt{\Delta^2 + \epsilon^2} d\epsilon.
\end{align}
The integral in (\ref{GAFIII}) can be estimated with the help of the following lemma.

\begin{lemma}\label{intN0sqrtlemma}
As $\Delta \downarrow 0$, it holds that
\begin{align}\nonumber
\int_{x \Delta}^4 N_0(\epsilon) \sqrt{\Delta^2 + \epsilon^2} d\epsilon
= &\; \frac{8}{\pi^2}
+ \frac{\Delta^2 (\ln \Delta)^2}{8 \pi^2}
	\\\nonumber
& + \frac{2x\sqrt{1 + x^2} + 2 \arcsinh(x)  - 1 - 10 \ln{2}}{8 \pi^2} \Delta^2 \ln \Delta
 	\\ \label{intN0sqrt}
& + O\big(\Delta^2(1 + x^2 \ln(1 + x)\big)
\end{align}
uniformly for $x$ in the sector $0 \leq x \leq 4/\Delta$.
\end{lemma}
\begin{proof}
We will derive the expansion in (\ref{intN0sqrt}) by integrating the expansion of Lemma \ref{intN0oneoversqrtlemma}. For $\Delta > 0$, we have
$$\frac{d}{d\Delta}\int_{x \Delta}^4  N_0(\epsilon) \sqrt{\Delta^2 + \epsilon^2} d\epsilon
= \int_{x \Delta}^4  N_0(\epsilon) \frac{\Delta}{\sqrt{\Delta^2 + \epsilon^2}} d\epsilon
- x N_0(x \Delta) \Delta \sqrt{1 + x^2}.$$
Substituting in (\ref{intN0oneoversqrt}), using the expansion (\ref{N0expansionfirstfew}), 
and integrating with respect to $\Delta$, we obtain, as $\Delta \downarrow 0$,
\begin{align*}
\int_{x \Delta}^4  N_0(\epsilon) \sqrt{\Delta^2 + \epsilon^2} d\epsilon
= &\; c_0
+ \frac{\Delta^2 (\ln \Delta)^2}{8 \pi^2}
+ \frac{2x\sqrt{1 + x^2} + 2 \arcsinh(x)  - 1 - 10 \ln{2}}{8 \pi^2} \Delta^2 \ln \Delta
 	\\
& + O\big(\Delta^2(1 + x^2 \ln(1 + x) \big)
\end{align*}
uniformly for $x$ in the sector $0 \leq x \leq 4/\Delta$, where $c_0$ is a constant. We determine $c_0$ by taking the limit $\Delta \downarrow 0$, which yields, using (\ref{N0epsilon8pi2}),
$$c_0 = \int_0^4  N_0(\epsilon) \epsilon d\epsilon = \frac{8}{\pi^2},$$
and then (\ref{intN0sqrt}) follows.
\end{proof}

We are now ready to compute the expansion of $\mathcal{G}_{\AF}(d_{0,2}^{\AF}, m_{1,2}^{\AF}, U, \mu)$.

\begin{lemma}\label{AFIIIlocalmaxlemma}
As $(U, \mu) \in \III_{U_0, \delta}$ tends to $(0,0)$,
\begin{align}\label{GAFsmallU2}
\mathcal{G}_{\AF}(d_{0,2}^{\AF}, m_{1,2}^{\AF}, U, \mu)
=&-\frac{16}{\pi^2} 
- \frac{512 e^{-\frac{4\pi}{\sqrt{U}}}}{\pi \sqrt{U}} 
\bigg(1 - \frac{(32 - \hat{\mu})^2}{512}\bigg)
+ O\big(e^{-\frac{4\pi}{\sqrt{U}}}\big),
\end{align}
where the error term is uniform with respect to $\mu$.
\end{lemma}
\begin{proof}
Utilizing Lemma \ref{intN0sqrtlemma} in (\ref{GAFIII}), we find
\begin{align}\nonumber
\mathcal{G}_{\AF}(d_{0,2}^{\AF}, m_{1,2}^{\AF}, U, \mu)
=&\; \frac{\Delta^2}{U} + \frac{U (d_{0,2}^{\AF})^2}{4} - \mu d_{0,2}^{\AF}
 - \frac{16}{\pi^2} 
- \frac{\Delta^2 (\ln \Delta)^2}{4 \pi^2}
 	\\   \nonumber
&- \frac{2x\sqrt{1 + x^2} + 2 \arcsinh(x)  - 1 - 10 \ln{2}}{4 \pi^2} \Delta^2 \ln \Delta
 	\\ \label{GAF2a}
& + O\big(\Delta^2(1 + x^2 \ln(1 + x)\big),
\end{align}
where $x := b_+^{\AF}/\Delta$.
From Lemma \ref{Delta2lemma} and Lemma \ref{Delta1lemma}, together with \eqref{Xidef} and Lemma \ref{Xionesolutionlemma}, we infer that $b_+^{\AF} = \hat{b}_+^{\AF} e^{-\frac{2\pi}{\sqrt{U}}}$ obeys the relation
$$0 = \Delta_1(b_+^{\AF})^2 - \Delta_2(b_+^{\AF})^2
= \Big(64 \big(16 - \hat{b}_+^{\AF}\big) - \big(\hat{\mu}^2 - (\hat{b}_+^{\AF})^2\big)  + O(\sqrt{U})\Big) e^{-\frac{4\pi}{\sqrt{U}}},$$
which implies that
\begin{align}\label{bplusAFexpansion}
\hat{b}_+^{\AF} = 32 - \hat{\mu} + O(\sqrt{U})
\end{align}
uniformly as $(U, \mu) \in \III_{U_0, \delta}$ tends to $(0,0)$.
On the other hand, by Lemma \ref{Delta2lemma} and Lemma \ref{Delta1lemma}, $\Delta = \hat{\Delta} e^{-\frac{2\pi}{\sqrt{U}}}$ where 
\begin{align}\label{DeltaAFexpansion}
\hat{\Delta} = \sqrt{\hat{\mu}^2 - (\hat{b}_+^{\AF})^2} + O(\sqrt{U})
= 8 \sqrt{16 - \hat{b}_+^{\AF}} + O(\sqrt{U})
\end{align}
uniformly as $(U, \mu) \in \III_{U_0, \delta}$ tends to $(0,0)$.
It follows from (\ref{bplusAFexpansion}) and (\ref{DeltaAFexpansion}) that 
\begin{align}\label{xAFexpansion}
x = \frac{\hat{b}_+^{\AF}}{\hat{\Delta}} = \frac{32 - \hat{\mu}}{8 \sqrt{16 - \hat{b}_+^{\AF}}} + O(\sqrt{U}).
\end{align}
In particular, the error term in (\ref{GAF2a}) can be replaced by $O\big(e^{-\frac{4\pi}{\sqrt{U}}}\big)$.
Moreover, by Lemma \ref{intbplus4N0lemma}, \eqref{AF1}, and (\ref{bplusAFexpansion}),
\begin{align}\label{d02AFexpansion}
d_{0,2}^{\AF} = 1- 2\int_{b_+^{\AF}}^{4}  N_0(\epsilon) d\epsilon
= \bigg(\frac{2\hat{b}_+^{\AF}}{\pi \sqrt{U}} + O(1)\bigg) e^{-\frac{2\pi}{\sqrt{U}}}
\end{align}
uniformly as $(U, \mu) \in \III_{U_0, \delta}$ tends to $(0,0)$.

Observing that the terms of order $e^{-\frac{4\pi}{\sqrt{U}}}/U$ cancel, we can write (\ref{GAF2a}) as 
\begin{align*}
\mathcal{G}_{\AF}(d_{0,2}^{\AF}&, m_{1,2}^{\AF}, U, \mu)
=  - \frac{16}{\pi^2} 
+ e^{-\frac{4\pi}{\sqrt{U}}}\bigg(\frac{U (\hat{d}_{0,2}^{\AF})^2}{4} - \hat{\mu} \hat{d}_{0,2}^{\AF}
- \frac{\hat{\Delta}^2 (\ln \hat{\Delta} )^2}{4 \pi^2}
+ \frac{\hat{\Delta}^2 \ln(\hat{\Delta}) \frac{2\pi}{\sqrt{U}}}{2 \pi^2}
 	\\  
&- \frac{2x\sqrt{1 + x^2} + 2 \arcsinh(x)  - 1 - 10 \ln{2}}{4 \pi^2} \hat{\Delta}^2 \bigg(\ln(\hat{\Delta}) - \frac{2\pi}{\sqrt{U}}\bigg)\bigg)
 + O\big(e^{-\frac{4\pi}{\sqrt{U}}}\big),
\end{align*}
where $d_{0,2}^{\AF} = \hat{d}_{0,2}^{\AF} e^{-\frac{2\pi}{\sqrt{U}}}$.
Substitution of (\ref{bplusAFexpansion}), (\ref{DeltaAFexpansion}), (\ref{xAFexpansion}), and (\ref{d02AFexpansion}) into the above formula gives (\ref{GAFsmallU2}).
\end{proof}

\begin{lemma}[Asymptotics of $\mathcal{F}_{\AF}$ in Sector III]\label{FAFSectorIIIlemma}
If $U_0 > 0$ is small enough, then the AF free energy $\mathcal{F}_{\AF}(U, \mu) = \mathcal{F}_{\AF}(U)$ is independent of $\mu$ for $(U, \mu) \in \III_{U_0, \delta}$, and $\mathcal{F}_{\AF}(U)$ depends smoothly on $U \in (0,U_0]$.
Furthermore, $\mathcal{F}_{\AF}(U)$ enjoys the following asymptotics as $(U, \mu) \in \III_{U_0, \delta}$ tends to $(0,0)$:
 \begin{align}\label{FAFexpansionSectorIII}
\mathcal{F}_{\AF}(U)
=&-\frac{16}{\pi^2}
-\frac{512 e^{-\frac{4 \pi }{\sqrt{U}}}}{\pi \sqrt{U}}  -\frac{128 e^{-\frac{4 \pi }{\sqrt{U}}}}{\pi^2} 
+ O\bigg(\frac{e^{-\frac{8\pi}{\sqrt{U}}}}{U}\bigg),
\end{align}
where the error term is uniform with respect to $\mu$.
\end{lemma}
\begin{proof}
Comparing (\ref{GAFsmallU}) and (\ref{GAFsmallU2}), we see that, decreasing $U_0  > 0$ if necessary, the solution $(d_{0,1}^{\AF}, m_{1,1}^{\AF})$ has lower free energy than the solution $(d_{0,2}^{\AF}, m_{1,2}^{\AF})$ for all $(U, \mu) \in \III_{U_0, \delta}$, i.e.,
$$\mathcal{G}_{\AF}(d_{0,1}^{\AF}, m_{1,1}^{\AF}, U, \mu) < \mathcal{G}_{\AF}(d_{0,2}^{\AF}, m_{1,2}^{\AF}, U, \mu).$$
Since we have shown that there are no other AF mean-field solutions in Sector III, (\ref{freeenergiesAF}) gives $\mathcal{F}_{\AF}(U, \mu) = \mathcal{G}_{\AF}(d_{0,1}^{\AF}, m_{1,1}^{\AF}, U, \mu)$, and the desired assertion follows from (\ref{GAFsmallU}). 
Since $d_{0,1}^{\AF} = 0$ and since $m_{1,1}^{\AF}(U)$ depends smoothly on $U \in (0,U_0]$ as a consequence of (\ref{Delta04def}), smoothness of $\mathcal{F}_{\AF}(U)$ follows from \eqref{lol3}.
\end{proof}

\subsection{Final steps}
By Proposition \ref{FsolutionIIIprop}, the F mean-field equations (\ref{Fmeanfieldeqs}) have no solution in $\R \times (0, +\infty)$ in Sector III, so we only have to consider the P and AF free energies. According to Lemma \ref{FAFSectorIIIlemma}, $\mathcal{F}_{\AF}(U)$ is independent of $\mu$ in Sector III. Moreover, $\mathcal{F}_{\mathrm{P}}(U, \mu)$ is a strictly decreasing function of $\mu > 0$ by (\ref{dFPdmuequalsd0P}) (because, shrinking $U_{0}>0$ if necessary, $d_{0}^{\mathrm{P}}>0$ everywhere in Sector III by Lemma \ref{FPsectorIIIlemma}). 
On the other hand, the asymptotic formulas (\ref{FPexpansionSectorIII}) and (\ref{FAFexpansionSectorIII}) for $\mathcal{F}_{\mathrm{P}}(U, \mu)$ and $\mathcal{F}_{\AF}(U)$ imply that
$$\mathcal{F}_{\mathrm{P}}(U, \mu) - \mathcal{F}_{\AF}(U) = (512 - \hat{\mu}^2) \frac{e^{-\frac{4 \pi }{\sqrt{U}}}}{\pi} \bigg(\frac{1}{\sqrt{U}} + O(1)\bigg)$$
as $(U, \mu) \in \III_{U_0, \delta}$ tends to $(0,0)$. 
We conclude that if $\delta \in (0,16\sqrt{2}-16)$, then for any small enough $U > 0$, the function $\mu \mapsto \mathcal{F}_{\mathrm{P}}(U, \mu) - \mathcal{F}_{\AF}(U)$ is strictly positive for $\mu = (16 + \delta)e^{-\frac{2 \pi }{\sqrt{U}}}$, strictly negative for $\mu = (32 - \delta)e^{-\frac{2 \pi }{\sqrt{U}}}$, and strictly decreasing between these two points. 
Consequently, decreasing $U_0 > 0$ if necessary, there is a unique function $\mu_{\III}(U)$ such that for $(U, \mu) \in \III_{U_0, \delta}$ it holds that
$$\begin{cases}
\mathcal{F}_{\AF}(U, \mu) < \mathcal{F}_{\mathrm{P}}(U, \mu) & \text{if $\mu < \mu_{\III}(U)$},	\\ 
\mathcal{F}_{\AF}(U, \mu) = \mathcal{F}_{\mathrm{P}}(U, \mu) & \text{if $\mu = \mu_{\III}(U)$}, \\ 
\mathcal{F}_{\AF}(U, \mu) > \mathcal{F}_{\mathrm{P}}(U, \mu) & \text{if $\mu > \mu_{\III}(U)$}.
\end{cases}
$$
Smoothness of $U \mapsto \mu_{\III}(U)$ is a consequence of the implicit function theorem and the smoothness of
$\mathcal{F}_{\mathrm{P}}$ and $\mathcal{F}_{\AF}$ established in Lemma \ref{Plemma} and Lemma \ref{FAFSectorIIIlemma}.
Equations (\ref{FPexpansionSectorIII}) and (\ref{FAFexpansionSectorIII}) imply that $\hat{\mu}_{\III} = \hat{\mu}_{\III}(U) := \mu_{\III}(U)e^{\frac{2\pi}{\sqrt{U}}}$ obeys
\begin{align*}
 \hat{\mu}_{III}^2 \bigg\{ 
&-\frac{1}{\pi  \sqrt{U}}
 + \frac{1-2 \ln (\frac{16}{\hat{\mu}_{III}})}{4 \pi^2}
+ \frac{\ln (\frac{16}{\hat{\mu}_{III}})}{\pi^3} \sqrt{U}
+ \frac{(\ln (\frac{16}{\hat{\mu}_{III}})-4) \ln
   (\frac{16}{\hat{\mu}_{III}})+1}{4 \pi^4} U
	\\ 
& + \frac{3 (4-3 \ln (\frac{16}{\hat{\mu}_{III}})) \ln
   (\frac{16}{\hat{\mu}_{III}})-5}{12 \pi^5} U^{3/2}
   + O(U^2)
   \bigg\}
   = 
-\frac{512 }{\pi \sqrt{U}}  -\frac{128 }{\pi^2} \qquad \text{as $U \downarrow 0$},
\end{align*}
from which the expansion (\ref{muIIIexpansion}) follows.
Furthermore, substituting (\ref{muIIIexpansion}) into a higher-order version of (\ref{d0PexpansionSectorIII}), we obtain the formula (\ref{d0PatmuIIIexpansion}) for $d_0^{\mathrm{P}}$.
Equation (\ref{freeenergiesatmuIII}) follows immediately from (\ref{FAFexpansionSectorIII}).
This completes the proof of Theorem \ref{sectorIIIth}.

\section{Conclusions}\label{conclusionssec}
We have introduced a method that allows Hartree--Fock phase diagrams for various Hubbard-like models to be constructed analytically. As an explicit example, we have considered the 2D Hubbard model at zero temperature. We have shown that the structure of the phase diagram of this model in Hartree--Fock theory restricted to P, F, and AF states can be obtained analytically using mathematically rigorous techniques. Our results are illustrated in Figure \ref{phasediagramfig} where the analytically computed phase boundaries are superimposed on the numerically evaluated phase diagram. Our formulas do not only reproduce, but also improve on the numerical phase diagram. For example, the free energy difference between the F and P phases is tiny for $\nu$ close to $\pm 1$, which makes it very difficult to determine the exact structure of the phase diagram numerically in this regime, see e.g. \cite[Fig. 1]{LW2007}. Our analytic formulas show that there is a critical value of the coupling parameter $U$ given by $4\pi$, such that in the limit $\nu \to \pm 1$ the system is in the F state for $U > 4\pi$ and in the P state for $U < 4\pi$.


\appendix

\section{Background on the Hubbard model}\label{modelapp}
We define the $n$-dimensional Hubbard model on the cubic lattice $\Lambda_L = \Z^n \cap [-L/2, L/2)^n$ containing $L^n$ sites, where $L \geq 0$ is an even integer and $n\in\{1,2,3,\ldots\}$. 
We also explain how the Hartree--Fock equations for the Hubbard model on $\Lambda_L $ are derived; the Hartree--Fock equations for the Hubbard model on $\Z^n$ are obtained in the limit $L\to\infty$.

\subsection{The fermion Fock space}
The $n$-dimensional Hubbard model is defined by a Hamiltonian operator $H$ acting on the {\it fermion Fock space} $\mathcal{F}$ defined by 
$$\mathcal{F} = \bigwedge \mathcal{H} = \bigwedge\nolimits^{0} \mathcal{H} \oplus \bigwedge\nolimits^1 \mathcal{H} \oplus \cdots \oplus \bigwedge\nolimits^{\dim \mathcal{H}} \mathcal{H},$$
where $\mathcal{H}$ is the {\it single-particle Hilbert space}. Since a single electron has to be at one of the $L^{n}$ lattice sites and can have either spin up or spin down, we can identify $\mathcal{H}$ with the set of functions from $\Lambda_L \times \{\uparrow, \downarrow\}$ to $\C$ endowed with the inner product 
$$\langle \psi_1, \psi_2 \rangle = \sum_{\mathbf{x}\in \Lambda_L} \sum_{\sigma = \uparrow, \downarrow} \overline{\psi_1(\mathbf{x}, \sigma)}\psi_2(\mathbf{x}, \sigma).$$
A natural orthonormal basis for $\mathcal{H} \cong \C^{2L^n}$ is
$\{\psi_{\mathbf{x} \sigma}\}_{\mathbf{x} \in \Lambda_L, \sigma \in \{\uparrow, \downarrow\}}$, where $\psi_{\mathbf{x} \sigma}$ is the function that assigns the value $1$ to $(\mathbf{x}, \sigma)$ and the value $0$ to all other points in $\Lambda_L \times \{\uparrow, \downarrow\}$. In particular, $\dim \mathcal{H} = 2L^n$.
For each $k = 0,1,\dots, 2L^n$, the space $\bigwedge\nolimits^k \mathcal{H}$ 
can be identified with the space of alternating linear maps from $(\mathcal{H}^*)^k \cong \mathcal{H}^k$ to $\C$. It
is spanned by vectors of the form
$$\psi_{e_1} \wedge \cdots \wedge \psi_{e_k}
= \sum_{P \in S_k} \sgn(P) \psi_{e_{P(1)}} \otimes \cdots \otimes \psi_{e_{P(k)}},$$
where $S_k$ is the set of permutations of $k$ elements, $e_j$ is short for $\mathbf{x}_j \sigma_j$, and 
$$(\psi_1 \otimes \cdots \otimes \psi_k)(X_1, \dots, X_k) := \prod_{j=1}^k \langle \psi_j, X_j \rangle_{\mathcal{H} \times \mathcal{H}^*}
\qquad \text{if $(X_1, \dots, X_k) \in (\mathcal{H}^*)^k$}.$$
The inner product on $\bigotimes \nolimits^k \mathcal{H}$ is defined on simple vectors of the form $\psi_{e_1} \otimes \cdots \otimes \psi_{e_k}$ by 
$$\langle \psi_{e_1} \otimes \cdots \otimes \psi_{e_k}, \psi_{f_1} \otimes \cdots \otimes \psi_{f_k} \rangle
= \prod_{j=1}^k \langle \psi_{e_j}, \psi_{f_j} \rangle = \prod_{j=1}^k \delta_{e_j, f_j}$$
and extended by bilinearity.
This induces an inner product on $\bigwedge\nolimits^k \mathcal{H}$ by restriction:
\begin{align*}
\langle \psi_{e_1} \wedge \cdots \wedge \psi_{e_k}, \psi_{f_1} \wedge \cdots \wedge \psi_{f_k} \rangle
& = 
\sum_{P, Q \in S_k} \sgn(PQ) \prod_{j=1}^k  \langle \psi_{e_{P(j)}}, \psi_{f_{Q(j)}} \rangle
	\\
& = 
k! \sum_{P\in S_k} \sgn(P) \prod_{j=1}^k \langle \psi_{e_{P(j)}}, \psi_{f_{j}} \rangle.
\end{align*}
In particular, if $i_1 < i_2 < \cdots < i_k$ with respect to some ordering of the index set labeling the points in $\Lambda_L \times \{\uparrow, \downarrow\}$, then
$$\langle \psi_{e_{i_1}} \wedge \cdots \wedge \psi_{e_{i_k}}, \psi_{e_{i_1}} \wedge \cdots \wedge \psi_{e_{i_k}} \rangle
= 
k! \sum_{P\in S_k} \sgn(P) \prod_{j=1}^k \delta_{P(i_j), i_j}
= k!.
$$
We conclude that the set
$$\bigg\{\frac{1}{\sqrt{k!}} \psi_{e_{i_1}} \wedge \cdots \wedge \psi_{e_{i_k}} \bigg| 1 \leq i_1 < i_2 < \cdots < i_k \leq \dim \mathcal{H}\bigg\}$$
is an orthonormal basis for $\bigwedge\nolimits^k \mathcal{H}$.
In particular, $\bigwedge\nolimits^k \mathcal{H}$ has dimension $\begin{pmatrix} \dim \mathcal{H} \\ k \end{pmatrix}$ and the fermion Fock space $\mathcal{F}$ is isomorphic to $\C^{4^{L^n}}$. 

\subsection{Creation and annihilation operators}
We define creation and annihilation operators on $\mathcal{F}$ as follows.
For any $e_0 \in \Lambda_L \times \{\uparrow, \downarrow\}$, the {\it creation operator} $c_{e_0}^\dagger:\mathcal{F} \to \mathcal{F}$ acts on the basis vector $\frac{1}{\sqrt{k!}}\psi_{e_1} \wedge \cdots \wedge \psi_{e_k}$ by
\begin{align}\label{creationoperatordef}
c_{e_0}^\dagger\bigg(\frac{1}{\sqrt{k!}}\psi_{e_1} \wedge \cdots \wedge \psi_{e_k}\bigg) = \frac{1}{\sqrt{(k+1)!}}\psi_{e_0} \wedge \psi_{e_1} \wedge \cdots \wedge \psi_{e_k}
\end{align}
and is extended to all of $\mathcal{F}$ by linearity.
Similarly, the {\it annihilation operator} $c_{e_0}:\mathcal{F} \to \mathcal{F}$ acts on $\frac{1}{\sqrt{k!}} \psi_{e_1} \wedge \cdots \wedge \psi_{e_k}$ by
\begin{align}\label{annihilationoperatordef}
c_{e_0} \bigg(\frac{1}{\sqrt{k!}} \psi_{e_1} \wedge \cdots \wedge \psi_{e_k}\bigg) 
= \sum_{j=1}^k (-1)^{j-1} \langle \psi_{e_0}, \psi_{e_j} \rangle 
\frac{1}{\sqrt{(k-1)!}} \psi_{e_1} \wedge \cdots \wedge \widehat{\psi_{e_j}} \wedge \cdots \wedge \psi_{e_k}
\end{align}
and is extended to all of $\mathcal{F}$ by linearity. In \eqref{annihilationoperatordef}, a hat indicates that the corresponding vector is omitted.
If $e_0 = \mathbf{x}_0 \sigma_0$, then we think of $c_{e_0}^\dagger$ as the operator that creates an electron of spin $\sigma_0 \in \{\uparrow, \downarrow\}$ at the lattice site $\mathbf{x}_0 \in \Lambda_L$, and of $c_{\mathbf{x} \sigma}$ as the operator that annihilates an electron of spin $\sigma_0$ at $\mathbf{x}_0$.
We observe that $c_{e_0}^\dagger$ is the adjoint of $c_{e_0}$:
\begin{align*}
\langle c_{e_0}^\dagger (\psi_{e_1} \wedge \cdots \wedge \psi_{e_k}), \psi_{f_0} \wedge \psi_{f_1} \wedge \cdots \wedge \psi_{f_k} \rangle
= \langle \psi_{e_1} \wedge \cdots \wedge \psi_{e_k}, c_{e_0}(\psi_{f_0} \wedge \psi_{f_1} \wedge \cdots \wedge \psi_{f_k}) \rangle.
\end{align*}
Indeed, on the one hand
\begin{align*}
\bigg\langle c_{e_0}^\dagger \bigg(\frac{1}{\sqrt{k!}} \psi_{e_1} \wedge \cdots \wedge \psi_{e_k}\bigg)&, \frac{1}{\sqrt{(k+1)!}}\psi_{f_0} \wedge \psi_{f_1} \wedge \cdots \wedge \psi_{f_k} \bigg\rangle
	\\
& =  \frac{1}{(k+1)!} \langle \psi_{e_0} \wedge \psi_{e_1} \wedge \cdots \wedge \psi_{e_k}, \psi_{f_0} \wedge \psi_{f_1} \wedge \cdots \wedge \psi_{f_k} \rangle
	\\
& = \sum_{P\in S_{k+1}} \sgn(P) \prod_{j=0}^{k} \langle \psi_{e_{P(j)}}, \psi_{f_{j}} \rangle,
\end{align*}
while on the other hand, if $T_j:\{1, \dots, k\} \to \{1, \dots, j-1, j+1, \dots, k\}$ is the map 
$$T_j(i) = \begin{cases} i, & i \in \{1, \dots, j-1\}, \\
i+1, & i \in \{j, \dots, k\},
\end{cases}$$
then
\begin{align*}
\bigg\langle & \frac{1}{\sqrt{k!}} \psi_{e_1} \wedge \cdots \wedge \psi_{e_k}, c_{e_0}\bigg(\frac{1}{\sqrt{(k+1)!}}\psi_{f_0} \wedge \psi_{f_1} \wedge \cdots \wedge \psi_{f_k}\bigg) \bigg\rangle
	\\
& = \frac{1}{k!} \bigg\langle \psi_{e_1} \wedge \cdots \wedge \psi_{e_k}, \sum_{j=0}^k (-1)^{j} \langle \psi_{e_0}, \psi_{f_j}\rangle  \psi_{f_0} \wedge \cdots \wedge \widehat{\psi_{f_j}} \wedge \cdots \wedge \psi_{f_k} \bigg\rangle
	\\
& = \sum_{j=0}^k (-1)^{j} \langle \psi_{e_0}, \psi_{f_j}\rangle 
\sum_{Q\in S_k} \sgn(Q) \prod_{i=1}^k \langle \psi_{e_{i}}, \psi_{f_{T_j(Q(i))}} \rangle
	\\
& = \sum_{P\in S_{k+1}} 
\sgn(P) \prod_{i=0}^k \langle \psi_{e_{i}}, \psi_{f_{P(i)}} \rangle.
\end{align*}

\subsection{Commutation relations}
The creation and annihilation operators obey the anticommutation relations
$$\{c_{e_0}^\dagger, c_{f_0}^\dagger\} = 0, \qquad \{c_{e_0}, c_{f_0}\} = 0, \qquad \{c_{e_0}, c_{f_0}^\dagger\} =  \langle \psi_{e_0}, \psi_{f_0} \rangle I$$
for any $e_0, f_0 \in \Lambda_L \times \{\uparrow, \downarrow\}$, where $\{A, B\} = AB + BA$. Indeed, the first two relations follow easily from the definitions, and the third can be verified as follows:
\begin{align*}
  \{c_{e_0}, c_{f_0}^\dagger\} & (\psi_{f_1} \wedge \cdots \wedge \psi_{f_k})
  = c_{e_0} c_{f_0}^\dagger(\psi_{f_1} \wedge \cdots \wedge \psi_{f_k})
  + c_{f_0}^\dagger c_{e_0} (\psi_{f_1} \wedge \cdots \wedge \psi_{f_k})
  	\\
 = &\; c_{e_0} \frac{\sqrt{k!}}{\sqrt{(k+1)!}}\psi_{f_0} \wedge \psi_{f_1} \wedge \cdots \wedge \psi_{f_k}
 	\\
&  + c_{f_0}^\dagger \sum_{j=1}^k (-1)^{j-1} \langle \psi_{e_0}, \psi_{f_j} \rangle 
\frac{\sqrt{k!}}{\sqrt{(k-1)!}} \psi_{f_1} \wedge \cdots \wedge \widehat{\psi_{f_j}} \wedge \cdots \wedge \psi_{f_k}
 	\\
 = &\; \frac{\sqrt{k!}}{\sqrt{(k+1)!}}
\sum_{j=0}^k (-1)^{j} \langle \psi_{e_0}, \psi_{f_j} \rangle 
\frac{\sqrt{(k+1)!}}{\sqrt{k!}} \psi_{f_0} \wedge \cdots \wedge \widehat{\psi_{f_j}} \wedge \cdots \wedge \psi_{f_k}
	\\
&  + \sum_{j=1}^k (-1)^{j-1} \langle \psi_{e_0}, \psi_{f_j} \rangle 
\frac{\sqrt{k!}}{\sqrt{(k-1)!}} 
\frac{\sqrt{(k-1)!}}{\sqrt{k!}}\psi_{f_0} \wedge \psi_{f_1} \wedge \cdots \wedge \widehat{\psi_{f_j}} \wedge \cdots \wedge \psi_{f_k}
 	\\
 = &\; 
\bigg(\sum_{j=0}^k (-1)^{j} + \sum_{j=1}^k (-1)^{j-1} \bigg)\langle \psi_{e_0}, \psi_{f_j} \rangle 
\psi_{f_0} \wedge \cdots \wedge \widehat{\psi_{f_j}} \wedge \cdots \wedge \psi_{f_k}
  	\\
 = &\;  \langle \psi_{e_0}, \psi_{f_0} \rangle 
\psi_{f_1} \wedge \cdots \wedge \psi_{f_k}.
\end{align*}

\subsection{The Hubbard Hamiltonian}
The Hamiltonian $H:\mathcal{F} \to \mathcal{F}$ for the $n$-dimensional Hubbard model depends on the hopping parameter $t>0$, the coupling parameter $U > 0$ that determines the strength of the on-site repulsion,
as well as the chemical potential $\mu \in \R$. It is defined by
\begin{align*}
H & := H_0 + U \sum_{\mathbf{x} \in \Lambda_L} (n_{\mathbf{x}\uparrow} -\tfrac12)(n_{\mathbf{x} \downarrow}-\tfrac{1}{2}),
	\\
H_0 & := \sum_{\mathbf{x}, \mathbf{y} \in \Lambda_L} \sum_{\sigma = \uparrow, \downarrow} t_{\mathbf{x}\mathbf{y}} c_{\mathbf{x} \sigma}^\dagger c_{\mathbf{y} \sigma} - \mu \sum_{\mathbf{x} \in \Lambda_L} (n_{\mathbf{x} \uparrow} + n_{\mathbf{x} \downarrow} - 1),
\end{align*}
where $c_{\mathbf{x} \sigma}^\dagger$ and $c_{\mathbf{x} \sigma}$ are the creation and annihilation operators defined in (\ref{creationoperatordef})--(\ref{annihilationoperatordef}) with $e_0=\mathbf{x} \sigma$, $n_{\mathbf{x} \sigma} := c_{\mathbf{x} \sigma}^\dagger c_{\mathbf{x} \sigma}$ is the {\it density operator}, and the {\it hopping matrix elements} $t_{\mathbf{x}\mathbf{y}}$ are given by
\begin{align}\label{txy2D}
t_{\mathbf{x}\mathbf{y}} 
= \begin{cases} -t & \text{if $\mathbf{x}$ and $\mathbf{y}$ are nearest neighbors},
	\\
0 & \text{otherwise}.
\end{cases}
\end{align}
As is customary, we use periodic boundary conditions when determining whether $\mathbf{x}$ and $\mathbf{y}$ are nearest neighbors, i.e., we identity $\Lambda_L$ with the torus $(\Z/L\Z)^n$.
If $H(t, U, \mu)$ denotes the Hamiltonian corresponding to the parameters $t, U, \mu$, then $H(t, tU, t\mu)=tH(1, U, \mu)$. It follows that the hopping parameter $t>0$ can be scaled out of the problem; in the main text, we have therefore set $t=1$.

\subsection{Equilibrium states}
A {\em state} $\rho$ (also called a {\em density matrix}) is a self-adjoint operator on $\mathcal{F}$ with nonnegative eigenvalues such that $\Tr_{\mathcal{F}}(\rho)=1$, where $\Tr_{\mathcal{F}}$ is the trace on $\mathcal{F}$. We are interested in the equilibrium state of the model as a function of the coupling strength $U$ and the chemical potential $\mu \in \R$; the equilibrium state is the most probable state of the system for given values of $U$ and $\mu$, and it is given by the state that minimizes the {\it grand canonical potential} (also known as the {\it Landau free energy}) $\Omega(\rho; U, \mu)$ defined by
$$\Omega(\rho; U, \mu) = \Tr_{\mathcal{F}}[\rho H] + \frac{1}{\beta} \Tr_{\mathcal{F}}[\rho \ln \rho],$$
where $\beta  =1/T \in (0, +\infty]$ is the inverse temperature.  
At temperature $1/\beta > 0$, the equilibrium state of the system is given by {\it the Gibbs state}
\begin{equation}\label{rhoG} 
	\rho_{\mathrm{G}} = \frac{e^{-\beta H}}{\Tr_{\mathcal{F}}(e^{-\beta H})}.
	\end{equation} 
Indeed, for a fixed temperature $1/\beta>0$, $\rho_{\mathrm{G}}$ is the unique absolute minimum of $\Omega$ in the set of all states $\rho$, see e.g. \cite{W1978}. We refer to $\Omega(\rho; U, \mu)/L^n$ as the {\em free energy density} of a state $\rho$.

\subsection{Doping}
Instead of using the chemical potential $\mu$ as a parameter, Figure \ref{phasediagramfig} uses the {\it doping} $\nu$ as a parameter along the horizontal axis. The doping $\nu$ of a state $\rho$ is defined as the average number of electrons per site in the system minus 1, so that $-1\leq \nu\leq 1$ and half-filling corresponds to $\nu=0$:
\begin{equation}\label{rho} 
\nu= \frac{1}{L^n} \sum_{\mathbf{x} \in \Lambda_L} \Tr_{\mathcal{F}}(\rho n_{\mathbf{x}}) - 1.
\end{equation}
It is usually more relevant to display the phase diagram in terms of the doping instead of the chemical potential because the doping is a quantity that can be physically measured.

\subsection{Hartree--Fock theory}
The Hartree--Fock approximation consists of minimizing $\Omega$ only over a subset of the set of all states. 
In unrestricted Hartree--Fock theory, the minimization is over all so-called quasi-free states, where a state $\rho$ is quasi-free if expectation values defined as $\langle A \rangle_{\rho} := \Tr_{\mathcal{F}}(\rho A)$ satisfy Wick's theorem, see \cite[p. 11]{BLS1994} for the precise definition. In particular, a quasi-free state $\rho$ satisfies
$$\langle c_{e_1}^\dagger c_{e_2}^\dagger c_{e_3} c_{e_4} \rangle_\rho
= \langle c_{e_1}^\dagger c_{e_2}^\dagger  \rangle_\rho \langle c_{e_3} c_{e_4} \rangle_\rho
- \langle c_{e_1}^\dagger c_{e_3} \rangle_\rho \langle  c_{e_2}^\dagger c_{e_4} \rangle_\rho
+ \langle c_{e_1}^\dagger c_{e_4} \rangle_\rho \langle c_{e_2}^\dagger c_{e_3} \rangle_\rho.$$

Define the spin operator $\vec{s}_{\mathbf{x}}:\mathcal{F} \to \mathcal{F}$ for $\mathbf{x} \in \Lambda_L$ by
\begin{align} 
\vec{s}_{\mathbf{x}} = (s^x_{\mathbf{x}},s^y_{\mathbf{x}},s^z_{\mathbf{x}})
= \sum_{\lambda, \tau \in \{\uparrow, \downarrow\}} c_{\mathbf{x}\lambda}^\dagger \vec{\sigma}_{\lambda \tau} c_{\mathbf{x}\tau},
\end{align} 
where $\vec{\sigma}=(\sigma^x,\sigma^y,\sigma^z)$ are the usual Pauli sigma matrices, i.e.,
\begin{align*} 
	s_{\mathbf{x}}^x = c_{\mathbf{x},\uparrow}^\dagger c_{\mathbf{x},\downarrow} + c_{\mathbf{x},\downarrow}^\dagger c_{\mathbf{x},\uparrow},\qquad
	s_{\mathbf{x}}^y = -i c_{\mathbf{x},\uparrow}^\dagger c_{\mathbf{x},\downarrow} + i c_{\mathbf{x},\downarrow}^\dagger c_{\mathbf{x},\uparrow},\qquad
	s_{\mathbf{x}}^z =  n_{\mathbf{x},\uparrow}-n_{\mathbf{x},\downarrow}. 
\end{align*} 	
Let $K$ and $B$ denote the sets of functions $d:\Lambda_L \to \R$ and $\vec{m}:\Lambda_L \to \R^3$, respectively (see also Remark \ref{BKremark}).
For $d \in K$ and $\vec{m} \in B$, define the Hartree--Fock Hamiltonian $H_{\mathrm{HF}}(d, \vec{m}): \mathcal{F} \to \mathcal{F}$ by
\begin{align*}
H_{\mathrm{HF}}(d, \vec{m}) := H_0 + \frac{U}{2}\sum_{\mathbf{x}}\big( d(\mathbf{x})(n_{\mathbf{x}}-1)-\vec{m}(\mathbf{x})\cdot\vec{s}_{\mathbf{x}} \big).
\end{align*} 
Define the {\it Hartree--Fock function} $\mathcal{G}: K \times B \times (0,+\infty) \times \R \to \R$ by
\begin{equation}\label{FHF}  
\mathcal{G}(d,\vec{m}, U, \mu) := \frac{1}{L^n}\bigg(\frac{U}{4}\sum_{\mathbf{x} \in \Lambda_L}\big( |\vec{m}(\mathbf{x})|^2-d(\mathbf{x})^2\big) 
- \Tr_{\mathcal{H}}\big(\Ln_\beta(h)\big)  \bigg),
\end{equation}  
where $\Tr_{\mathcal{H}}$ is the trace on $\mathcal{H}$, $|\vec{m}(\mathbf{x})|^2=\vec{m}(\mathbf{x})\cdot  \vec{m}(\mathbf{x})$, 
$$\Ln_\beta(E) := \begin{cases}
\frac{1}{\beta} \ln\big( 2 \cosh\big(\frac{\beta E}{2}\big)\big), & \beta \in (0, +\infty),
	\\
\frac{|E|}{2}, & \beta = +\infty,
\end{cases}$$	
and $h:\mathcal{H}\to \mathcal{H}$ is the linear operator on $\mathcal{H}$ whose matrix elements with respect to the basis $\{\psi_{\mathbf{x} \sigma}\}_{\mathbf{x} \in \Lambda_L, \sigma \in \{\uparrow, \downarrow\}}$ are given by
\begin{align*}
	h_{\mathbf{x}\sigma;\mathbf{y}\sigma'}
	& := \big(t_{\mathbf{x}\mathbf{y}}-\mu\delta_{\mathbf{x},\mathbf{y}}\big)\delta_{\sigma,\sigma'}
+ \frac {U}{2}\delta_{\mathbf{x},\mathbf{y}}\big(d(\mathbf{x})\delta_{\sigma,\sigma'} - \vec{m}(\mathbf{x})\cdot\vec{\sigma}_{\sigma\sigma'}\big).
\end{align*}
Our notation is such that, if $\vec{m} = (m^x, m^y, m^z)$, then $\vec{m}\cdot\vec{\sigma}_{\sigma\sigma'} = m^x\sigma^x_{\sigma\sigma'} + m^y\sigma^y_{\sigma\sigma'} + m^z\sigma^z_{\sigma\sigma'}$. Moreover, $\sigma^x_{\uparrow \uparrow}$ indicates the $(11)$-entry of $\sigma^x$, $\sigma^x_{\uparrow \downarrow}$ indicates the $(12)$-entry of $\sigma^x$, etc.
Since $h$ is hermitian, $\Ln_{\beta}(h)$ is well-defined.
Indeed, since $h$ is hermitian, there exists an orthonormal basis $\{\psi_j\}_1^{2L^n}$ for $\mathcal{H}$ such that $h\psi_j = E_j\psi_j$ for each $j$, where $\{E_j\}_1^{2L^n} \subset \R$ are the eigenvalues of $h$. The linear map $\Ln_\beta(h):\mathcal{H} \to \mathcal{H}$ and the trace $\Tr_{\mathcal{H}}\big(\Ln_\beta(h)\big)$ are given by
\begin{align}\label{Lnbetahexplicit}
\Ln_\beta(h)(\psi) = \sum_{j=1}^{2L^n} \psi_j \Ln_\beta(E_j) \langle \psi_j, \psi\rangle \quad \text{and} \quad \Tr_{\mathcal{H}}\big(\Ln_\beta(h)\big) = \sum_{j=1}^{2L^n} \Ln_\beta(E_j).
\end{align}

Let us first consider the case of strictly positive temperature $T = 1/\beta > 0$. 
In this case, if $(d^*, \vec{m}^*)$ is an extremizer of $\mathcal{G}$ in the sense that
\begin{align}\label{Gextremizer}
\mathcal{G}(d^*,\vec{m}^*, U, \mu) = \max_{d \in K} \mathcal{G}(d,\vec{m}^*, U, \mu)
= \min_{\vec{m} \in B} \max_{d \in K} \mathcal{G}(d,\vec{m}, U, \mu),
\end{align}
then the state $\rho^{*}$ defined by
\begin{align}\label{rhostar}
\rho^* := \frac{e^{-\beta H_{\mathrm{HF}}(d^*, \vec{m}^*)}}{\Tr_{\mathcal{F}} (e^{-\beta H_{\mathrm{HF}}(d^*, \vec{m}^*)})}
\end{align}
is a minimizer of the grand canonical potential over all quasi-free states (note that $\rho^*$ is the Gibbs state associated with the quadratic Hamiltonian $H_{\mathrm{HF}}(d^*, \vec{m}^*)$), and the value of the Hartree--Fock function at $(d^*,\vec{m}^*)$ is the free energy density of $\rho^*$, i.e.,
\begin{align}\nonumber
 \mathcal{F}(U, \mu) := &\, \min \bigg\{\frac{\Omega(\rho; U, \mu)}{L^n}\, \bigg| \, \text{$\rho$ is a quasi-free state}\bigg\} = \frac{\Omega(\rho^*; U, \mu) }{L^n}
 	\\ \label{FOmega}
 =&\; \mathcal{G}(d^*,\vec{m}^*, U, \mu).
\end{align}
Moreover, $d^*(\mathbf{x})$ and $\vec{m}^*(\mathbf{x})$ are the expected doping and the expected spin at $\mathbf{x}$ in the state $\rho^*$, respectively, i.e.,
\begin{align}\label{dstarmstar}
& d^*(\mathbf{x}) = \langle n_{\mathbf{x}}\rangle_{\rho^*} - 1, 
\qquad
 \vec{m}^*(\mathbf{x}) = \langle \vec{s}_{\mathbf{x}} \rangle_{\rho^*}.
\end{align}
We refer to \cite{BP1996} (see also \cite{LLphysrevlong}) for proofs of the above facts.

Similar statements hold in the case of zero temperature $T = 0$, if $\rho^*$ and $(d^*, \vec{m}^*)$ are chosen as limits of the corresponding quantities for $T > 0$ as $T \downarrow 0$, see \cite[Theorem 2]{BP1996}. 

\subsection{Hartree--Fock equations}
The {\it Hartree--Fock equations} (also known as the {\it mean-field equations}) for unrestricted Hartree--Fock theory are
\begin{align}\label{unrestrictedmeanfieldeqs}
\begin{cases}
\frac{\partial \mathcal{G}}{\partial d(\mathbf{x})}(d,\vec{m}, U, \mu) = 0, \\
 \frac{\partial \mathcal{G}}{\partial \vec{m}(\mathbf{x})}(d,\vec{m}, U, \mu) = \vec{0},
 \end{cases} \quad \text{for all $\mathbf{x} \in \Lambda_L$}.
 \end{align}
More explicitly, we see from (\ref{FHF}) that the Hartree--Fock equations are
\begin{align}\label{unrestrictedmeanfieldeqs2}
\begin{cases}
d(\mathbf{x}) = - \Tr_{\mathcal{H}}\big(\Ln_\beta'(h) (\mathbbm{1}_{\mathbf{x}} \otimes I)\big),
	 \\
\vec{m}(\mathbf{x}) = - \Tr_{\mathcal{H}}\big(\Ln_\beta'(h) (\mathbbm{1}_{\mathbf{x}} \otimes \vec{\sigma})\big),
 \end{cases} \quad \text{for all $\mathbf{x} \in \Lambda_L$}.
 \end{align}
where the second line is a short way of writing the three equations $m^A(\mathbf{x}) = - \Tr_{\mathcal{H}}\big(\Ln_\beta'(h) (\mathbbm{1}_{\mathbf{x}} \otimes \sigma^A)\big)$, $A = x,y,z$, and the operators $\mathbbm{1}_{\mathbf{x}} \otimes I$ and $\mathbbm{1}_{\mathbf{x}} \otimes \vec{\sigma}$ act on $\psi \in \mathcal{H}$ by
 $$(\mathbbm{1}_{\mathbf{x}} \otimes I)\psi = \sum_{\sigma \in \{\uparrow, \downarrow\}} \psi_{\mathbf{x}\sigma} \langle \psi_{\mathbf{x}\sigma}, \psi \rangle, \qquad
 (\mathbbm{1}_{\mathbf{x}} \otimes \vec{\sigma})\psi = \sum_{\sigma, \sigma' \in \{\uparrow, \downarrow\}} \psi_{\mathbf{x}\sigma}  \vec{\sigma}_{\sigma \sigma'} \langle \psi_{\mathbf{x}\sigma'}, \psi \rangle.$$
For every $\vec{m} \in B$, there is a unique $d^*_{\vec{m}} \in K$ such that $\mathcal{G}(d^*_{\vec{m}},\vec{m}, U, \mu) = \max_{d \in K} \mathcal{G}(d,\vec{m}, U, \mu)$, because $\mathcal{G}$ is strictly concave in $d$, see \cite[p. 8]{BP1996}.
At $d^*_{\vec{m}} \in K$, we have $\frac{\partial \mathcal{G}}{\partial d(\mathbf{x})}(d^*_{\vec{m}},\vec{m}, U, \mu) = 0$ for every $\mathbf{x}$. 
By considering the local extrema of the function $\vec{m} \mapsto \mathcal{G}(d^*_{\vec{m}},\vec{m}, U, \mu)$, 
we therefore conclude from (\ref{Gextremizer}) and (\ref{FOmega}) that
\begin{align}\label{freeenergyunrestricted}
\mathcal{F}(U, \mu) = \min\{\mathcal{G}(d, \vec{m}, U, \mu) \, |\, \text{$(d, \vec{m}) \in K \times B$ solves (\ref{unrestrictedmeanfieldeqs})}\}.
\end{align}
Indeed, by (\ref{unrestrictedmeanfieldeqs2}),
$d^*_{\vec{m}}(\mathbf{x})
= - \Tr_{\mathcal{H}}\big(\Ln_\beta'(h) (\mathbbm{1}_{\mathbf{x}} \otimes I)\big)
= - \Tr_{\mathcal{H}}\big(A_{\mathbf{x}}\big)$, where the operator $A_{\mathbf{x}}:\mathcal{H} \to \mathcal{H}$ is defined by 
\begin{align}\label{Axdef}
A_{\mathbf{x}} := (\mathbbm{1}_{\mathbf{x}} \otimes I) \Ln_\beta'(h) (\mathbbm{1}_{\mathbf{x}} \otimes I).
\end{align}
Hence $d^*_{\vec{m}}(\mathbf{x}) = - a_{\mathbf{x}} - b_{\mathbf{x}}$, where $a_{\mathbf{x}},b_{\mathbf{x}}$ are the eigenvalues of $A_{\mathbf{x}}$ restricted to the two-dimensional subspace $(\mathbbm{1}_{\mathbf{x}} \otimes I) \mathcal{H}$. Since $|\Ln_{\beta}'(E)| \leq 1/2$ for all $E\in \R$, we have $a_{\mathbf{x}},b_{\mathbf{x}} \in [-1/2,1/2]$. It follows that
\begin{align}\label{dstarbound}
d^*_{\vec{m}}(\mathbf{x}) \in [-1, 1] \quad \text{for every $\mathbf{x} \in \Lambda_L$}.
\end{align}
Thus, by \eqref{FHF}, $\vec{m} \to \mathcal{G}(d^*_{\vec{m}},\vec{m}, U, \mu)$ grows like $|\vec{m}(\mathbf{x})|^2$ as $|\vec{m}(\mathbf{x})| \to +\infty$, so the minimum in (\ref{Gextremizer}) over $\vec{m} \in B$ is attained at a finite local extremum, which gives (\ref{freeenergyunrestricted}).

\begin{remark}\label{BKremark}\upshape
For easy comparison with \cite{BP1996}, we point out that the sets $K$ and $B$ in (\ref{Gextremizer}) and (\ref{freeenergyunrestricted}) may be replaced by the smaller sets
\begin{align*}
& K' = \{d:\Lambda_L \to \R \,|\; \text{$|d(\mathbf{x})| \leq 1$ for all $\mathbf{x} \in \Lambda_L$}\},
	\\
& B' = \{\vec{m}:\Lambda_L \to \R^3 \,|\; \text{$|\vec{m}(\mathbf{x})| \leq 1$ for all $\mathbf{x} \in \Lambda_L$}\},
\end{align*}
without changing the results. Indeed, suppose $(d, \vec{m}) \in K \times B$ solves (\ref{unrestrictedmeanfieldeqs}). 
By (\ref{dstarbound}), we have $d \in K'$.
Furthermore, we find from (\ref{Lnbetahexplicit}) and (\ref{unrestrictedmeanfieldeqs2}) that
\begin{align*}
\vec{m}(\mathbf{x}) & =
- \Tr_{\mathcal{H}}\big(\Ln_\beta'(h) (\mathbbm{1}_{\mathbf{x}} \otimes \vec{\sigma})\big)
 =
-\sum_{j,k=1}^{2L^n} \langle \psi_k, \psi_j \rangle \Ln_\beta'(E_j) \langle \psi_j, (\mathbbm{1}_{\mathbf{x}} \otimes \vec{\sigma})\psi_k \rangle
	\\
& =
-\sum_{j,k=1}^{2L^n} \sum_{\sigma, \sigma' \in \{\uparrow, \downarrow\}}  \vec{\sigma}_{\sigma \sigma'} \langle \psi_{\mathbf{x}\sigma'}, \psi_k \rangle
\langle \psi_k, \psi_j \rangle \Ln_\beta'(E_j) \langle \psi_j, \psi_{\mathbf{x}\sigma} \rangle
	\\
& = - \sum_{j =1}^{2L^n} \sum_{\sigma, \sigma' \in \{\uparrow, \downarrow\}} \Ln_\beta'(E_j) \langle \psi_j, \psi_{\mathbf{x}\sigma} \rangle 
\vec{\sigma}_{\sigma \sigma'} \langle \psi_{\mathbf{x}\sigma'}, \psi_j\rangle.
\end{align*}
Using the identity $\vec{\sigma}_{\sigma \sigma'} \cdot \vec{\sigma}_{\tau \tau'} = 2\delta_{\sigma, \tau'}\delta_{\sigma', \tau} - \delta_{\sigma, \sigma'} \delta_{\tau, \tau'}$, we obtain after some simplifications
\begin{align}
|\vec{m}(\mathbf{x})|^2 = 2\Tr_{\mathcal{H}}(A_{\mathbf{x}}^2) - \Tr_{\mathcal{H}}(A_{\mathbf{x}})^2,
\end{align}
where $A_{\mathbf{x}}:\mathcal{H} \to \mathcal{H}$ is the operator in (\ref{Axdef}). If $a_{\mathbf{x}},b_{\mathbf{x}} \in [-1/2,1/2]$ are the eigenvalues of $A_{\mathbf{x}}$ restricted to the two-dimensional subspace $(\mathbbm{1}_{\mathbf{x}} \otimes I) \mathcal{H}$, then 
$$|\vec{m}(\mathbf{x})|^2 = 2(a_{\mathbf{x}}^2 + b_{\mathbf{x}}^2) - (a_{\mathbf{x}} + b_{\mathbf{x}})^2 = (a_{\mathbf{x}}-b_{\mathbf{x}})^2 \leq 1.$$
This shows that any solution $(d, \vec{m}) \in K \times B$ of (\ref{unrestrictedmeanfieldeqs}) in fact lies in $K' \times B'$, and completes the proof of the claim.
\end{remark}

\subsection{Hartree--Fock theory restricted to P, F, and AF states}
The state $\rho^{*}$ defined in (\ref{rhostar}) is the best approximation of the true Gibbs state among all states of the form
\begin{align}\label{HFGibbs}
\frac{e^{-\beta H_{\mathrm{HF}}(d, \vec{m})}}{\Tr_{\mathcal{F}} (e^{-\beta H_{\mathrm{HF}}(d, \vec{m})})}
\end{align}
where $d \in K$ and $\vec{m} \in B$. In Hartree--Fock theory restricted to P, F, and AF states, we consider the best approximation of the true Gibbs state among all states of the form (\ref{HFGibbs}) with $d$ and $\vec{m}$ of the following form:
\begin{subequations}\label{PFAF}
\begin{align}
 \text{P:} & \quad \text{$d(\mathbf{x}) = d_0$ and $\vec{m}(\mathbf{x}) = 0$ for all $\mathbf{x} \in \Lambda_L$},
	\\
 \text{F:} & \quad \text{$d(\mathbf{x}) = d_0$ and $\vec{m}(\mathbf{x}) = m_0 \vec{e}$ for all $\mathbf{x} \in \Lambda_L$},
	\\
 \text{AF:} & \quad \text{$d(\mathbf{x}) = d_0$ and $\vec{m}(\mathbf{x}) = (-1)^{\mathbf{x}}m_1 \vec{e}$ for all $\mathbf{x} \in \Lambda_L$},
\end{align}
\end{subequations}
where $d_0 \in \R$, $m_0 > 0$, $m_1 > 0$, and $\vec{e} \in \R^3$ with $\|\vec{e}\|=1$, are independent of $\mathbf{x}$, and $(-1)^{\mathbf{x}} := (-1)^{x_1 + \dots + x_n}$.

The expressions (\ref{PHartreeFockFunction})--(\ref{AFHartreeFockFunction}) for the P, F, and AF Hartree--Fock functions are obtained by evaluating (\ref{FHF}) with $d$ and $\vec{m}$ of the form (\ref{PFAF}) and taking the thermodynamic limit $L \to \infty$.
Indeed, if $d(\mathbf{x}) = d_0$ and $\vec{m}(\mathbf{x}) = (m_0 + (-1)^{\mathbf{x}}m_1) \vec{e}$, then (see e.g. \cite[Eq. (24)]{LLphysrevlong})
\begin{align*} 
\lim_{L \to +\infty} \mathcal{G}(d,\vec{m}, U, \mu) = \frac{U}{4}(m_0^2+m_1^2-d_0^2)
 - \frac{1}{2} \int_{[-\pi,\pi]^n}  \sum_{r,r'=\pm 1} \Ln_\beta(E_{r,r'}(\varepsilon(\mathbf{k}))) \frac{d^n \mathbf{k}}{(2\pi)^n},
\end{align*} 	
where the so-called effective band relations $E_{r,r'}(\varepsilon(\mathbf{k}))$ are expressed in terms of the functions $\varepsilon(\mathbf{k}) := -2t\sum_{i=1}^n\cos(k_i) $ and
\begin{align}\label{Errp}
	E_{r,r'}(\epsilon) :=  \frac{Ud_0}{2}-r\frac{Um_0}{2} -\mu 
	+ r'\sqrt{\epsilon^2 + \Big(\frac{Um_1}{2}\Big)^2}.
\end{align} 
Utilizing the density of states $N_0(\epsilon)= \int_{[-\pi,\pi]^n}\delta(\epsilon-\varepsilon(\mathbf{k})) \frac{d^n\mathbf{k}}{(2\pi)^n}$, we can write this as
\begin{align*} 
\lim_{L \to +\infty} \mathcal{G}(d,\vec{m}, U, \mu) = \frac{U}{4}(m_0^2+m_1^2-d_0^2)
 - \frac{1}{2} \int_{\R}  N_0(\epsilon) \sum_{r,r'=\pm 1} \Ln_\beta(E_{r,r'}(\epsilon)) d\epsilon,
\end{align*} 
which after simplification yields (\ref{PHartreeFockFunction})--(\ref{AFHartreeFockFunction}) in the special case of $t = 1$, $n = 2$,  and $\beta = +\infty$ (note that, for later convenience, we employ $\Ln_{+\infty}(E)=\frac12|E|=E\theta(E)-\frac{E}{2}$, $\int_{\R} N_0(\epsilon) d\epsilon =1$, and $N_{0}(\epsilon) = N_{0}(-\epsilon)$ to obtain the formulas in  (\ref{PHartreeFockFunction})--(\ref{AFHartreeFockFunction})).
The definitions (\ref{freeenergies}) of the P, F, and AF free energy densities are similarly a consequence of specializing (\ref{freeenergyunrestricted}) using (\ref{PFAF}).

\section{Properties of $N_0(\epsilon)$}\label{N0app}
We derive and recall some properties of the function $N_{0}(\epsilon)$ defined in (\ref{2Ddensityofstates}) that are needed in the main text. 

\subsection{Integrals involving $N_0(\epsilon)$}
Our first lemma computes several moments of $N_{0}(\epsilon)$.

\begin{lemma}\label{N0momentslemma}
The function $N_{0}(\epsilon)$ satisfies
\begin{align}
\int_{\R} \epsilon^{j} N_{0}(\epsilon) d\epsilon = \begin{cases}
0 & \mbox{if } j = 1, 3, 5, \dots, \\
1 & \mbox{if } j=0, \\
4 & \mbox{if } j=2, \\
36 & \mbox{if } j=4, \\
400 & \mbox{if } j=6.
\end{cases}
\end{align}
\end{lemma}
\begin{proof} 
Since the function $N_0$ is even, the claim for odd integers $j \geq 1$ is immediate.

For even integers $j\geq 0$, we can obtain the result from the identity \cite[Lemma~3.1]{LLa0}
\begin{align}\label{generalintegral} 
\int_0^4 \epsilon^s N_0(\epsilon)d\epsilon = \frac{4^s}{2\pi}\left( \frac{\Gamma(\tfrac{1}{2}+\tfrac{s}{2})}{\Gamma(1+\tfrac{s}{2})}\right)^2 \quad \text{for $s>-1$}. 
\end{align} 
Indeed, specializing (\ref{generalintegral}) to $s=2n=0,2,4,\ldots$ and using that $N_0(\epsilon)$ is even and non-zero only for $|\epsilon|\leq 4$, we get 
$$
\int_{\R} \epsilon^{2n} N_0(\epsilon)d\epsilon = 2\int_0^4 \epsilon^{2n} N_0(\epsilon)d\epsilon =  \frac{4^{2n}}{\pi}\left( \frac{\Gamma(\tfrac{1}{2}+n)}{\Gamma(1+n)}\right)^2= \binom{2n}{n}^2
\quad (n=0,1,2,3,\ldots) , 
$$
using $\Gamma(1+n)=n!$ and $\Gamma(\tfrac12+n)=(2n)!\sqrt{\pi}/(4^n n!)$; in particular, this gives the result for $j=2n=0,2,4,6$.
\end{proof}

The next lemma computes the integral $\int_0^4 N_0(\epsilon) \epsilon d\epsilon$. This integral (multiplied by $-2$) shows up in multiple places in the paper: it is the P free energy at $\mu = 0$ (see Lemma \ref{Plemma}), it is the leading term in the expansion of $\mathcal{F}_P$ in Sector III (see Lemma \ref{FPsectorIIIlemma}), and it is the leading term in the expansion of the AF Hartree--Fock function at each of the two AF mean-field solutions in Sector III (see Lemmas \ref{AFsolutionIIIlemma} and \ref{AFIIIlocalmaxlemma}).

\begin{lemma}
The function $N_{0}(\epsilon)$ satisfies
\begin{align}\label{N0epsilon8pi2}
\int_0^4  \epsilon N_0(\epsilon)  d\epsilon = \frac{8}{\pi^2}.
\end{align}
\end{lemma}
\begin{proof} 
This is the special case $s=1$ of \eqref{generalintegral},
$$
\int_0^4 \epsilon N_0(\epsilon)d\epsilon = \frac{4}{2\pi}\left( \frac{\Gamma(1)}{\Gamma(\tfrac{3}{2})}\right)^2 = \frac{2}{\pi}\left( \frac{1}{\frac12\sqrt{\pi}}  \right)^2 
= \frac{8}{\pi^2}.
$$
\end{proof}

\subsection{Expansion of $N_0(\epsilon)$ as $\epsilon \uparrow 4$}
By \cite[Eq. (4.8)]{LL2025}, we have 
$$N_0(\epsilon) = \frac{1}{\pi^2(4-\epsilon)}\int_0^1 \frac{dv}{\sqrt{v}\sqrt{v - w_1} \sqrt{w_2 - v} \sqrt{1-v}} \qquad \text{for $\epsilon \in (0,4)$},$$
where principal branches are used for the square roots and
$$w_1 = -\frac{\epsilon}{4-\epsilon}, \qquad w_2 = \frac{4}{4-\epsilon}.$$
Taylor expanding the integrand as $\epsilon \uparrow 4$ and performing the resulting integrals, we find 
\begin{align}\label{N0near4}
N_0(\epsilon) = \sum_{j=0}^6 N_0^{(j)}(4 - \epsilon)^j + O((4 - \epsilon)^7) \qquad \text{as $\epsilon \uparrow 4$},
\end{align}
where
\begin{align*}
& N_0^{(0)} = \frac{1}{4\pi^2} \int_0^1 \frac{dv}{\sqrt{v}\sqrt{1-v}} = \frac{1}{4\pi},
\quad
N_0^{(1)} = \frac{1}{32\pi^2} \int_0^1 \frac{dv}{\sqrt{v}\sqrt{1-v}} = \frac{1}{32\pi},
	\\
& N_0^{(2)} = \frac{1}{512 \pi^2}\int_0^1 \frac{4 v^2-4 v+3}{ \sqrt{1-v} \sqrt{v}} dv = \frac{5}{1024 \pi },
\quad
N_0^{(3)} = \frac{7}{8192 \pi }, \quad N_0^{(4)} = \frac{169}{1048576 \pi }, 
	\\
&N_0^{(5)} = \frac{269}{8388608 \pi },
 \quad N_0^{(6)} = \frac{1781}{268435456 \pi }.
\end{align*}

\subsection{Expansion of $N_0(\epsilon)$ as $\epsilon \downarrow 0$}
By \cite[Theorem 2.1]{LL2025}, the following asymptotic formula is valid as $\epsilon \downarrow 0$:
\begin{align}  \nonumber
N_0(\epsilon) 
= &\; \frac{\ln(\frac{16}{\epsilon})}{2 \pi^2} + \frac{\epsilon^2 (\ln(\frac{16}{\epsilon})-1)}{128 \pi^2}
+ \frac{3 \epsilon^4 (6 \ln(\frac{16}{\epsilon})-7)}{2^{16} \pi^2}
 +\frac{5}{3}\frac{\epsilon^6 (30 \ln(\frac{16}{\epsilon})-37)}{2^{22} \pi^2}
 	\\ \label{N0expansionfirstfew}
& + \frac{35}{3} \frac{\epsilon^8 (420 \ln(\frac{16}{\epsilon})-533)}{2^{33} \pi^2}
+ \frac{63}{5} \frac{\epsilon^{10} (1260\ln(\frac{16}{\epsilon})-1627)}{2^{39} \pi^2} + O\Big(\epsilon^{12} \ln{\frac{1}{\epsilon}}\Big).
\end{align}

\section{Proof of Lemma \ref{Delta1primelemma}}\label{Delta1primeapp}
Differentiating (\ref{AFequation2Delta1}) with respect to $b_+$ and solving the resulting equation for $\Delta_1'(b_+)$, we obtain
\begin{align}\label{Delta1primePQ}
\Delta_1'(b_+) = \frac{P(b_+)}{Q(b_+)},
\end{align}
where the functions $P(b_+)$ and $Q(b_+)$ are defined by
$$P(b_+) := - \frac{N_0(b_+) \Delta_1(b_+)}{ \sqrt{\Delta_1(b_+)^2 + b_+^2}}, \qquad
Q(b_+) := \Delta_1(b_+)^2\int_{b_+}^4  
 \frac{N_0(\epsilon)}{(\Delta_1(b_+)^2 + \epsilon^2)^{3/2}} d\epsilon.$$
Lemma \ref{Delta1lemma} and the expansion (\ref{N0expansionfirstfew}) of $N_0$ imply that
\begin{align}\nonumber
P(b_+)
&= 
-  \frac{\Big(\frac{\ln(\frac{16}{b_+})}{2 \pi^2} + O(b_+^2 \ln(\frac{16}{b_+}))\Big) \Big(8 \sqrt{16 - \hat{b}_+} + O(\sqrt{U})\Big) e^{-\frac{2\pi}{\sqrt{U}}}}{\Big(32 - \hat{b} + O(\sqrt{U})\Big)e^{-\frac{2\pi}{\sqrt{U}}}  }
	\\ \label{Pbplusexpansion}
& = 
-  \frac{8 \sqrt{16 - \hat{b}_+}}{(32 - \hat{b})\pi \sqrt{U}}(1 + O(\sqrt{U})) \qquad \text{as $U \downarrow 0$},
\end{align}
uniformly for $\hat{b}_+$ in compact subsets of $(0, 16)$.

To compute the expansion of $Q(b_+)$ as $U \downarrow 0$, we need the following lemma.

\begin{lemma}\label{intN0oneoverthreehalveslemma}
As $\Delta \downarrow 0$, it holds that
\begin{align}\nonumber
\int_{x \Delta}^4 \frac{N_0(\epsilon)}{(\Delta^2 + \epsilon^2)^{3/2}} d\epsilon
=&\; \frac{\ln(\Delta) \left(\frac{x}{\sqrt{x^2+1}}-1\right)}{2 \pi^2 \Delta ^2}+ O\bigg(\frac{1}{\Delta^2}\bigg)
\end{align}
uniformly for $x$ in compact subsets of $(0, +\infty)$.
\end{lemma}
\begin{proof}
We write the integral in the statement as
\begin{align}\label{intsplitA}
\int_{x \Delta}^4  \frac{N_0(\epsilon)}{(\Delta^2 + \epsilon^2)^{3/2}} d\epsilon
= A_1(\Delta, x) + A_2(\Delta, x)
+ E(\Delta, x),
\end{align}
where 
\begin{align*}
  A_1(\Delta, x) & := \int_{x \Delta}^4 \frac{\ln(\frac{16}{\epsilon})}{2 \pi^2} \frac{1}{(\Delta^2 + \epsilon^2)^{3/2}} d\epsilon,
\quad
  A_2(\Delta, x) := \int_{x \Delta}^4 \frac{\epsilon^2 (\ln(\frac{16}{\epsilon})-1)}{128 \pi^2} \frac{1}{(\Delta^2 + \epsilon^2)^{3/2}} d\epsilon,
\end{align*}
and
$$E(\Delta, x) := \int_{x \Delta}^4  \bigg(N_0(\epsilon) - \frac{\ln(\frac{16}{\epsilon})}{2 \pi^2} - \frac{\epsilon^2 (\ln(\frac{16}{\epsilon})-1)}{128 \pi^2}\bigg)\frac{1}{(\Delta^2 + \epsilon^2)^{3/2}} d\epsilon.$$
In view of (\ref{N0expansionfirstfew}), we have
\begin{align*}
|E(\Delta, x)| \leq C \int_0^4  \epsilon^4 \ln(\frac{16}{\epsilon})\frac{1}{\epsilon^3} d\epsilon \leq C
\end{align*}
for all $\Delta \geq 0$ and all $x \geq 0$ such that $x \Delta \leq 4$.
Also, letting $y = \epsilon/\Delta$, we can write
\begin{align*} 
  A_1(\Delta, x) & = \frac{1}{\Delta^2} \int_{x}^{4/\Delta} \frac{\ln(\frac{16}{y}) - \ln(\Delta)}{2 \pi^2} \frac{1}{(1 + y^2)^{3/2}} dy,
	\\
  A_2(\Delta, x) & = \int_{x}^{4/\Delta} \frac{y^2 (\ln(\frac{16}{y}) - \ln(\Delta) -1)}{128 \pi^2}\frac{1}{(1 + y^2)^{3/2}} dy.
\end{align*}
Using the primitive function
$$\frac{d}{dy} \frac{\frac{y (\ln(\frac{16}{y})-\ln{\Delta})}{\sqrt{y^2+1}}+\arcsinh(y)}{2 \pi ^2}
   = \frac{\ln(\frac{16}{y}) - \ln(\Delta)}{2 \pi^2} \frac{1}{(1 + y^2)^{3/2}},$$
we obtain, as $\Delta \downarrow 0$,
$$A_1(\Delta, x) = \frac{\ln(\Delta) \left(\frac{x}{\sqrt{x^2+1}}-1\right)}{2 \pi^2 \Delta ^2}+ O\bigg(\frac{1}{\Delta^2}\bigg)$$
uniformly for $x$ in compact subsets of $(0, +\infty)$.
Since $A_2(\Delta, x)$ satisfies
\begin{align*}
  A_2(\Delta, x) = O\bigg(\int_{x}^{4/\Delta} \frac{|\ln(\frac{16}{y})| + |\ln(\Delta)|}{128 \pi^2} \frac{1}{y} dy\bigg)
  = O(|\ln(\Delta)|^2)
\end{align*}
as $\Delta \downarrow 0$ uniformly for $x$ in compact subsets of $(0, +\infty)$, the desired conclusion follows.
\end{proof}

By Lemma \ref{Delta1lemma}, $x := b_+/\Delta_1(b_+)$ remains in a compact subset of  $(0,+\infty)$ as $U \downarrow 0$ if $\hat{b}_+$ lies in a compact subset of $(0, 16)$.
Hence Lemma \ref{intN0oneoverthreehalveslemma} in combination with Lemma \ref{Delta1lemma} yields
\begin{align}\label{Qbplusexpansion}
Q(b_+) & = \frac{\ln(\Delta_1(b_+)) \big(\frac{x}{\sqrt{x^2+1}}-1\big)}{2 \pi^2} + O(1)
 =
\frac{2 (16 - \hat{b}_+)}{\pi (32 - \hat{b}_+) \sqrt{U}}(1 + O(\sqrt{U})).
\end{align}
Substituting (\ref{Pbplusexpansion}) and (\ref{Qbplusexpansion}) into (\ref{Delta1primePQ}), we obtain the expansion of $\Delta_1'(b_+)$ stated in Lemma \ref{Delta1primelemma}. This completes the proof.

\medskip\noindent {\bf Acknowledgements.} 
We thank J. Henheik, A. B.  Lauritsen, and V. Melin for helpful discussions and collaborations on related projects. C.C. is a Research Associate of the Fonds de la Recherche Scientifique - FNRS. C.C. also acknowledges support from the European Research Council (ERC), Grant Agreement No. 101115687. E.L. acknowledges support from the Swedish Research Council, Grant No.\ 2023-04726. J.L. acknowledges support from the Swedish Research Council, Grant No.\ 2021-03877.

\bibliographystyle{plain}

\end{document}